\documentclass[USenglish,a4paper,11pt]{scrartcl}

\usepackage[utf8]{inputenc}
\usepackage{lmodern}

\usepackage{paralist}

\newcommand{\menge}[1]{\left\{#1\right\}}
\newcommand{\mengest}[2]{\left\{#1 \mid #2\right\}}
\newcommand{\nor}[1]{\left\|#1\right\|}

\newcommand{\cO}{\mathcal{O}}

\newcommand{\BP}{bin packing\xspace}

\usepackage{amsmath}
\usepackage{amsthm}
\usepackage{thmtools}
\usepackage{amsfonts}
\usepackage{amssymb}
\usepackage{graphicx}
\usepackage{tikz}
\usepackage{nicefrac}
\usepackage{ifthen}
\usepackage{caption}
\usepackage{subcaption}
\usepackage[left=1in,right=1in,top=1in,bottom=1.5in]{geometry}
\usepackage{pifont}
\newcommand{\cmark}{\ding{51}}%
\newcommand{\xmark}{\ding{55}}%

\usetikzlibrary{shapes,backgrounds,fit,patterns, decorations.pathmorphing}

\usepackage{csquotes}
\usepackage[final]{microtype}

\usetikzlibrary{decorations.pathreplacing, positioning, arrows, calc}
\usepackage[numbers]{natbib}

\newcommand{\bin}[3]{
\node (#1)[rectangle,minimum height=1.5cm,label=below:{#3},draw,xshift={#2}] {};
}

\newcommand{\drawgroups}[4]{
\edef\points{}
\foreach \x [count=\xi] in {#2}{
\ifthenelse{\equal{\x}{}}{
\node (#1\xi)[shape=rectangle,draw=white,minimum
    height=1cm, minimum width=1cm,xshift=2.1*\xi cm] {$\ldots$};
}
{
 \node (#1\xi)[draw,shape=rectangle,fill=black!20,minimum
    height=1cm, minimum width=2cm,xshift=2.1*\xi cm] {$\x$};
}
    \xdef\points{(#1\xi) \points}
    }   
 \node (all) [draw,ellipse,xscale=0.9,yscale=1.2, fit=\points]{}; 
 
 \foreach \x / \y in {#3}{
 \ifnumgreater{\x}{\y}{
 \draw (#1\x.north west) edge[->,>=open triangle 90,thick,bend right=40] node[above,fill=white,yshift=0.08cm] {#4} (#1\y.north east);
 }
 {
 \draw (#1\x.north east) edge[->,>=open triangle 90,thick,bend right=-40] node[above,fill=white,yshift=0.08cm] {#4} (#1\y.north west);
 }
}
 }

\newcommand{\drawbothgroups}[6]{
\drawgroups{i}{#1}{#2}{#6}
\begin{scope}[xshift=#5cm]
\drawgroups{j}{#3}{#4}{#6}
\end{scope}
}

\newtheorem{algo}{Algorithm}

\theoremstyle{plain}
\newtheorem*{theorem*}{Theorem}
\newtheorem*{lemma*}{Lemma}
\newtheorem{theorem}{Theorem}
\newtheorem{corollary}{Corollary}
\newtheorem{lemma}{Lemma}

\DeclareMathOperator{\online}{on}

\newcommand{\OPT}{\operatorname{\text{\textsc{opt}}}}
\newcommand{\SIZE}{\operatorname{\text{\textsc{size}}}}
\newcommand{\PACK}{\operatorname{\text{\textsc{pack}}}}
\newcommand{\INSERT}{\operatorname{\text{\textsc{insert}}}}

\newcommand{\RIGHT}{\operatorname{\text{\textsc{right}}}}
\newcommand{\LIN}{\operatorname{\text{\textsc{lin}}}}

\usepackage[lined,algonl,ruled,noend]{algorithm2e}

\usepackage[printonlyused]{acronym}

\acrodef{afptas}[AFPTAS]{asymptotic fully polynomial time approximation
  scheme}
  \acrodef{aptas}[APTAS]{asymptotic polynomial time approximation
  scheme}
\acrodef{ilp}[ILP]{integer linear program}
\acrodef{lp}[LP]{linear program}
\acrodef{ptas}[PTAS]{polynomial time approximation scheme}
\acrodef{amfptas}[AMFPTAS]{asymptotic (fully) polynomial time approximation scheme}
  
\usepackage{authblk}
\usepackage{nameref}
\author[1]{Sebastian Berndt}
\author[2]{Klaus Jansen}
\author[2]{Kim-Manuel Klein}
\affil[1]{Institute for Theoretical Computer Science, Universit\"at zu L\"ubeck\\berndt@tcs.uni-luebeck.de}
\affil[2]{Department of Computer Science, Christian-Albrechts-University to Kiel\\ \{kj,kmk\}@informatik.uni-kiel.de}

\renewcommand\Authands{ and }

\title{Fully Dynamic Bin Packing Revisited\footnote{Supported by DFG Project, Entwicklung und Analyse von effizienten polynomiellen Approximationsschemata f\"ur Scheduling- und verwandte Optimierungsprobleme, Ja 612/14-1.}}

\date{}

\begin{document}

\maketitle
\begin{abstract}
   We consider the \emph{fully dynamic bin packing} problem, where items
  arrive and depart in an online fashion and repacking of previously
  packed items is allowed. The goal is, of course, to minimize both the
  number of bins used as well as the amount of
  repacking. A recently introduced way of measuring the repacking costs
  at each timestep is the \emph{migration factor,} defined as the total
  size of repacked items divided by the size of an arriving or departing
  item. Concerning the trade-off between number of
  bins and migration factor, if we wish to achieve an asymptotic
  competitive ration of $1 + \epsilon$ for the number of bins, a
  relatively simple argument proves a lower bound of
  $\Omega(\nicefrac{1}{\epsilon})$ for the migration factor. We establish a nearly
  matching upper bound of $O(\nicefrac{1}{\epsilon}^4 \log
  \nicefrac{1}{\epsilon})$ using a new dynamic rounding technique and
  new ideas to handle small items in a dynamic setting such that no
  amortization is needed. The running time of our algorithm is
  polynomial in the number of items $n$ \emph{and} in
  $\nicefrac{1}{\epsilon}$. The previous best trade-off was for an
  asymptotic competitive ratio of $\nicefrac{5}{4}$ for the bins (rather
  than $1+\epsilon$) and needed an amortized number of $O(\log n)$
  repackings (while in our scheme the number of repackings is
  independent of~$n$ and non-amortized).
\end{abstract}

\section{Introduction}
For the classical \BP problem, we are given a set $I$ of items with a
size function $s\colon I\to (0,1]$ and need to pack them into as few unit
sized bins as possible. In practice, the
complete instance is often not known in advance, which has lead to the
definition of a variety of \emph{online} versions of the bin packing
problem. First, in the classical \emph{online bin packing}
\cite{ullman1971}, items arrive over time and
have to be packed on arrival. Second, in \emph{dynamic bin packing}
\cite{coffman1983}, items may also depart over time.
This dynamic bin packing model is often used for instance in
\begin{compactitem}
\item the placement and movement of 
virtual machines onto different servers for cloud computing \cite{beloglazov2010energy,
bobroff2007dynamic,srikantaiah2008energy,verma2008pmapper,jung2008generating,jung2009cost},
\item the development of guaranteed quality of service channels over certain multi-frequency
time division multiple access systems \cite{park2000efficient},
\item the placement of processes, which require different resources, onto physical host machines \cite{stolyar2013infinite, stolyar2013large},
\item the resource allocation in a cloud network where the cost depends upon different parameters \cite{daudjee2014fault,li2014dynamic}.
\end{compactitem}
Third and fourth, we may allow already packed items to be slightly
rearranged, leading to online bin packing with repacking (known as
\emph{relaxed online bin packing}) \cite{gambosi2000} and dynamic bin
packing with repacking (known as \emph{fully dynamic bin packing})
\cite{ivkovic1998}. See Figure \ref{fig:overview} for a short overview
on the different models.

\begin{figure}[h]
  \centering
\begin{tabular}[h]{l|c|c}
  Name & Deletion & Repacking\\
  \hline
  Online Bin Packing & \xmark & \xmark\\
  Relaxed Online Bin Packing & \xmark & \cmark\\
  Dynamic Bin Packing & \cmark & \xmark\\
  Fully Dynamic Bin Packing & \cmark & \cmark
\end{tabular}
  \caption{Overview of online models}
  \label{fig:overview}
\end{figure}

The amount of repacking can be measured in different ways. We can either
count the total number of moved items at each timestep or the sum of the
sizes of the moved items at each timestep. If one wants to count the
number of moved items, one typically counts a group of tiny items as a
single move. A \emph{shifting move} \cite{gambosi2000} thus involves
either a single large item or a bundle of small items in the same bin of
total size $s$ with $\nicefrac{1}{10}\leq s\leq \nicefrac{1}{5}$. Such a
bundle may consists of up to $\Omega(n)$ (very small) items. If an
algorithm measures the repacking by shifting moves, a new tiny item may
lead to a large amount of repacking. In order to guarantee that a tiny
item $i$ with size $s(i)$ only leads to a small amount of repacking, one
may allow to repack items whose size adds up to at most $\beta\cdot
s(i)$. The term $\beta$ is called the \emph{migration factor}
\cite{sanders2009}. Note that shifting moves and migration factor are
incomparable in the sense that a small migration factor does not imply a
small number of shifting moves and vice versa.

In order to measure the quality of an online algorithm, we compare the
costs incurred by an online algorithm with the costs incurred by an
optimal offline algorithm. An \emph{online algorithm} receives as input
a \emph{sequence} of items $I=(i_{1},i_{2},i_{3},\ldots)$ and decides at
each timestep $t$, where to place the item $i_t$ without knowing future
items $i_{t+1},i_{t+2},\ldots$. We denote by
$I(t)=(i_{1},i_{2},\ldots,i_{t})$ the instance containing the first $t$
items of the instance $I$ and by $\OPT(I(t))$ the minimal number of bins
needed to pack all items in $I(t)$. Note that the packings corresponding
to $\OPT(I(t))$ and $\OPT(I(t+1))$ may differ significantly, as those
packings do not need to be consistent. For an online algorithm $A$, we
denote by $A(I(t))$ the number of bins generated by the algorithm on the
input sequence $I(t)$. Note that $A$ must make its decision online,
while $\OPT(I(t))$ is the optimal value of the offline instance. The
quality of an algorithm for the online \BP problem is typically measured
by its \emph{asymptotic competitive ratio}. An online algorithm $A$ is
called an \emph{asymptotic $\alpha$-competitive algorithm}, if there is
a function $f\in o(\OPT)$ such that $A(I(t))\leq
\alpha\OPT(I(t))+f(I(t))$ for all instances $I$ and all $t\leq |I|$. The
minimum $\alpha$ such that $A$ is an asymptotic $\alpha$-competitive
algorithm is called the \emph{asymptotic competitive ratio of $A$},
denoted by $r_{\infty}^{\online}(A)$, i.\,e., the ratio is defined as
$r_{\infty}^{\online}(A)=\min\{\alpha \mid A$ is an asymptotic
$\alpha$-competitive algorithm$\}$. The online algorithm $A$ thus has a
double disadvantage: It does not know future items and we compare its
quality to the optimal offline algorithm which may produce arbitrary
different packings at time $t$ and time $t+1$. In order to remedy this
situation, one may also compare the solution generated by $A$ to a
non-repacking optimal offline algorithm. This non-repacking optimal
offline algorithm knows the complete instance, but is not allowed to
repack.

In this work, we present new results in fully dynamic bin packing where
we measure the quality of an algorithm against a repacking optimal
offline algorithm and achieve a asymptotic competitive ratio of $1 +
\epsilon$. The amount of repacking is bounded by
$\mathcal{O}(\nicefrac{1}{\epsilon}^4\log
(\nicefrac{1}{\epsilon}))$. While we measure the amount of repacking in
terms of the migration factor, we also prove that our algorithm uses at
most $\mathcal{O}(\nicefrac{1}{\epsilon}^4\log
(\nicefrac{1}{\epsilon}))$ shifting moves. Our algorithm runs in time
polynomial in the instance size and in $\nicefrac{1}{\epsilon}$.

\subsection{Previous Results on Online Variants of Bin Packing}

\subsubsection*{Online Bin Packing}
\label{sec:online-bp}
The classical version of online \BP problem was introduced by Ullman
\cite{ullman1971}.  In this classical model items arrive over time and
have to be packed at their arrival, while \emph{one is not allowed to
  repack already packed items}.  Ullman gave the very first online
algorithm \textsc{FirstFit} for the problem and proved that it its
absolute competitive ratio is at most $2$. The next algorithm
\textsc{NextFit} was given by Johnson \cite{johnson1974fast}, who proved
that its absolute competitive is also at most $2$. The analysis of the
\textsc{FirstFit} algorithm was refined by Johnson, Demers, Ullman,
Garey and Graham \cite{johnson1974worst}, who proved that its asymptotic
competitive ratio is at most $\nicefrac{17}{10}$. A revised version of
\textsc{FirstFit}, called \textsc{Revised FirstFit} was shown to have
asymptotic competitive ratio of at most $\nicefrac{5}{3}$ by Yao
\cite{yao1980new}.  A series of developments of so called \emph{harmonic
  algorithms} for this problem was started by Lee and Lee
\cite{lee1985simple} and the best known algorithm of this class which has
asymptotic competitive ratio at most $1{.}58889$ was given by Seiden
\cite{seiden2002}.
The lower bound on the absolute approximation ratio of $\nicefrac{3}{2}$
also holds for the asymptotic competitive ratio as shown by Yao
\cite{yao1980new}. This
lower bound was first improved independently by Brown
\cite{brown1979lower} and Liang \cite{liang1980lower} to $1{.}53635$ and
subsequently to $1{.}54014$ by van Vliet \cite{vliet1992} and finally to
$1{.}54037$ by Balogh, B{\'e}k{\'e}si and Galambos \cite{balogh2010}. 

\subsubsection*{Relaxed Online Bin Packing Model}
In contrast to the classical online \BP problem, Gambosi, Postiglione
and Talamo \cite{gambosi2000} considered the online case where one is
\emph{allowed to repack items}. They called this model the \emph{relaxed
  online \BP model} and proved that the lower bound on the competitive
ratio in the classical online \BP model can be beaten. They presented an
algorithm that uses $3$ \emph{shifting moves} and has an asymptotic
competitive ratio of at most $\nicefrac{3}{2}$, and an algorithm that
uses at most $7$ shifting moves and has an asymptotic competitive ratio
of $\nicefrac{4}{3}$. In another work, Ivkovi\'{c} and Lloyd
\cite{ivkovic1997} gave an algorithm that uses $\mathcal{O}(\log n)$
\emph{amortized} shifting moves and achieves an asymptotic competitive
ratio of $1+\epsilon$. In this amortized setting, shifting moves can be
saved up for later use and the algorithm may repack the whole instance
sometimes. Epstein and Levin \cite{epstein2006robust} used the measure
of the migration factor to give an algorithm that has an asymptotic
competitive ratio of $1+\epsilon$ and a migration factor of
$2^{\mathcal{O}((1/\epsilon)\log^{2}(1/\epsilon))}$. This result was
improved by Jansen and Klein \cite{jansen2013binpacking} who achieved
polynomial migration. Their algorithm uses a migration factor of
$\mathcal{O}(\nicefrac{1}{\epsilon}^{4})$ to achieve an asymptotic
competitive ratio of $1+\epsilon$.

Concerning lower bounds on the migration factor, Epstein and Levin
\cite{epstein2006robust} showed that no optimal solution can be
maintained while having a constant migration factor (independent of
$\nicefrac{1}{\epsilon}$). Furthermore, Balogh, B{\'e}k{\'e}si,
Galambos and Reinelt \cite{balogh2008lower} proved that a lower bound on the asymptotic competitive ratio of $1{.}3877$ holds, if the amount of repacking is measured by the number of items and one is only allowed to repack a \emph{constant number of items}.

\subsubsection*{Dynamic Bin Packing}
An extension to the classical online \BP model was given by Coffman,
Garey and Johnson \cite{coffman1983}, called the \emph{dynamic bin
  packing} model. In addition to the insertion of items, \emph{items
  also depart} over time. \emph{No repacking is allowed} in this
model. It is easily seen that no algorithm can achieve a constant
asymptotic competitive ratio in this setting. In order to measure the
performance of an online algorithm $A$ in this case, they compared the
\emph{maximum number of bins used by $A$} with the \emph{maximum number
  of bins used by an optimal offline algorithm}, i.\,e., an algorithm
$A$ in this dynamic model is called an \emph{asymptotic
  $\alpha$-competitive algorithm}, if there is a function $f\in
o(\text{max-\textsc{opt}})$, where $\text{max-\textsc{opt}}(I)=\max_{t}
\OPT(I(t))$ such that $\max_{t} A(I(t))\leq \alpha\cdot
\max_{t}\OPT(I(t))+f(I)$ for all instances $I$.  The minimum of all such
$\alpha$ is called the \emph{asymptotic competitive ratio of $A$}.
Coffman, Garey and Johnson modified the \textsc{FirstFit} algorithm and
proved that its asymptotic competitive ratio is at most
$2{.}897$. Furthermore, they showed a lower bound of $2{.}5$ on the
asymptotic competitive ratio when the performance of the algorithm is
compared to a repacking optimal offline algorith, i.\,e., $\max_{t}
\OPT(I(t))$.

In the case that the performance of the algorithm is compared to an
optimal non-repacking offline algorithm, Coffman, Garey and Johnson
showed a lower bound of $2{.}388$. This lower bound on the non-repacking
optimum was later improved by Chan, Lam and Wong \cite{chan2008dynamic}
to $2{.}428$ and even further in a later work by Chan, Wong and Yung
\cite{chan2009dynamic} to $2{.}5$.

\subsubsection*{Fully Dynamic Bin Packing}
We consider the dynamic \BP when repacking of already packed items is allowed. This model
was first investigated by Ivkovi\'{c} and Lloyd \cite{ivkovic1998} and is called \emph{fully dynamic \BP}. 
In this model, items arrive and depart in an online fashion and limited repacking is
allowed. The quality of an algorithm is measured by the asymptotic competitive ratio as defined in the classical online model (no maximum is taken as in the dynamic bin packing model). Ivkovi\'{c} and Lloyd developed an algorithm that uses amortized $\mathcal{O}(\log n)$  many shifting moves (see definition above) to achieve an asymptotic competitive ratio of $\nicefrac{5}{4}$. 

\subsubsection*{Related Results on the Migration Factor}
Since the introduction of the migration factor, several problems were
considered in this model and different robust algorithms for these
problems have been developed. Following the terminology of Sanders,
Sivadasan and Skutella \cite{sanders2009} we sometimes use the term
\emph{(online) approximation ratio} instead of competitive ratio. Hence,
we also use the term \ac{aptas} and \ac{afptas} in the context of online
algorithms.  If the migration factor of an algorithm $A$ only depends
upon the approximation ratio $\epsilon$ and not on the size of the
instance, we say that \emph{$A$ is an robust algorithm}.

In the case of online bin packing, Epstein and Levin
\cite{epstein2006robust} developed the first robust \ac{aptas} for the
problem using a migration factor of $2^{\mathcal{O}((1/\epsilon^{2})
  \log (1/\epsilon))}$. They also proved that there is no online
algorithm for this problem that has a constant migration factor and that
maintains an optimal solution. The \ac{aptas} by Epstein and Levin was
later improved by Jansen and Klein \cite{jansen2013binpacking}, who
developed a robust \ac{afptas} for the problem with migration factor
$\mathcal{O}(\nicefrac{1}{\epsilon^4})$. In their paper, they developed
new \ac{lp}/\ac{ilp} techniques, which we make use of to obtain
polynomial migration. It was shown by Epstein and Levin \cite{epsteinu}
that their \ac{aptas} for bin packing can be generalized to packing
$d$-dimensional cubes into unit cubes. Sanders, Sivadasan and Skutella
\cite{sanders2009} developed a robust \ac{ptas} for the scheduling
problem on identical machines with a migration factor of
$2^{\mathcal{O}((1/\epsilon) \log^2(1/\epsilon))}$. Skutella and
Verschae \cite{skutella2010} studied the problem of maximizing the
minimum load given $n$ jobs and $m$ identical machines. They also
considered a dynamic setting, where jobs may also depart. They showed
that there is no robust \ac{ptas} for this machine covering problem with
constant migration. The main reason for the nonexistence is due to very
small jobs. By using an amortized migration factor, they developed a
\ac{ptas} for the problem with amortized migration of
$2^{\mathcal{O}((1/\epsilon) \log^2(1/\epsilon))}$.

\subsection{Our Contributions}
\subsubsection*{Main Result}
In this work, we investigate the \emph{fully dynamic bin packing}
model. We measure the amount of repacking by the \emph{migration factor};
but our algorithm uses a bounded number of shifting moves as well.
Since the work of Ivkovi\'{c} and Lloyd from 1998 \cite{ivkovic1998}, no
progress was made on the fully dynamic \BP problem concerning the
asymptotic competitive ratio of $\nicefrac{5}{4}$. It was also unclear
whether the number of shifting moves (respectively migration factor)
must depend on the number of packed items $n$. In this paper we give
positive answers for both of these concerns. We develop an algorithm
that provides at each time step $t$ an approximation guarantee of
$(1+\epsilon)\OPT(I(t)) + \mathcal{O}(\nicefrac{1}{\epsilon} \log
(\nicefrac{1}{\epsilon}))$. The algorithm uses a migration factor of
$\mathcal{O}(\nicefrac{1}{\epsilon^4}\cdot
\log(\nicefrac{1}{\epsilon}))$ by repacking at most
$\mathcal{O}(\nicefrac{1}{\epsilon^3}\cdot
\log(\nicefrac{1}{\epsilon}))$ bins. Hence, the generated solution can
be arbitrarily close to the optimum solution, and for every fixed
$\epsilon$ the provided migration factor is constant (it does not depend
on the number of packed items). The running time is polynomial in $n$
and $\nicefrac{1}{\epsilon}$. In case that no deletions are used, the
algorithm has a migration factor of
$\mathcal{O}(\nicefrac{1}{\epsilon^3}\cdot
\log(\nicefrac{1}{\epsilon}))$, which beats the best known migration
factor of $\mathcal{O}(\nicefrac{1}{\epsilon^4})$ by Jansen and Klein
\cite{jansen2013binpacking}. Since the number of repacked bins is
bounded, so is the number of shifting moves as it requires at most
$O(\nicefrac{1}{\epsilon})$ shifting moves to repack a single
bin. Furthermore, we prove that there is no asymptotic approximation
scheme for the online \BP problem with a migration factor of
$o(\nicefrac{1}{\epsilon})$ even in the case that no items depart (and
even if $\mathcal{P}=\mathcal{NP}$). 

\subsubsection*{Technical Contributions}
We use the following techniques to achieve our results:
\begin{compactitem}
\item In order to obtain a lower bound on the migration factor in
  Section \ref{sec:bound}, we construct a series of instances that
  provably need a migration factor of $\Omega(\nicefrac{1}{\epsilon})$
  in order to have an asymptotic approximation ratio of $1+\epsilon$.
\item In Section \ref{sec:rounding}, we show how to handle large items
  in a fully dynamic setting. The fully dynamic setting involves more
  difficulties in the rounding procedure, in contrast to the setting
  where large items may not depart, treated in
  \cite{jansen2013binpacking}. A simple adaption of the dynamic
  techniques developed in \cite{jansen2013binpacking} does not work (see
  introduction of Section \ref{sec:rounding}). We modify the offline
  rounding technique by Karmarkar and Karp \cite{karmarkar1982} such
  that a feasible rounding structure can be maintained when items are
  inserted or removed. This way, we can make use of the
  \acs{lp}-techniques developed in Jansen and Klein
  \cite{jansen2013binpacking}.
\item In Section \ref{sec:small}, we explain how to deal with small
  items in a dynamic setting. In contrast to the setting where departure
  of items is not allowed, the fully dynamic setting provides major
  challenges in the treatment of small items. An approach is thus
  developed where small items of similar size are packed near each
  other. We describe how this structure can be maintained as new items
  arrive or depart. Note that the algorithm of Ivkovi\'{c} and Lloyd
  \cite{ivkovic1998} relies on the ability to manipulate up to
  $\Omega(n)$ very small items in constant time. See also their updated
  work for a thorough discussion of this issue \cite{ivkovic2009fully}.
\item In order to unify the different approaches for small and large
  items, in Section \ref{sec:general}, we develop an advanced structure for
  the packing. We give novel techniques and ideas to manage this
  mixed setting of small and large items. The advanced structure makes
  use of a potential function, which bounds the number of bins that need
  to be reserved for incoming items.
\end{compactitem}


\section{Lower Bound}
\label{sec:bound}
We start by showing that there is no robust (asymptotic) approximation scheme for \BP with migration factor of $o(\nicefrac{1}{\epsilon})$, even if $\mathcal{P}=\mathcal{NP}$. This improves the lower bound given by Epstein and Levin \cite{epstein2006robust}, which states that no algorithm for \BP, that maintains an optimal solution can have a constant migration factor. Previously it was not clear whether there exists a robust approximation algorithm for bin packing with sublinear migration factor or even a constant migration factor.

\begin{theorem}
\label{thm:bound}
For a fixed migration factor $\gamma > 0$, there is no robust approximation algorithm for \BP with asymptotic approximation ratio better than $1 + \frac{1}{6\lceil\gamma\rceil+5}$.
\end{theorem}

\begin{proof}
Let $\mathcal{A}$ be an approximation algorithm with migration factor $\gamma > 0$ and $c=\lceil \gamma\rceil$. We will now construct an instance such that the asymptotic approximation ratio of $\mathcal{A}$ with migration factor $c$ is at least $1+ \frac{1}{6c+5}$. The instance contains only two types of items: An $A$-item has size $a=\frac{\nicefrac{3}{2}}{3c+2}$ and an $B$-item has size $b=\nicefrac{1}{2}-\nicefrac{a}{3}$.
For a $M\in \mathbb{N}$, let 
\begin{align*}
I_M=[\underbrace{(b,\text{Insert}),(b,\text{Insert}),\ldots,(b,\text{Insert})}_{2M},\underbrace{(a,\text{Insert}),(a,\text{Insert}),\ldots,(a,\text{Insert})}_{2M(c+1)}]
\end{align*}
be the instance consisting of $2M$ insertions of $B$-items, followed by $2M(c+1)$ insertions of $A$-items. Denote by $r(t)$ the approximation ratio of the algorithm at time $t\in \mathbb{N}$. The approximation ratio of the algorithm is thus $r=\max_{t}\{r(t)\}$.

The insertion of the $B$-items produces a packing with $\beta_1$ bins containing a single $B$-item and $\beta_2$ bins containing two $B$-items. These are the only possible packings and hence $\beta_1+2 \beta_2=2M$. The optimal solution is reached if $\beta_1=0,\beta_2=M$. We thus have an approximation ratio of 
\begin{align*}
r(2M)=:r_1=\frac{\beta_1+\beta_2}{M}=\frac{2M-\beta_2}{M},
\end{align*}
which is strictly monotonically decreasing in $\beta_2$.  

The $A$-items, which are inserted afterwards, may either be put into bins which only contain $A$-items or into bins which contain only one $B$-item. The choice of $a,b$ implies $2\cdot b+a>1$ which shows that no $A$-item can be put into a bin containing two $B$-items. Denote by $\alpha$ the number of bins containing only $A$-items. The existing $B$-items may not be moved as the choice of $a,b$ implies $b>c\cdot a>\gamma\cdot a$. At most $\frac{\nicefrac{1}{2}+\nicefrac{a}{3}}{a}=c+1$ items of type $A$ may be put into the bins containing only one $B$-item. Note that this also implies that a bin which contains one $B$-item and $c+1$ items of type $A$ is filled completely. The optimal packing thus consists of $2M$ of those bins and the approximation ratio of the solution is given by 
\begin{align*}
r(2M(c+2))=:r_2=\frac{\beta_1+\beta_2+\alpha}{2M}=\frac{2M-2\beta_2+\beta_2+\alpha}{2M}=\frac{2M-\beta_2+\alpha}{2M}.
\end{align*}
There are at most $\beta_1\cdot (c+1)$ items of type $A$ which can be put into bins containing only one $B$-item. The remaining $(2M-\beta_1)(c+1)$ items of type $A$ therefore need to be put into bins containing only $A$-items. We can thus conclude $\alpha\geq (2M-\beta_1)(c+1) a=(2M-2M+2\beta_2)(c+1)a=2\beta_2(c+1)a$. As noted above, $\frac{\nicefrac{1}{2}+\nicefrac{a}{3}}{a}=c+1$ and thus $(c+1)a=\nicefrac{1}{2}+\nicefrac{a}{3}$.
Hence the approximation ratio is at least 
\begin{align*}
&r_2=\frac{\beta_1+\beta_2+\alpha}{2M}\geq \frac{2M-\beta_2+2\beta_2(\nicefrac{1}{2}+\nicefrac{a}{3})}{2M}=\\
&\frac{2M+\beta_2(-1+1+\nicefrac{2a}{3})}{2M}=\frac{2M+\beta_2\cdot \nicefrac{2a}{3}}{2M},
\end{align*}
which is strictly monotonically increasing in $\beta_2$. 

As $r\geq \max\{r_1,r_2\}$, a lower bound on the approximation ratio is thus given if $r_1=r_2$ by $\frac{2M-\beta}{M}=\frac{2M+\beta\cdot \nicefrac{2a}{3}}{2M}$ for a certain $\beta$. Solving this equation leads to $\beta=\frac{M}{\nicefrac{a}{3}+1}$. The lower bound is thus given as 
\begin{align*}
r\geq \frac{2M-\beta}{M}=2-\frac{1}{\nicefrac{a}{3}+1}=1+ \frac{1}{6c+5}
\end{align*}
by the choice of $a$. Note that this lower bound is independent from $M$. Hence, $r$ is also a lower bound on the asymptotic approximation ratio of any algorithm as the instance size grows with $M$. 
\end{proof}
We obtain the following corollary:
\begin{corollary}
There is no robust/dynamic (asymptotic) approximation scheme for \BP with a migration factor $\gamma \leq \nicefrac{1}{6}(\nicefrac{1}{\epsilon}-11) = \Theta(\nicefrac{1}{\epsilon})$.
\end{corollary}

\section{Dynamic Rounding}
\label{sec:rounding}
The goal of this section is to give a robust \ac{afptas} for the case that only large items arrive and depart. In the first subsection we present a general rounding structure. In the second subsection we give operations on how the rounding can be modified such that the general structure is preserved. We give the final algorithm in Section \ref{sec:dynamicbinpacking}, which is performed, when large items arrive or depart. Finally, the correctness is proved by using the \ac{lp}/\ac{ilp} techniques developed in \cite{jansen2013binpacking}. 

In \cite{jansen2013binpacking}, the last two authors developed a dynamic rounding technique based on an offline rounding technique from Fernandez de la Vega and Lueker \cite{de1981bin}. However, a simple adaption of these techniques does not work in the dynamic case where items may also depart. In the case of the offline rounding by Fernandez de la Vega and Lueker, items are sorted and then collected in groups of the same cardinality. As a new item arrives in an online fashion, this structure can be maintained by inserting the new item to its corresponding group. By shifting the largest item of each group to the left, the cardinality of each group (except for the first one) can be maintained. However, shifting items to the right whenever an item departs leads to difficulties in the \ac{lp}/\ac{ilp} techniques. As the rounding for a group may increase, patterns of the existing \ac{lp}/\ac{ilp} solution might become infeasible.  We overcome these difficulties by developing a new dynamic rounding structure and operations based on the offline rounding technique by Karmarkar and Karp \cite{karmarkar1982}. We felt that the dynamic rounding technique based on Karmarkar and Karp is easier to analyze since the structure can essentially be maintained by shifting items.

A \BP instance consists of a set of \emph{items} $I=\{i_{1},i_{2},\ldots,i_{n}\}$ with \emph{size function} $s:I\to [0,1]\cap \mathbb{Q}$. A feasible solution is a partition $B^{1},\ldots,B^{k}$ of $I$ such that $\sum_{i\in B^{j}}s(i)\leq 1$ for $j=1,\ldots,k$. We call a partition $B^{1},\ldots,B^{k}$ a \emph{packing} and a single set $B^{j}$ is called a \emph{bin}. The goal is to find a solution with a minimal number of bins. If the item $i$ is packed into the bin $B^{j}$, we write $B(i)=j$. The smallest value of $k\in \mathbb{N}$ such that a packing with $k$ bins exists is denoted by $\OPT(I,s)$ or if the size function is clear by $\OPT(I)$. A trivial lower bound is given by the value $\SIZE(I,s)=\sum_{i\in I}s(i)$. 

\subsection{Rounding}
\label{subsec:rounding}
To obtain an \ac{lp} formulation of fixed (independent of $|I|$) dimension, we use a rounding technique based on the offline \ac{afptas} by Karmarkar and Karp \cite{karmarkar1982}. In order to use the technique for our dynamic setting, we give a more general rounding. This generalized rounding has a certain structure that is maintained throughout the algorithm and guarantees an approximate solution for the original instance. First, we divide the set of items into \emph{small} ones and \emph{large} ones. An item $i$ is called \emph{small} if $s(i) <\nicefrac{\epsilon}{14}$, otherwise it is called \emph{large}. Instance $I$ is partitioned accordingly into a set of large items $I_{L}$ and a set of small items $I_{S}$. We treat small items and large items differently. Small items can be packed using an algorithm presented in Section \ref{sec:smallitems} while large items will be assigned using an \ac{ilp}. In this section we discuss how to handle large items. 

First, we characterize the set of large items more precisely by their sizes. We say that two large items $i,i'$ are in the same size category if there is a $\ell\in \mathbb{N}$ such that $s(i)\in (2^{-(\ell+1)},2^{-\ell}]$ and $s(i')\in (2^{-(\ell+1)},2^{-\ell}]$. Denote the set of all size categories by $W$. As every large item has size at least $\nicefrac{\epsilon}{14}$, the number of size categories is bounded by $\log(\nicefrac{1}{\epsilon})+5$. Next, items of the same size category are characterized by their \emph{block}, which is either $A$ or $B$ and their \emph{position} $r\in \mathbb{N}$ in this block. 
Therefore, we partition the set of large items into a set of groups $G \subseteq W \times \{A,B\}\times \mathbb{N}$. A group $g \in G$ consists of a triple $(\ell,X,r)$ with size category $\ell \in W$, block $X \in \{A,B \}$ and position $r \in \mathbb{N}$.
The \emph{rounding function} is defined as a function $R: I_L \mapsto G$ that maps each large item $i\in I_L$ to a \emph{group} $g\in G$. By $g[R]$ we denote the set of items being mapped to the group $g$, i.\,e.,\ $g[R]=\mengest{i\in I_L}{R(i)=g}$. 

 Let $q(\ell,X)$ be the maximal $r\in \mathbb{N}$ such that $|(\ell,X,r)[R]|> 0$. If $(\ell,X_1,r_1)$ and $(\ell,X_2,r_2)$ are two different groups, we say that $(\ell,X_1,r_1)$ is \emph{left} of $(\ell,X_2,r_1)$, if $X_1=A$ and $X_2=B$ or $X_1=X_2$ and $r_1<r_2$. We say that $(\ell,X_1,r_1)$ is \emph{right} of $(\ell,X_2,r_2)$ if it is not left of it.

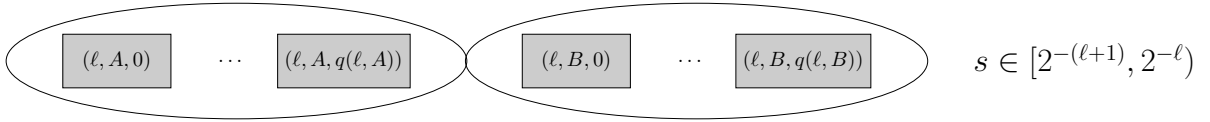
\begin{figure}[ht]
  \centering
  \scalebox{1.0}{
  \resizebox{\textwidth}{!}{
    \begin{tikzpicture}
    \drawbothgroups{(\ell,A,0),,{(\ell,A,q(\ell,A))}}{}{(\ell,B,0),,{(\ell,B,q(\ell,B))}}{}{8.5}{1}
    \node at (20,0) {{\LARGE $s\in [2^{-(\ell+1)},2^{-\ell})$}};    
            
  \end{tikzpicture}
  }}
  \caption{Grouping in $(\ell,A,\cdot)$ and $(\ell,B,\cdot)$}
\end{figure}


Given an instance $(I,s)$ and a rounding function $R$, we define the rounded size function $s^R$ by rounding the size of every large item $i \in g[R]$ up to the size of the largest item in its group, hence $s^R(i)=\max\mengest{s(i')}{R(i')=R(i)}$. We denote by $\OPT(I,s^R)$ the value of an optimal solution of the rounded instance $(I,s^R)$.

Depending on a parameter $k$, we define the following properties for a rounding function~$R$.
\begin{compactdesc}
\item[(a)\label{prop:a}] For each $i\in (\ell,X,r)[R]$ we have $2^{-(\ell+1)}<s(i)\leq 2^{-\ell}$. 
\item[(b)\label{prop:b}] For each $i\in (\ell,X,r)[R]$ and each $i'\in (\ell,X,r')[R]$ and $r<r'$, we have $s(i)\geq s(i')$.
\item[(c)\label{prop:c}] For each $\ell\in W$ and $1 \leq r\leq  q(\ell,A) $ we have $|(\ell,A,r)[R]| = 2^{\ell} k$ and $|(\ell,A,0)[R]|\leq 2^{\ell}k$.
\item[(d)\label{prop:d}] For each $\ell\in W$ and each $0 \leq r\leq q(\ell,B)-1 $ we have $|(\ell,B,r)[R]| = 2^{\ell} (k-1)$  and furthermore $|(\ell,B,q(\ell,B))[R]|\leq 2^{\ell}(k-1)$.
\end{compactdesc}
Property \nameref{prop:a} guarantees that the items are categorized correctly according to their sizes. 
Property \nameref{prop:b} guarantees that items of the same size category are sorted by their size and properties
\nameref{prop:c} and \nameref{prop:d} define the number of items in each group. 
\begin{lemma}\label{lem1}
For $k = \left\lfloor \frac{\SIZE(I_{L})\cdot
    \epsilon}{2(\lfloor \log(\nicefrac{1}{\epsilon})\rfloor
    +5)}\right\rfloor$ the number of non-empty groups in $G$ is
    bounded from above by $\mathcal{O}(\nicefrac{1}{\epsilon}\log(\nicefrac{1}{\epsilon}))$ assuming that $\SIZE(I_L) > \nicefrac{8}{\epsilon}\cdot (\lceil \log (\nicefrac{1}{\epsilon})\rceil+5)$.
\end{lemma}

 \begin{proof}
Using the definition of $k$ and the assumption, we show $\frac{2\SIZE(I_{L})}{k-1}\leq \nicefrac{8}{\epsilon}(\lfloor \log(\nicefrac{1}{\epsilon})\rfloor +5)$. We have
\begin{align*}
&\frac{2\SIZE(I_L)}{k-1} =
\frac{2\SIZE(I_L)}{\left\lfloor \frac{\SIZE(I_{L})\cdot
    \epsilon}{2(\lfloor \log(\nicefrac{1}{\epsilon})\rfloor +5)}\right\rfloor-1}\leq
    \frac{2\SIZE(I_L)}{\frac{\SIZE(I_L)\cdot \epsilon}{2(\lfloor \log( \nicefrac{1}{\epsilon})\rfloor +5)}-2}=\\
    &\frac{2\SIZE(I_L)}{\frac{\SIZE(I_L)\cdot \epsilon- 4(\lfloor\log (\nicefrac{1}{\epsilon})\rfloor +5)}{2(\lfloor\log( \nicefrac{1}{\epsilon})\rfloor +5)}}=
    \frac{2\SIZE(I_L)\cdot 2(\lfloor \log(\nicefrac{1}{\epsilon})\rfloor +5)}{\SIZE(I_L)\cdot \epsilon - 4(\lfloor\log(\nicefrac{1}{\epsilon})\rfloor +5)}
\end{align*}
As $\SIZE(I_L)> \nicefrac{8}{\epsilon}\cdot (\lceil \log (\nicefrac{1}{\epsilon})\rceil+5)$, we have $\nicefrac{\epsilon}{2}\SIZE(I_L) > 4(\lfloor \log(\nicefrac{1}{\epsilon})\rfloor +5)$. We can thus bound:

\begin{align*}
&\frac{2\SIZE(I_L)\cdot 2(\lfloor \log(\nicefrac{1}{\epsilon})\rfloor +5)}{\SIZE(I_L)\cdot \epsilon - 4(\lfloor\log(\nicefrac{1}{\epsilon})\rfloor +5)}\leq 
\frac{2\SIZE(I_L)\cdot 2(\lfloor \log(\nicefrac{1}{\epsilon})\rfloor +5)}{\SIZE(I_L)\cdot \epsilon - \nicefrac{\epsilon}{2}\SIZE(I_L) +1}=\\
&\frac{2\SIZE(I_L)\cdot 2(\lfloor \log(\nicefrac{1}{\epsilon})\rfloor +5)}{\SIZE(I_L)\cdot \nicefrac{\epsilon}{2}}=
\frac{4(\lfloor \log(\nicefrac{1}{\epsilon})\rfloor +5)}{\nicefrac{\epsilon}{2}}=
\frac{8(\lfloor \log(\nicefrac{1}{\epsilon})\rfloor +5)}{\epsilon}
\end{align*}

Note that property \nameref{prop:c} and property \nameref{prop:d} imply $|I(\ell)|\geq (q(\ell,A)+q(\ell,B)-2)2^{\ell}(k-1)$ . Hence property \nameref{prop:a} implies that $\SIZE(I(\ell),s)\geq |I(\ell)|2^{-\ell+1}\geq (q(\ell,A)+q(\ell,B)-2)(k-1)/2$ and therefore $q(\ell,A)+q(\ell,B)\leq 2\SIZE(I(\ell))/(k-1) +2$. We can now bound the total number of used groups by
 \begin{align*}
 &\sum_{\ell\in W}q(\ell,A)+q(\ell,B)\leq \sum_{\ell\in W} \left (\frac{2\SIZE(I(\ell))}{k-1}+2\right )\\
 &=2|W|+\frac{2}{k-1}\sum_{\ell\in W}\SIZE(I(\ell))= 2|W|+\frac{2}{k-1}\SIZE(I_{L})\\
 &\leq 2|W|+\frac{8}{\epsilon}(\lfloor \log(\nicefrac{1}{\epsilon})\rfloor +5) \leq\\
 & 2\cdot (\log(\nicefrac{1}{\epsilon})+5)+\frac{8}{\epsilon}(\log (\nicefrac{1}{\epsilon})+5)=\\
 &(\nicefrac{8}{\epsilon}+2)(\log (\nicefrac{1}{\epsilon})+5)\in \mathcal{O}(\nicefrac{1}{\epsilon}\log(\nicefrac{1}{\epsilon}))
 \end{align*}
 The total number of used groups is therefore bounded by $\mathcal{O}(\nicefrac{1}{\epsilon}\log(\nicefrac{1}{\epsilon}))$.

\end{proof}

The following lemma shows that the rounding function does in fact yield a $(1+\epsilon)$-approximation.

\begin{lemma}\label{lem2}

Given an instance $(I,s)$ with items greater than $\epsilon/14$ and a rounding function~$R$ fulfilling properties
  \nameref{prop:a} to \nameref{prop:d}, then $\OPT(I,s^R) \leq (1+\epsilon)\mathit{OPT}(I,s)$.

\end{lemma}

\begin{proof}
As $(I,s)$ only contains large items, $I_L=I$. Define for every $\ell$ the instances
$J_{\ell}=\bigcup_{r=2}^{q(\ell,A)}(\ell,A,r)[R]\cup \bigcup_{r=0}^{q(\ell,B)}(\ell,B,r)[R]$,
$J=\bigcup_{\ell\in W}J_{\ell}$ and
$K=\bigcup_{\ell\in W}(\ell,A,0)[R]\cup
(\ell,A,1)[R]$. We will now prove, that the error generated by this rounding
is bounded by $\epsilon$. As each solution to $J\cup K$ yields a
solution to $J$ and a solution to $K$, we get $\OPT(J\cup
K,s^R)\leq \OPT(J,s^{R})+\OPT(K,s^{R})$. For $i\in
(\ell,A,0)[R]\cup (\ell,A,1)[R]$, we have $s(i)\leq \max\mengest{s(i')}{i'\in
  (\ell,A,0)[R]}\leq 2^{-\ell}$ because of property \nameref{prop:a}. We can therefore pack at least $2^{\ell}$ items
from $(\ell,A,0)[R]\cup (\ell,A,1)[R]$ into a single bin. Hence, we get with property \nameref{prop:c}:
\begin{align*}
  &\OPT((\ell,A,0)[R]\cup (\ell,A,1)[R]),s^{R})\\
  &\leq (|(\ell,A,0)[R]|+|(\ell,A,1)[R]|)\cdot 2^{-\ell}\\
  &=2k
\end{align*}
We can therefore bound $\OPT(K,s^{R})$ as follows:
\begin{align*}
  \OPT(K,s^{R})&\leq
  \sum_{\ell\in W}\OPT((\ell,A,0)[R]\cup (\ell,A,1)[R]),s^{R})\\
  &\leq \sum_{\ell\in W}2k\\
  &\leq 2(\lfloor\log(\nicefrac{1}{\epsilon})\rfloor+5)k\\
  &=2\lfloor\frac{\SIZE(I)\epsilon}{2(\lfloor\log(\nicefrac{1}{\epsilon})\rfloor+5)}\rfloor\cdot
  (\lfloor\log(\nicefrac{1}{\epsilon})\rfloor+5)\\
  &\leq
  2\frac{\SIZE(I)\epsilon}{2(\lfloor\log(\nicefrac{1}{\epsilon})\rfloor+5)}\cdot (\lfloor\log(\nicefrac{1}{\epsilon})\rfloor+5)\\
  &=\epsilon \SIZE(I)\\
  &\leq \epsilon \OPT(I,s)\
\end{align*}
Using property \nameref{prop:b} for each item in $((\ell,X,r+1)[R]),s^R)$ we find a unique larger item in $(\ell,X,r)[R]$.
    Therefore we have for every item in the rounded instance $(J,s^R)$ an item with larger size in instance $(I,s)$ and hence
	\begin{align*}
	\OPT(J,s^{R})\leq \OPT(I,s).
   \end{align*}
   
The optimal value of the rounded solution can be bounded by 
\begin{align*}
\OPT(I,s^{R})\leq \OPT(J,s^{R})+\OPT(K,s^{R})\leq (1+\epsilon)\OPT(I,s).
\end{align*}
\end{proof}

We therefore have a rounding function, which generates only $\cO(\nicefrac{1}{\epsilon}\log(\nicefrac{1}{\epsilon}))$ different item sizes and the generated error is bounded by $\epsilon$.

\subsection{Rounding Operations}

\label{subsec:operations}

Let us consider the case where large items arrive and depart in an online fashion.
Formally this is described by a sequence of pairs $(i_{1},A_{1}),\ldots,(i_{n},A_{n})$ where $A_{i}\in \{\operatorname{Insert},\operatorname{Delete}\}$. 
At each time $t\in \{1,\ldots,n\}$ we need to pack the item $i_{t}$ into the corresponding packing of $i_{1},\ldots,i_{t-1}$ if $A_{i}=\operatorname{Insert}$ or remove the item $i_{t}$ from the corresponding packing of $i_{1},\ldots,i_{t-1}$ if $A_{i}=\operatorname{Delete}$.
We will denote the instance $i_{1},\ldots,i_{t}$ at time $t$ by $I(t)$ and the corresponding packing by $B_t$. We will also round our items and denote the rounding function at time $t$ by $R_t$. 
The large items of $I(t)$ are denoted by $I_{L}(t)$. At time $t$ we are allowed to repack  several items with a total size of $\beta\cdot s(i_{t})$ but we intend to keep the migration factor $\beta$ as small as possible. The term $\operatorname{repack}(t)=\sum_{i, B_{t-1}(i)\neq B_t(i)}s(i)$ denotes the sum of the items which are moved at time $t$, the \emph{migration factor} $\beta$ of an algorithm is then defined as $\max_t\menge{\nicefrac{\operatorname{repack}(t)}{s(i_t)}}$. As the value of $\SIZE$ will also change over the time, we define the value $\kappa(t)$ as
\begin{align*}
\kappa(t)=\frac{\SIZE(I_{L}(t))\cdot \epsilon}{2(\lfloor \log (\nicefrac{1}{\epsilon})\rfloor+5)}.
\end{align*}
As shown in Lemma \ref{lem1}, we will make use of the value $k(t):= \lfloor \kappa(t)\rfloor$.

We present operations that modify the current rounding $R_t$ and packing $B_t$ with its corresponding \ac{lp}/\ac{ilp} solutions to give a solution for the new instance $I(t+1)$. At every time $t$ the rounding $R_t$ maintains properties $\nameref{prop:a}$ to $\nameref{prop:d}$. Therefore
the rounding provides an asymptotic approximation ratio of $1+\epsilon$ (Lemma \ref{lem2}) while maintaining only $\mathcal{O}(\nicefrac{1}{\epsilon}\log(\nicefrac{1}{\epsilon}))$ many groups (Lemma \ref{lem1}). We will now present a way how to adapt this rounding to a dynamic setting, where items arrive or depart online.

Our rounding $R_{t}$ is manipulated by different \emph{operations}, called the \emph{insert, delete, shiftA} and \emph{shiftB} operation. Some ideas behind the operations are inspired by Epstein and Levin \cite{epstein2006robust}. The insert operation is performed whenever a large item arrives and the delete operation
is performed whenever a large item departs. The shiftA/shiftB operations are used to modify the number of groups that are contained in the $A$ and $B$ block. As we often need to filter the largest items of a group $g$ belonging to a rounding $R$, we denote this item by $\lambda(g,R)$.

\begin{itemize}

\item shift:
A shift operation takes two groups $(\ell,X_1,r_1)$ and
$(\ell,X_2,r_2)$, where $(\ell,X_1,r_1)$ is left of $(\ell,X_2,r_2)$,
and a rounding function $R$ and produces a new rounding function $R'$
and packing $B'$ by shifting the largest item from $(\ell,X_2,r_2)$ to
$(\ell,X_2,r_2-1)$ and so on until $(\ell,X_1,r_1)$ is reached.
\begin{itemize}
\item For all groups $g$ left of $(\ell,X_1,r_1)$ or right of $(\ell,X_2,r_2)$ set $g[R']=g[R]$. 
\item As we move an items out of $(\ell,X_2,r_2)$, set 
\begin{align*}
(\ell,X_2,r_2)[R']=(\ell,X_2,r_2)[R]\setminus \lambda((\ell,X_2,r_2),R).
\end{align*}
\item As we move an item into $(\ell,X_1,r_1)$, set
\begin{align*}
(\ell,X_1,r_1)[R']=(\ell,X_1,r_1)[R]\cup \lambda(\RIGHT(\ell,X_1,r_1),R).
\end{align*}
\end{itemize}

Whenever a shift-operation on $(\ell,X_{1},r_{1})$ and $(\ell,X_{2},r_{2})$ is performed, the \ac{lp} solution $x$ and the corresponding \ac{ilp} solution $y$ is updated to $x'$ and $y'$.
Let $C_i$ be a configuration containing $\lambda((\ell,X_2,r_2),R)$ with $x_i\geq 1$. Let $C_j = C_i \setminus s(\lambda((\ell,X_2,r_2),R))$ be the configuration without $\lambda((\ell,X_2,r_2),R)$. Set $x'_j = x_j +1$, $y'_j = y_j +1$ and $x'_i = x_i -1$, $y'_i = y_i -1$. 
In order to add the new item in $(\ell,X_1,r_1)$, set $x_h' = x_h +1$ and $y'_h = y_h +1$ for the index $h$ with $C_h = \{1:s(\lambda((\ell,X_1,r_1),R)) \}$. The remaining configurations do not change.

\begin{figure}[ht]
\centering
 \scalebox{0.9}{
   \resizebox{\textwidth}{!}{
 \begin{tikzpicture}
  \drawbothgroups{,(\ell,X_1,r_1),,{(\ell,A,q(\ell,A))}}{4/3,3/2}{(\ell,B,0),,(\ell,X_2,r_2),}{3/2,2/1}{11}{}
 \draw (j1.north west) edge[->,>=open triangle 90,thick,bend right=40] (i4.north east);  
 \end{tikzpicture}
 }}
 \caption{shift with parameters $(\ell,X_1,r_1)$ and $(\ell,X_2,r_2)$}
\end{figure}
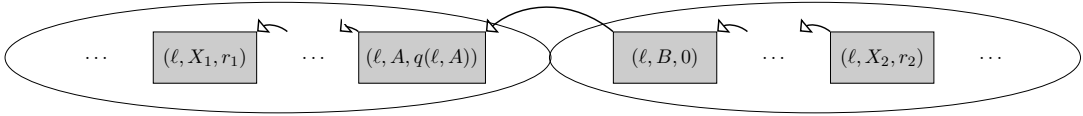

\item Insert: To insert item $i_t$, find the corresponding group $(\ell,X,r)$ with
\begin{compactitem}
\item $s(i_t)\in [\ell,2\ell)$,
\item $\min\mengest{s(i)}{i\in (\ell,X,r-1)}> s(i_t)$ and
\item $s(\lambda((\ell,X,r+1),R))\leq s(i_t)$.
\end{compactitem}
We will then insert $i_t$ into $(\ell,X,r)$ and get the rounding $R'$ by shifting the largest element of $(\ell,X,r)$ to $(\ell,X,r-1)$ and the largest item of $(\ell,X,r-1)$ to $(\ell,X,r-2)$ and so on until $(\ell,A,0)$ is reached. Formally, set $R^*(i_{t})=(\ell,X,r)$ and $R^*(i_j)=R(i_j)$ for $j\neq t$.
The rounding function $R'$ is then obtained by applying the shift operation on $R^*$ i.e. the new rounding is $R'=\operatorname{shift}((\ell,A,0),(\ell,X,r),R^*)$.

In order to pack the new item, let $i$ be the index with $C_i=\{1:s(\lambda((\ell,X,r),R'))\}$, as $i_t$ is rounded to the largest size in $(\ell,X,r)[R]$ after the shift. Place item $i_t$ into a new bin by setting $B'(i_t)=\max_j B(i_j)+1$ and $x_i' = x_i +1$ and $y_i'=y_i +1$. 

If $|(\ell, A ,0)[R']|=2^{\ell}\cdot k+1$, we have to create a new rounding group $(\ell, A ,-1)$. Additionally we shift the largest item in $(\ell,A,0)[R']$ to  the new group $(\ell,A,-1)[R']$. The final rounding $R''$ is then obtained by setting $(\ell,A,r)[R'']=(\ell,A,r-1)[R']$ i.e. incrementing the number of each rounding group by $1$. Note that the largest item in $(\ell,A,0)[R']$ is already packed into a bin of its own due to the shift operation. Hence, no change in the packing or the \ac{lp}/\ac{ilp} solution is needed. The insert operation thus yields a new packing $B'$ (or $B''$) which uses two more bins than the packing $B$.

\begin{figure}[ht]
\centering
  \scalebox{0.7}{
\resizebox{\textwidth}{!}{
  \begin{tikzpicture}
   \drawgroups{i}{(\ell,A,0),,(\ell,X,r),,{(\ell,X,q(\ell,X))}}{3/2,2/1}{}
  
  \node (i) [draw,rectangle,below = of i3,yshift=0.45cm] (i){$i$}; 
  \draw[->,thick] (i) to (i3);
  
  \end{tikzpicture}}}
  \caption{Insert $i$ into $(\ell,X,\cdot)$}
\end{figure}
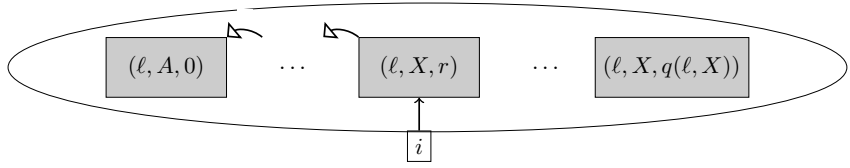

\item Delete: To delete item $i_t$ from the 
group $(\ell,X,r)$ with $R(i_t)=(\ell,X,r)$, we remove $i_t$ from this group and move the largest item from $(\ell,X,r+1)$ into $(\ell,X,r)$ and the largest item from $(\ell,X,r+2)$ into $(\ell,X,r+1)$ and so on until $(\ell,B,q(\ell,B))$. Formally the rounding $R'$ is described by the expression $\operatorname{shift}((\ell,X,r),(\ell,B,q(\ell,B)),R^*)$ where
\begin{align*}
g[R^*]=
\begin{cases}
(\ell,X,r)[R]\setminus \{i_t\} & g=(\ell,X,r)\\
g[R] & \text{else}
\end{cases}.
\end{align*}
As a single shift operation is used, the delete operation yields a new packing $B'$ which uses one more bin than the packing $B$.

For the \ac{lp}/ \ac{ilp} solution let $C_i$ be a configuration containing $\lambda((\ell,B,q(\ell,B)),R)$ with $x_i\geq 1$. Let $C_j = C_i \ s(\lambda((\ell,B,q(\ell,B)),R))$ be the configuration without the item $\lambda((\ell,B,q(\ell,B)),R)$. Set $x'_j = x_j +1$, $y'_j = y_j +1$ and
$x'_i = x_i -1$, $y'_i = y_i -1$. Set $B'(i_j) = B(i_j)$ for all  $j \neq t$ in order to remove the item $i_t$ from the packing.

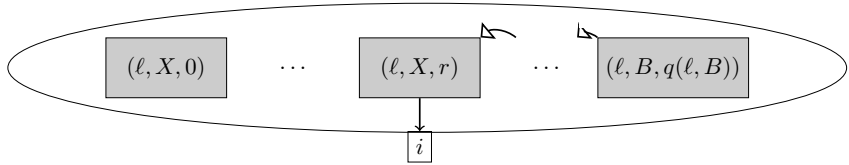
\begin{figure}[ht]
\centering
  \scalebox{0.7}{
  \resizebox{\textwidth}{!}{
  \begin{tikzpicture}
   \drawgroups{i}{(\ell,X,0),,(\ell,X,r),,{(\ell,B,q(\ell,B))}}{4/3,5/4}{}

  \node (i) [draw,rectangle,below = of i3,yshift=0.45cm] (i){$i$}; 
  \draw[->,thick] (i3) to (i);
  
  \end{tikzpicture}}
  }
  \caption{Delete $i$ from $(\ell,X,\cdot)$}
\end{figure}

\end{itemize}

To control the number of groups in $A$ and $B$ we introduce operations shiftA and shiftB that increase or decrease the number of groups in $A$ respectively $B$. An operation shiftA increases the number of groups in $A$ by $1$ and decreases the number of groups in $B$ by $1$. Operations shiftB is doing the inverse of shiftA.

\begin{itemize}

\item shiftA:
In order to move a group from $B$ to $A$ we will perform exactly $2^\ell$ times the operation $\operatorname{shift}((\ell,B,0),(\ell,B,q(\ell,B)),R)$ to receive the rounding $R^*$. Instead of opening a new bin for each of those $2^\ell$ items in every shift operation, we rather open one bin containing all items. Since every item in the corresponding size category has size $\leq 2^{-\ell}$, the items fit into a single bin.
The group $(\ell,B,0)$ has now the same size as the groups in $(\ell,A,\cdot)$. We transfer $(\ell,B,0)$ to block $A$. Hence we define for the final rounding $R'$ that $(\ell,A,r)[R']=(\ell,A,r)[R^*]$ for $r=0,\ldots,q(\ell,A)$ and $(\ell,A,q(\ell,A)+1)[R']=(\ell,B,0)[R^*]$ as well as $(\ell,B,r)[R']=(\ell,B,r+1)[R^*]$ for $r=0,\ldots,q(\ell,B)-1$. The resulting packing $B'$ hence uses one more bin than the packing $B$.

\begin{figure}[ht]
\centering
  \scalebox{0.7}{
  \resizebox{\textwidth}{!}{
  \begin{tikzpicture}
     \drawgroups{i}{(\ell,B,0),,(\ell,B,r),,{(\ell,B,q(\ell,B))}}{2/1,3/2,4/3,5/4}{\large{$2^\ell$}}

  \end{tikzpicture}}
  }
  \caption{shiftA}
\end{figure}
\item shiftB:
In order to move a group from $A$ to $B$ we will perform exactly $2^\ell$ times the operation $\operatorname{shift}((\ell,A,0),(\ell,A,q(\ell,A)),R)$ to receive the rounding $R^*$. As before in shiftA, we open a single bin containing all of the $2^\ell$ items. The group $(\ell,A,q(\ell,A))$ has now the same size as the groups in $(\ell,B,\cdot)$. We transfer $(\ell,A,q(\ell,A))$ to block $B$. Similar to shiftA we define for the final rounding $R'$ that $(\ell,A,r)[R']=(\ell,A,r)[R^*]$ for $r=0,\ldots,q(\ell,A)-1$ and $(\ell,B,0)[R']=(\ell,A,q(\ell,A))[R^*]$ as well as $(\ell,B,r+1)[R']=(\ell,B,r)[R^*]$. The resulting packing $B'$ hence uses one more bin than the packing $B$.

\end{itemize}

\begin{lemma}
\label{lem4}
Let $R$ be a rounding function fulfilling properties $\nameref{prop:a}$ to $\nameref{prop:d}$. Applying one of the operations insert, delete, shiftA or shiftB on $R$ results in a rounding function $R'$ fulfilling properties $\nameref{prop:a}$ to $\nameref{prop:d}$. 
\end{lemma}

\begin{proof}
Property $\nameref{prop:a}$ is always fulfilled as no item is moved between different size categories and the insert operation inserts an item into its appropriate size category.

As the order of items never changes and the insert operation inserts an item into the appropriate place, property $\nameref{prop:b}$ also holds.

For properties $\nameref{prop:c}$ and $\nameref{prop:d}$ we first note that the operation $\operatorname{shift}(g,g',R)$ increases the number of items in $g$ by $1$ and decreases the number of items in $g'$ by $1$.
The insert operation consists of adding a new item to a group $g$ followed by a $\operatorname{shift}((\ell,A,0),g,R)$ operation. Hence the number of items in every group except for $(\ell,A,0)$ (which is increased by $1$) remains the same. The delete operation consists of removing an item from a group $g$ followed by a $\operatorname{shift}(g,(\ell,B,q(\ell,B)),R)$ operation. Therefore the number of items in all groups except for $(\ell,B,q(\ell,B))$ (which is decreased by $1$) remains the same. As the number of items in $(\ell,A,0)$ and $(\ell,B,q(\ell,B))$ are treated seperately and may be smaller than $2^{\ell}\cdot k$ respectively $2^{\ell}\cdot (k-1)$, the properties $\nameref{prop:c}$ and $\nameref{prop:d}$ are always fulfilled for the insert and the delete operation. Concerning the shiftA operation we increase the number of items in a group $(\ell,B,0)$ by $2^\ell$. Therefore it now contains $2^{\ell}(k-1)+2^\ell= 2^{\ell}\cdot k$ items, which equals the number of items in groups of block $A$. As this group is now moved to block $A$, the properties $\nameref{prop:c}$ and $\nameref{prop:d}$ are fulfilled. Symmetrically the shiftB operation decreases the number of items in a group $(\ell,A,q(\ell,A))$ by $2^{\ell}$. Therefore the number of items in the group is now $2^{\ell}\cdot k - 2^\ell=2^{\ell}\cdot (k-1)$, which equals the number of items in the groups of block $B$. As this group is now moved to block $B$, the properties $\nameref{prop:c}$ and $\nameref{prop:d}$ are fulfilled. 
\end{proof}

According to Lemma \ref{lem1} the rounded instance $(I,s^R)$ has $\mathcal{O}(\nicefrac{1}{\epsilon}\log (\nicefrac{1}{\epsilon}))$ different item sizes (given a suitable $k$). Using the \ac{lp} formulation of Eisemann \cite{eisemann1957trim}, the resulting \ac{lp} called $LP(I,s^R)$ has $m = \mathcal{O}(\nicefrac{1}{\epsilon}\log(\nicefrac{1}{\epsilon}))$ constraints. 
We say a packing $B$ \emph{corresponds} to a rounding $R$ and an integral solution $y$ of the \ac{ilp} if all items in $(I,s^R)$ are
packed by $B$ according to $y$.

\begin{lemma}
\label{lem5}
Applying any of the operations insert, delete, shiftA or shiftB on a
rounding function $R$ and \ac{ilp} solution $y$ with corresponding packing $B$
defines a new rounding function $R'$ and a new integral solution $y'$. Solution $y'$ is a feasible solution of $LP(I,s^{R'})$.
\end{lemma}

\begin{proof}
    We have to analyze how the \ac{lp} for instance $(I,s^{R'})$ changes in comparison to the \ac{lp} for instance $(I,s^R)$.	\\
   {\bf Shift Operation:} A single $\operatorname{shift}(g_1,g_2,R)$ operation moves one item from each group $g$ between $g_1$ and $g_2$ into $g$ and one item out of $g$. As no item is moved out of $g_1$ and no item is moved into $g_2$, the number of items in $g_1$ is increased by $1$ and the number of items in $g_2$ is decreased by $1$. The right hand side of the $LP(I,s^R)$ is defined by the cardinalities $|g[R]|$ of the rounding groups $g$ in $R$. As only the cardinalities of $g_1$ and $g_2$ change by $\pm 1$ the right hand side changes accordingly to $\pm 1$ in the corresponding components of $y$. The moved item from $g_2$ is removed from the configuration and a new configuration containing the new item of $g_1$ is added. The \ac{lp} and \ac{ilp} solutions $x$ and $y$ are being modified such that $\lambda(g_2,R)$ is removed from its configuration and a new configuration is added such that the enhanced right hand side of $g_1$ is covered.
   Since the largest item $\lambda(g,R)$ of every group $g$ between $g_1$ and $g_2$ is shifted to its left group, the size $s^{R'}(i)$ of item $i\in g[R]$ is defined by $s^{R'}(i)=s(\iota(g,R))$, where $\iota(g,R)$ is the second largest item of $g[R]$. Therefore each item in $(I,s^{R'})$ is rounded to a smaller or equal value as $s(\iota(g,R))\leq s(\lambda(g,R))$. All configurations of $(I,s^R)$ can thus be transformed into feasible configurations of $(I,s^{R'})$.\\
    {\bf Insert Operation:} The insert operation consists of inserting the new item into its corresponding group $g$ followed by a shift operation. Inserting the new item into $g$ increases the right hand side of the \ac{lp} by $1$. To cover the increased right hand side, we add a new configuration $\{1:s^{R'}(i)\}$ containing only the new item. In order to reflect the change in the \ac{lp} solution, the new item is added into an additional bin. The remaining changes are due to the shift operation already treated above.\\
    {\bf Delete Operation:} The delete operation consists of removing an item $i$ from its corresponding group $g$ followed by a shift operation. Removing the new item from $g$ decreases the right hand side of the \ac{lp} by $1$. The current \ac{lp} and \ac{ilp} solutions $x$ and $y$ do not need to be changed to cover the new right hand side. The remaining changes are due to the shift operation already treated above.\\
    {\bf shiftA/shiftB Operation:} As the shiftA and shiftB operations consist only of repeated use of the shift operation, the correspondence between the packing and the \ac{lp}/\ac{ilp} solution follow simply by induction.
\end{proof}

\subsection{Algorithm for Dynamic Bin Packing}
\label{sec:dynamicbinpacking}
We will use the operations from the previous section to obtain a dynamic algorithm for \BP with respect to large items. The operations insert and delete are designed to process the input depending of whether an item is to be inserted or removed. Keep in mind that the parameter $k = \lfloor \kappa\rfloor =  \left\lfloor\frac{\SIZE(I_L)\cdot \epsilon}{2(\lfloor \log(\nicefrac{1}{\epsilon})\rfloor +5)} \right\rfloor$ changes over time as $\SIZE(I_L)$ may increase or decrease. In order to fulfill the properties $\nameref{prop:c}$ and $\nameref{prop:d}$, we need to adapt the number of items per group whenever $k$ changes. The shiftA and shiftB operations are thus designed to manage the dynamic number of items in the groups as $k$ changes.
Note that a group in the $A$-block with parameter $k$ has by definition the same number of items as a group in the $B$-block with parameter $k-1$ assuming they are in the same size category. If $k$ increases, the former $A$ block is treated as the new $B$ block in order to fulfill the properties $\nameref{prop:c}$ and $\nameref{prop:d}$ while a new empty $A$ block is introduced. To be able to rename the blocks, the $B$ block needs to be empty. Accordingly the $A$ block needs to be empty if $k$ decreases in order to treat the old $B$ block as new $A$ block. Hence we need to make sure that there are no groups in the $B$-block if $k$ increases and vice versa, that there are no groups in the $A$-block if $k$ decreases.

We denote the number of all groups in the $A$-blocks at time t by $A(t)$ and the number of groups in $B$-blocks at time $t$ by $B(t)$. To make sure that the $B$-block (respectively the $A$-block) is empty when $k$ increases (decreases) the ratio $\frac{A(t)}{A(t)+B(t)}$ needs to correlate to the fractional digits of $\kappa(t)$ at time $t$ denoted by $\Delta(t)$. Hence we partition the interval $[0,1)$ into exactly $A(t)+B(t)$ smaller intervals $J_i=\left[\frac{i}{A(t)+B(t)},\frac{i+1}{A(t)+B(t)}\right)$. We will make sure that $\Delta(t)\in J_i$ iff $\frac{A(t)}{A(t)+B(t)}\in J_i$. Note that the term $\frac{A(t)}{A(t)+B(t)}$ is $0$ if the $A$-block is empty and the term is $1$ if the $B$-block is empty. This way, we can make sure that as soon as $k(t)$ increases, the number of $B$-blocks is close to $0$ and as soon as $k(t)$ decreases, the number of $A$-blocks is close to $0$. Therefore, the $A,B$-block can be renamed whenever $k(t)$ changes. The algorithm uses shiftA and shiftB operations to adjust the number of $A$- and $B$-blocks. Recall that a shiftA operation reduces the number of groups in the $B$-block by $1$ and increases the number of groups in the $A$-block by $1$ (shiftB works vice versa). Let $d$ be the number of shiftA/shiftB operations that need to be performed to adjust $\frac{A(t)}{A(t)+B(t)}$. 


\begin{figure}[ht]
  \begin{subfigure}{.99\textwidth}
\centering
\scalebox{0.9}{
\begin{tikzpicture}
\draw (0,0) -- (10,0);
\draw[very thick] (0,-7pt) -- (0,7pt);
\node at (0,-15pt) {$k(t-1)$};
\draw[very thick] (5,-7pt) -- (5,7pt);
\node at (5,-15pt) {$k(t-1)+1$};

\draw[very thick] (10,-7pt) -- (10,7pt);
\node at (10,-15pt) {$k(t-1)+2$};

\foreach \x in {0,1,...,9}{
  \draw (\x+1,-2pt) -- (\x+1,2pt);
}
\node at (0.5,10pt) {$J_{0}$};
\node at (1.5,10pt) {$J_{1}$};
\node at (2.5,10pt) {$\ldots$};
\node at (3.5,10pt) {$J_{j}$};
\node at (4.5,10pt) {$\ldots$};

\node[rectangle,draw,minimum width=8cm, minimum height=1cm] (A) at
(4,-2.7) {$A(t-1)$};
\node[below = of A,yshift=1cm] {98\%};
\node[right = of A, rectangle,draw,minimum width=1cm, minimum
height=1cm,xshift=-0.9cm] (B) {$B(t-1)$};
\node[below = of B,yshift=1cm] {2\%};

\node at (3.2,-23pt) (dt) {$\Delta(t-1)$};
\draw[thick,->] (dt) to (3.8,0);
\end{tikzpicture}
}
\caption{Before Insert}
  \end{subfigure}

  \begin{subfigure}{.99\textwidth}
\centering
\scalebox{0.9}{
\begin{tikzpicture}[every text node part/.style={align=center}]
\draw (0,0) -- (10,0);
\draw[very thick] (0,-7pt) -- (0,7pt);
\node at (0,-25pt) {$k(t)-1$\\ $\|$ \\$k(t-1) $};
\draw[very thick] (5,-7pt) -- (5,7pt);
\node at (5,-25pt) {$k(t)$\\$\|$ \\$k(t-1)+1$};

\draw[very thick] (10,-7pt) -- (10,7pt);
\node at (10,-25pt) {$k(t)+1$\\$\|$\\$k(t-1)+2$};

\foreach \x in {0,1,...,9}{
  \draw (\x+1,-2pt) -- (\x+1,2pt);
}
\node at (5.5,10pt) {$J_{0}$};
\node at (6.5,23pt) {$J_j$\\ $\|$ \\ $J_{1}$};
\node at (7.5,10pt) {$\ldots$};
\node at (8.5,10pt) {$\ldots$};
\node at (9.5,10pt) {$\ldots$};

\node[rectangle,draw,minimum width=1cm, minimum height=1cm] (A) at
(1,-2.7) {$A(t)$};
\node[below = of A,yshift=1cm] {1\%};
\node[right = of A, rectangle,draw,minimum width=8cm, minimum
height=1cm,xshift=-0.9cm] (B) {$B(t)$};
\node[below = of B,yshift=1cm] {99\%};

\node at (6.0,-23pt) (dt) {$\Delta(t)$};
\draw[thick,->] (dt) to (6.7,0);
\end{tikzpicture}
}
    
\caption{After Insert}
    \end{subfigure}
  \caption{Comparison of the situation before and after an Insert Operation}
\end{figure}

In the following algorithm we make use of an algorithm called \textsc{improve}, which was developed in \cite{jansen2013binpacking} to reduce the number of used bins. Using \textsc{improve}(x) on a packing $B$ with approximation guarantee $\max_i B(i) \leq (1+\bar{\epsilon})\OPT + C$ for some $\bar{\epsilon} = \mathcal{O}(\epsilon)$ and some additive term $C$ yields a new packing $B'$ with approximation guarantee $\max_i B(i) \leq (1+\bar{\epsilon})\OPT + C-x$. We use the operations in combination with the improve algorithm to obtain a fixed approximation guarantee.

\begin{algo}[\ac{afptas} for large items] \label{alg-afptas}
  \ 
  \begin{small}

    \begin{algorithm}[H]
    \TitleOfAlgo{Insertion}
    \If{SIZE($I(t)) < (m+2)(\nicefrac{1}{\delta} +2)$ or 
	SIZE$(I(t)) < 8 (\nicefrac{1}{\delta} +1)$}{use offline Bin Packing}
    \Else{
    \textsc{improve}(2);
    insert($i$)\;
    \tcp{Shifting to the correct interval}
    Let $J_i$ be the interval containing $\Delta(t)$\;
    Let $J_j$ be the interval containing $\frac{A(t)}{A(t)+B(t)}$\;
    Set $d=i-j$\;
    \If(\tcp*[h]{Modulo $A(t)+B(t)$ when $k$ increases}){$k(t) > k(t-1)$}{
        $d$ = $d$ + $(A(t)+B(t))$\;
    }
    \tcp{Shifting $d$ groups from $B$ to $A$}
    \For{$p :=0$ to $|d|-1$}{
        \If{i+p = A(t) + B(t)}{Rename($A,B$);}
        \textsc{improve}(1);
        shiftA\;
    }
    }
    \end{algorithm}
  \end{small}
  \begin{small}

\begin{algorithm}[H]
    \TitleOfAlgo{Deletion}
    \If{SIZE($I(t)) < (m+2)(\nicefrac{1}{\delta} +2)$ or 
	SIZE$(I(t)) < 8 (\nicefrac{1}{\delta} +1)$}{use offline Bin Packing}
    \Else{
    \tcp{Departing item $i$}
    \textsc{improve}(4);
    delete($i$)\;
    \textsc{ReduceComponents}\;
    \tcp{}
    \tcp{Shifting to the correct interval}
    Let $J_i$ be the interval containing $\Delta(t)$\;
    Let $J_j$ be the interval containing $\frac{A(t)}{A(t)+B(t)}$\;
    Set $d=i-j$\;
    \If(\tcp*[h]{Modulo $A(t)+B(t)$ when $k$ decreases}){$k(t) < k(t-1)$}{
        d = d - (A(t)+B(t))\;
    }
    \tcp{Shifting $d$ groups from A to B}
    \For{$p :=0$ to $|d|-1$}{
        \If{i-p = 0}{Rename(A,B);}
        \textsc{improve}(3);
        shiftB\;
    }
    }
    \end{algorithm}
  \end{small}    
\end{algo}
Note that as exactly $d$ groups are shifted from $A$ to $B$ (or $B$ to $A$) we have by definition that $\Delta(t) \in  \left[ \frac{A(t)}{A(t)+B(t)}, \frac{A(t)+1}{A(t)+B(t)}\right)$ at the end of the algorithm. Note that $d$ can be bounded by $11$.

\begin{lemma}
\label{lem:disbounded}
At most $11$ groups are shifted from $A$ to $B$ (or $B$ to $A$) in Algorithm \ref{alg-afptas}.
\end{lemma}

\begin{proof}
Since the value $|\SIZE(I(t-1))-\SIZE(I(t))|$ changes at most by $1$ we can bound $|\kappa(t-1) - \kappa(t)|$ by $\frac{\epsilon}{2(\lfloor \log(1/\epsilon)\rfloor +5)}\leq \frac{\epsilon}{\log(\nicefrac{1}{\epsilon})+5}$ to obtain the change in the fractional part. By  Lemma \ref{lem1} the number of intervals (=the number of groups) is bounded by $(\frac{8}{\epsilon}+2)(\log (\nicefrac{1}{\epsilon})+5)$. Using $\Delta(t-1) \in [ \frac{A(t-1)}{A(t-1)+B(t-1)}, \frac{A(t-1)+1}{A(t-1)+B(t-1)})$ and the fact that the number of groups $A(t-1)+B(t-1)$ increases or decreases at most by $1$, we can give a bound for the parameter $d$ in both cases by
\begin{align*}
&d \leq \frac{D}{\text{interval length}} +1 = D \cdot \#intervals +1 \leq \\
&\left((\frac{\epsilon}{\log (\nicefrac{1}{\epsilon})+5})\cdot (\frac{8}{\epsilon}+2)\cdot (\log(\nicefrac{1}{\epsilon})+5)\right)+1=\\
&8+2\epsilon +1 < 11
\end{align*}
Hence, the number of shiftA/shiftB operations is bounded by $11$.
\end{proof}

\begin{lemma}
\label{lem:binpackingalg}
    Every rounding function $R_t$ produced by Algorithm \ref{alg-afptas} fulfills properties $\nameref{prop:a}$ to $\nameref{prop:d}$ with parameter $k(t)= \left\lfloor \frac{\SIZE(I_{L})\cdot \epsilon}{2(\lfloor \log(1/\epsilon)\rfloor +5)}\right\rfloor$.
\end{lemma}

\begin{proof}
Since Algorithm \ref{alg-afptas} uses only the operations insert, delete, shiftA and shiftB, the properties $\nameref{prop:a}$to$\nameref{prop:d}$ are always fulfilled by Lemma $\ref{lem4}$ and the \ac{lp}/\ac{ilp} solutions $x,y$ correspond to the rounding function by Lemma $\ref{lem5}$. 

Furthermore, the algorithm is designed such that whenever $k$ increases the $B$-block is empty and the $A$-block is renamed to be the new $B$-block. Whenever $k$ decreases the $A$-block is empty and the $B$-block is renamed to be the new $A$-block. Therefore the number of items in the groups is dynamically adapted to match with the parameter $k$.
\end{proof}

\subsection{Large items}
In this section we prove that Algorithm \ref{alg-afptas} is a dynamic robust \ac{afptas} for the \BP problem if all items have size at least $\nicefrac{\epsilon}{14}$. The treatment of small items is described in Section \ref{sec:small} and the general case is described in Section \ref{sec:general}. 

We will prove that the migration between packings $B_t$ and $B_{t+1}$ is bounded by $\mathcal{O}(\nicefrac{1}{\epsilon^3}\log(\nicefrac{1}{\epsilon}))$ and that we can guarantee an asymptotic approximation ratio such that $\max B_{t}(i) \leq (1+2\Delta) \OPT(I(t),s) + \text{poly}(\nicefrac{1}{\Delta})$ for a parameter $\Delta = \mathcal{O}(\epsilon)$ and for every $t \in \mathbb{N}$.
The Algorithm \textsc{improve} was developed in \cite{jansen2013binpacking} to improve the objective value of an \ac{lp} with integral solution $y$ and corresponding fractional solution $x$. For a vector $z \in \mathbb{R}^n$ let $V(z)$ be the set of all integral vectors $v = (v_1, \ldots v_n)^T$ such that $0 \leq  v_i \leq  z_i$.

Let $x$ be an approximate solution of the \ac{lp} $\min \mengest{\nor{x}_1}{Ax \geq b, x \geq 0 }$ with $m$ inequalities and let $\nor{x}_1 \leq (1+ \delta) \LIN$ and $\nor{x}_1 \geq 2 \alpha (1/ \delta +1)$, where $\LIN$ denotes the fractional optimum of the \ac{lp} and $\alpha\in \mathbb{N}$ is part of the input of the algorithm (see Jansen and Klein \cite{jansen2013binpacking}). 
Let $y$ be an approximate integer solution of the \ac{lp} with $\nor{y}_1 \leq \LIN +2C$ for some value 
$C \geq \delta \LIN$ and with $\nor{y}_1 \geq (m+2)(1/\delta +2)$. 
Suppose that both $x$ and $y$ have only $\leq C$ non-zero components. For every component $i$ we suppose that $y_i \geq x_i$. Furthermore we are given indices $a_1, \ldots ,a_K$, such that the non-zero components $y_{a_j}$ are sorted in non-decreasing order, i.\,e., $y_{a_1} \leq \ldots \leq y_{a_K}$.
\begin{algo}[\textsc{improve}]\label{improve}
\ 
  \begin{enumerate}
   \item Set $x^{var} := 2 \frac{ \alpha(1 / \delta +1)}{\nor{x}}x$, $x^{fix} := x - x^{var}$ and 
   $b^{var} = b - A(x^{fix})$
    \item Compute an approximate solution $\hat{x}$ of the \ac{lp} $\min \mengest{\nor{x}_1}{Ax \geq b^{var}, x\geq 0 }$
	with ratio $(1+ \delta/2)$
	\item If $\nor{x^{fix} + \hat{x}}_1 \geq \nor{x}_1$ then set $x' = x$, 
	$\hat{y} = y$ and goto step 9
  \item Choose the largest $\ell$ such that the sum of the smallest components $y_1, \ldots , y_{\ell}$ is bounded by
  $\sum_{1\leq i \leq \ell} y_{a_i} \leq (m+2)(1/ \delta +2)$
	\item For all $i $ set $\bar{x}^{fix}_{i} = 
	\begin{cases} 0 & \text{if }i= a_j, j \leq \ell \\
	x^{fix}_i & \text{else}
	\end{cases}$ 
	and $\bar{y}_i = \begin{cases} 0 & \text{if }i= a_j, j \leq \ell \\
	y_i & \text{else}
	\end{cases}$
	\item Set $\bar{x} = \hat{x} + x_{\ell}$ where $x_{\ell}$ is a
          vector consisting of the components 
	$x_{a_1}, \ldots ,x_{a_{\ell}}$. Reduce the number of non-zero components to at most $m+1$.
  \item $x' = \bar{x}^{fix} + \bar{x}$
  \item For all non-zero components $i$ set $\hat{y}_i = \max \{\lceil x'_i \rceil , \bar{y}_i \}$
	\item If possible choose $d \in V(\hat{y}-x')$ such that $\nor{d}_1 = \alpha (1/ \delta +1)$ otherwise
  choose $d \in V(\hat{y}-x')$ such that $\nor{d}_1 < \alpha (1/ \delta +1)$ is maximal.
  \item Return $y' = \hat{y} -d$
  \end{enumerate}
\end{algo}

In the following we prove that the algorithm \textsc{improve} applied to the \BP \ac{ilp} actually generates a new improved packing $B'$ from the packing $B$ with corresponding \ac{lp} and \ac{ilp} solutions $x'$ and $y'$. We therefore use Theorem \ref{thm-improve} and Corollary \ref{cor-improve} that were proven in \cite{jansen2013binpacking}.

\begin{theorem}\label{thm-improve}
    Let $x$ be a solution of the \ac{lp} with $\nor{x}_1 \leq (1+\delta)
    \LIN$ and furthermore $\nor{x}_1 \geq 
	2 \alpha (1/ \delta +1)$. Let $y$ be an integral
	 solution of the \ac{lp} with  $\nor{y'}_1 \leq \LIN +2C$ for some value $C \geq \delta \LIN$
	 and with $\nor{y}_1 \geq (m+2)(1/\delta +2)$.
	 Solutions $x$ and $y$ have the same number of non-zero components and for each component we have 
	 $x_i \leq y_i$.
	The Algorithm $\textsc{improve}(\alpha)$ then returns a fractional solution $x'$ with $\nor{x'}_1 \leq (1+ \delta)\LIN -\alpha$ and an integral solution
	$y''$ where one of the two properties hold:
	 $\nor{y'}_1 = \nor{y}_1 - \alpha$ or $\nor{y'}_1 = \nor{x'}_1 + C$. 
	 Both, $x'$ and $y'$ have at most $C$
	non-zero components and the distance between $y'$ and $y$ is bounded by $\nor{y'-y}_1 
	= \mathcal{O}(\frac{m + \alpha}{\delta})$.
\end{theorem}
\begin{corollary}\label{cor-improve}
    Let $\nor{x}_1 = (1+ \delta')\LIN$ for some $\delta' \geq \delta$ and $\nor{x}_1 \geq 2 \alpha (1/ \delta +1)$
	and let $\nor{y}_1 \leq \LIN + 2C$ for some $C \geq \delta'\LIN$ and $\nor{y}_1 \geq (m+2)(1/\delta +2)$. 
	Solutions $x$ and $y$ have the same number of non-zero components and for each component we have 
	 $x_i \leq y_i$.
	 Then Algorithm $\textsc{improve}(\alpha)$ returns a fractional solution
	$x'$ with $\nor{x'}_1 \leq \nor{x}_1 - \alpha = (1+ \delta')\LIN - \alpha$ and integral solution $y'$ where one of the two properties hold:
	 $\nor{y'}_1 = \nor{y}_1 - \alpha$ or $\nor{y'}_1 = \nor{x}_1 - \alpha + C$. 
	Both, $x'$ and $y'$ have at most $C$
	non-zero components and the distance between $y'$ and $y$ is bounded by $\nor{y'-y}_1 
	\in \mathcal{O}(\frac{m + \alpha}{\delta})$.
\end{corollary}
Let $\Delta = \epsilon + \delta + \epsilon \delta$ and $C = \Delta \OPT(I,s) + m$.

\begin{theorem}\label{thm-packing}
    Given a rounding function $R$ and an \ac{lp} defined for $(I,s^{R})$, let $x$ be a fractional solution of the \ac{lp} with
	$\nor{x}_1 \leq (1+ \Delta) \OPT(I,s)$, $\nor{x}_1 \geq 2\alpha(1/\delta +1)$ and $\nor{x}_1 = (1+\delta')\LIN(I,s^{R})$ for some $\delta'>0$. Let $y$ be an integral solution of the \ac{lp} with $\nor{y}_1 \geq (m+2)(1/\delta +2)$ and corresponding packing $B$ such that $\max_i B (i) = \nor{y}_1 \leq (1+ 2\Delta) \OPT(I,s)+m$.
	Suppose $x$ and $y$ have the same number $\leq C$ of non-zero components and for all components $i$ we have
	$y_i \geq x_i$.	Then Algorithm $\textsc{improve}(\alpha)$ on $x$ and $y$ returns a new fractional solution $x'$ with $\nor{x'}_1 \leq (1+ \Delta) 
	\OPT(I,s) - \alpha$ and also a new integral solution $y'$ with corresponding packing $B'$ such that
	\begin{align*}
    \max_i B' (i) =\nor{y'}_1 \leq (1+ 2 \Delta) \OPT(I,s) +m- \alpha.
    \end{align*}
	Further, both solutions $x'$ and $y'$ have the same number $\leq C$ of non-zero components and for each component we have
	$x'_i \leq y'_i$. The number of changed bins from the packing
        $B$ to the packing $B'$ is bounded by $\mathcal{O}(\frac{m}{\delta})$.
    \end{theorem}
\begin{proof}

    To use Theorem \ref{thm-improve} and Corollary \ref{cor-improve} we have to prove that certain conditions follow from the requisites of Theorem \ref{thm-packing}.
    We have $\max_i B (i) = \nor{y}_1 \leq (1+ 2\Delta) \OPT(I,s)+m$ by condition. Since
	$ \OPT(I,s) \leq  \OPT(I,s^{R})$ we obtain for the integral solution $y$ that
	$\nor{y}_1 \leq 2\Delta \OPT(I,s)+m + \OPT(I,s^{R}) \leq  2 \Delta \OPT(I,s)+ m + \LIN(I,s^{R}) +m$.
	Hence by definition of $C$ we get $\nor{y}_1 \leq  \LIN(I,s^{R}) + 2C$. This is one requirement to use Theorem \ref{thm-improve}
	or Corollary \ref{cor-improve}.
	We distinguish the cases where  $\delta' \leq \delta$ and $\delta' > \delta$ and look at them separately.
	
	Case 1: $\delta' \leq \delta$.
    For the parameter $C$ we give a lower bound by the inequality $C > \Delta \OPT(I,s) = (\delta + \epsilon + \delta \epsilon)\OPT(I,s)$. Lemma \ref{lem2} shows that $\OPT(I,s^R) \leq (1+\epsilon)\OPT(I,s)$ and therefore yields  
    \begin{align*}
    &\frac{\delta + \epsilon + \delta \epsilon}{1+ \epsilon} \OPT(I,s^R)
    = \frac{(1+\delta)(1+\epsilon)-1}{1+\epsilon} \OPT(I,s^R)\\
    &= (1+\delta)\OPT(I,s^R) - \frac{1}{1+\epsilon}\OPT(I,s^R)\\
        &\geq \delta \OPT(I,s^R) \geq \delta LIN(I,s^R)
    \end{align*}
    and hence $C > \delta \LIN(I,s^R)$. We can therefore use Theorem \ref{thm-improve}.
    
	Algorithm \textsc{improve} returns by Theorem \ref{thm-improve} a $x'$ with $\nor{x'}_1 \leq (1+\delta)\LIN(I,s^{R})-\alpha \leq (1+\delta)\OPT(I,s^{R})-\alpha$ and an integral solution $y'$ with	$\nor{y'}_1 \leq \nor{x'}_1 + C$ or $\nor{y'}_1 \leq \nor{y}_1 - \alpha$. Using that $\OPT(I,s^R) \leq (1+\epsilon)\OPT(I,s)$ we can conclude $\nor{x'}_1 \leq (1+ \delta)(1+ \epsilon)\OPT(I,s) - \alpha = (1+\Delta)\OPT(I,s) - \alpha$. In the case where $\nor{y'}_1 \leq \nor{x'}_1 + C$ we can bound the number of bins of the new packing $B'$ by	$\max_i B' (i) = \nor{y'}_1 \leq \nor{x'}_1 + C \leq (1 + 2 \Delta) \OPT(I,s)+ m - \alpha$.
	In the case that $\nor{y'}_1 \leq \nor{y}_1 - \alpha$ we obtain $\max_i B' (i) = \nor{y'}_1 \leq \nor{y}_1 - \alpha \leq (1+ 2\Delta) \OPT(I,s) +m- \alpha$. Furthermore we know by Theorem \ref{thm-improve} that $x'$ and $y'$ have at most $C$ non-zero components.
	
	Case 2: $\delta' > \delta$.
    First we prove that $C$ is bounded from below. Since $\nor{x}_1 = (1+\delta') \LIN(I,s^R) \leq (1+ \Delta) \OPT(I,s)\leq (1+ \Delta) \OPT(I,s^R) \leq (1+ \Delta) \OPT(I,s^R) \leq (1+ \Delta) (\LIN(I,s^R) + \frac{m}{2}) \leq \LIN(I,s^R) +C$ we obtain that $C\geq \delta' \LIN(I,s^R)$, which is a requirement to use Corollary \ref{cor-improve}.
	By using Algorithm \textsc{improve} on solutions $x$ with $\nor{x}_1 = (1+\delta')\LIN(I,s^{R})$ and $y$ with $\nor{y}_1 \leq  \LIN(I,s^{R}) + 2C$ we obtain by Corollary \ref{cor-improve} a fractional solution $x'$ with
	$\nor{x'}_1 \leq \nor{x}_1 - \alpha \leq (1+\Delta)\OPT(I,s) - \alpha$ and an integral solution $y'$ with either $\nor{y'}_1 \leq \nor{y}_1 - \alpha$ or $\nor{y'}_1 \leq \nor{x}_1 + C - \alpha$.
	So for the new packing $B'$ we can guarantee that $\max_i B' (i) = \nor{y'}_1 \leq \nor{y}_1 - \alpha =  \max_i B (i) - \alpha	\leq (1+ 2\Delta) \OPT(I,s) +m - \alpha$ if $\nor{y'}_1 \leq \nor{y}_1 - \alpha$. In the case that $\nor{y'}_1 \leq \nor{x}_1 + C - \alpha$, we can guarantee that $\max_i B' (i) = \nor{y'}_1 \leq \nor{x}_1 + C - \alpha \leq (1+ \Delta) \OPT(I,s)+ C - \alpha \leq (1+ 2\Delta) \OPT(I,s) +m - \alpha$. Furthermore we know by Corollary \ref{thm-improve} that $x'$ and $y'$ have at most $C$ non-zero components.
    
Theorem \ref{thm-improve} as well as Corollary \ref{cor-improve} state that the distance $\nor{y'-y}_1$ is bounded by $\mathcal{O}(\nicefrac{m}{\delta})$. Since $y$ corresponds directly to the packing $B$ and the new integral solution $y'$ corresponds to the new packing $B'$, we know that only $\mathcal{O}(\nicefrac{m}{\delta})$ bins of $B$ need to be changed to obtain packing $B'$.
\end{proof}

In order to prove correctness of Algorithm \ref{alg-afptas}, we will make use of the auxiliary Algorithm \ref{reducecomponents} (\textsc{ReduceComponents}). Due to a delete-operation, the value of the optimal solution $\OPT(I,s)$ might decrease. Since the number of non-zero components has to be bounded by $C = \Delta \OPT(I,s) + m$, the number of non-zero components might have to be adjusted down. The following algorithm describes how a fractional solution $x'$ and an integral solution $y'$ with reduced number of non-zero components can be computed such that $\nor{y-y'}_1$ is bounded. The idea behind the algorithm is also used in the \textsc{Improve} algorithm. The smallest $m+2$ components are reduced to $m+1$ components using a standard technique presented for example in \cite{beling1998}. Arbitrary many components of $x'$ can thus be reduced to $m+1$ components without making the approximation guarantee worse.
\begin{algo}[\textsc{ReduceComponents}]\label{reducecomponents}
\ 
  \begin{enumerate}
  \item Choose the smallest non-zero components $y_{a_1}, \ldots , y_{a_{m+2}}$.
  \item If $\sum_{1\leq i \leq m+2} y_{a_i} \geq (1/ \Delta +2)(m+2)$ then return $x=x'$ and $y=y'$
  
  \item Reduce the components $x_{a_1}, \ldots , x_{a_{m+2}}$ to $m+1$ components $\hat{x}_{b_1}, \ldots , \hat{x}_{b_{m+1}}$ with $\sum_{j=1}^{m+2}x_{a_{j}}=\sum_{j=1}^{m+1}\hat{x}_{b_{j}}$.
  \item For all $i $ set $x'_i = 
	\begin{cases} \hat{x}_i +x_i & \text{if $i= b_j$ for some } j \leq m \\
	0 & \text{if $i= a_j$ for some } j \leq m+1 \\
	x_i & \text{else}
	\end{cases}$

    and $\hat{y}_i = 	\begin{cases} \lceil \hat{x}_i + x'_i \rceil & \text{if $i= b_j$ for some } j \leq m \\
 	0 & \text{if $i= a_j$ for some } j \leq m+1 \\
	y_i & \text{else}
	\end{cases}$
  \item If possible choose $d \in V(\hat{y}-x')$ such that $\nor{d}_1 = m+1$ otherwise
  choose $d \in V(\hat{y}-x')$ such that $\nor{d}_1 < m+1$ is maximal.
  \item Return $y' = \hat{y} -d$
  \end{enumerate}
\end{algo}
The following theorem shows that the algorithm above yields a new fractional solution $x'$ and a new integral solution $y'$ with a reduced number of non-zero components.
\begin{theorem}
\label{thm:reduce}
    Let $x$ be a fractional solution of the \ac{lp} with	$\nor{x}_1 \leq (1+ \Delta) \OPT(I,s)$. Let $y$ be an integral solution of the \ac{lp} with $\nor{y}_1 \leq (1+ 2\Delta) \OPT(I,s)+m$. Suppose $x$ and $y$ have the same number $\leq C+1$ of non-zero components and for all components $i$ we have
	$y_i \geq x_i$.	Using the Algorithm $\textsc{ReduceComponents}$ on $x$ and $y$ returns a new fractional solution $x'$ with $\nor{x'}_1 \leq (1+ \Delta) \OPT(I,s)$ and a new integral solution $y'$ with $\nor{y'}_1 \leq (1+ 2 \Delta) \OPT(I,s) +m$. Further, both solutions $x'$ and $y'$ have the same number of non-zero components and for each component we have $x'_i \leq y'_i$. The number of non-zero components can now be bounded by $\leq C$. Furthermore, we have that $\nor{y-y'}_1 \leq 2\cdot (1/\Delta +3) (m+2)$.
\end{theorem}
\begin{proof}
    Case 1: $\sum_{1\leq i \leq m+2} y_{a_i} \geq (1/ \Delta +2)(m+2)$. We will show that in this case, $x$ and $y$ already have $\leq C$ non-zero components. In this case the algorithm returns $x' = x$ and $y' =y$. Since $\sum_{1\leq i \leq m+2} y_{a_i} \geq (1/ \Delta +2)(m+2)$ the components $y_{a_1}, \ldots , y_{a_{m+2}}$ have an average size of at least $(1/ \Delta +2)$ and since $y_{a_1}, \ldots , y_{a_{m+2}}$ are the smallest components, all components of $y$ have average size at least $(1/ \Delta +2)$. The size $\nor{y}_1$ is bounded by $(1+ 2\Delta) \OPT(I,s)+m$. Hence the number of non-zero components can be bounded by $\frac{(1+ 2\Delta) \OPT(I,s)+m}{\nicefrac{1}{\Delta}+2} \leq \Delta \OPT(I,s)+ \Delta m \leq C$.

Case 2: $\sum_{1\leq i \leq m+1} y_{a_i} < (1/ \Delta +2)(m+2)$. We have to prove different properties for the new fractional solution $x'$ and the new integral solution $y'$.

\textbf{Number of non-zero components}: The only change in the number of non-zero components is in step 3 of the algorithm, where the number of non-zero components is reduced by $1$. As $x,y$ have at most $C+1$ non-zero components, $x',y'$ have at most $C$ non-zero components. In step 4 of the algorithm, $\hat{y}$ is defined such that $\hat{y}_i \geq x'_i$. In step 5 of the algorithm $d$ is chosen such that $\hat{y}_i -d \geq x'_i$. Hence we obtain that $y'_i = \hat{y}_i -d \geq x'_i$.

\textbf{Distance between $y$ and $y'$}: The only steps where components of $y$ changes are in step 4 and 5. The distance between $y$ and $\hat{y}$ is bounded by the sum of the components that are set to $0$, i.\,e., $\sum_{j=1}^{m+2}y_{a_{j}}$ and the sum of the increase of the increased components $\sum_{j=1}^{m+1}\lceil \hat{x}_{b_{j}}\rceil \leq \sum_{j=1}^{m+1}\hat{x}_{b_{j}} +m+1 = \sum_{j=1}^{m+2}x_{a_{j}} +m+1$. As $\sum_{j=1}^{m+2}x_{a_{j}}\leq \sum_{j=1}^{m+2}y_{a_{j}} < (1/ \Delta +2)(m+2)$, we obtain that the distance between $y$ and $\hat{y}$ is bounded by $2\cdot (1/\Delta +2)(m+2)+m+1$. Using that $\nor{d}_1 \leq m+1$, the distance between $y$ and $y'$ is bounded by $\nor{y'-y}_1 < 2\cdot (1/\Delta +3) (m+2)$.

\textbf{Approximation guarantee}: The fractional solution $x$ is modified by condition of step 3 such that the sum of the components does not change. Hence $\nor{x'}_1=\nor{x}_1\leq (1+\Delta)\OPT(I,s)$.\\
Case 2a: $\nor{d}_1 < m+1$. Since $d$ is chosen maximally we have for every non-zero component that $y'_i-x'_i <1$. Since there are at most $C=\Delta\OPT(I,s)+m$ non-zero components we obtain that $\nor{y'}_1 \leq \nor{x'}_1 + C \leq (1+ 2\Delta)\OPT(I,s)+m$.
Case 2b: $\nor{d}_1 = m+1$.
By definition of $\hat{y}$ we have $\nor{\hat{y}}_1 \leq \nor{y}_1 + \sum_{j=1}^{m+1} \lceil \hat{x_{b_j}}+x_{b_j}\rceil  - \sum_{j=1}^{m+2}x_{a_j}  \leq \nor{y}_1 + m+1$. We obtain for $y'$ that $\nor{y'}_1 = \nor{\hat{y}}_1 - \nor{d}_1 \leq \nor{y}_1 + m+1 - (m+1) =\nor{y}_1 \leq (1+ 2\Delta) \OPT(I,s)+m$.
\end{proof}

\begin{theorem}
\label{thm-main}
Algorithm \ref{alg-afptas} is an \ac{afptas} with migration factor at most $\mathcal{O}(\frac{1}{\epsilon^3}\cdot \log(\nicefrac{1}{\epsilon}))$ for the fully dynamic \BP problem with respect to large items.
\end{theorem}

\begin{proof}
Set $\delta = \epsilon$. Then $\Delta = 2 \epsilon + \epsilon^2 = \mathcal{O}(\epsilon)$. We assume in the following that $\Delta \leq 1$ (which holds for $\epsilon\leq \sqrt{2}-1$).

We prove by induction that four properties hold for any packing $B_t$ and the corresponding \ac{lp} solutions. Let $x$ be a fractional solution of the \ac{lp} defined by the instance $(I_t,s^{R_{t}})$ and $y$ be an integral solution of this \ac{lp}.
    The properties $(2)$ to $(4)$ are necessary to apply Theorem \ref{thm-packing} and property $(1)$ provides the wished approximation ratio for the \BP problem.
	\begin{enumerate}
		\item[(1)\label{prop:1}] $\max_i B_t(i) = \nor{y}_1 \leq (1+ 2\Delta)\OPT(I(t),s) +m$ (the number of bins is bounded)
		\item[(2)\label{prop:2}] $\nor{x}_1  \leq (1+ \Delta) \OPT(I(t),s)$
		\item[(3)\label{prop:3}] for every configuration $i$ we have $x_i \leq y_i$
		\item[(4)\label{prop:4}] $x$ and $y$ have the same number of non-zero components and that number is bounded by $\Delta \OPT(I(t),s) +m$
	\end{enumerate}
	To apply Theorem \ref{thm-packing} we furthermore need a guaranteed minimal size for $\nor{x}_1$ and $\nor{y}_1$.
	According to Theorem \ref{thm-packing} the integral solution $y$ needs $\nor{y}_1 \geq (m+2)(\nicefrac{1}{\delta} +2)$ and 
	$\nor{x}_1 \geq 8 (\nicefrac{1}{\delta} +1)$ as we set $\alpha \leq 4$.
	By condition of the while-loop the call of \textsc{improve} is made iff $SIZE(I_t,s) \geq 8 (\nicefrac{1}{\delta} +1)$ and $SIZE(I_t,s) \geq (m+2)(\nicefrac{1}{\delta} +2)$. Since $\nor{y}_1 \geq \nor{x}_1 \geq SIZE(I_t,s)$ the requirements for the minimum size are fulfilled. As long as the instance is smaller than $8 (\nicefrac{1}{\delta} +1)$ or $(m+2)(\nicefrac{1}{\delta} +2)$ an offline algorithm for \BP is used. Note that there is an offline algorithm which fulfills properties $(1)$ to $(4)$ as shown by Jansen and Klein \cite{jansen2013binpacking}.
	
    Now let $B_t$ be a packing with $SIZE(I_t,s) \geq 8 (\nicefrac{1}{\delta} +1)$ and $SIZE(I_t,s) \geq (m+2)(\nicefrac{1}{\delta} +2)$ for instance $I_t$ with solutions $x$ and $y$ of the \ac{lp} defined by $(I(t),s^{R_t})$. Suppose by induction that the properties $(1)$ to $(4)$ hold for the instance $I_t$. We have to prove that these properties also hold for the instance $I(t+1)$ and the corresponding solutions $x''$ and $y''$. The packing $B_{t+1}$ is created by the repeated use of an
	call of \textsc{improve} for $x$ and $y$ followed by an operation (insert, delete, shiftA or shiftB).
    We will prove that the properties $(1)$ to $(4)$ hold after a call of \textsc{improve} followed by an operation.
	\\{\bf improve:} Let $x'$ be the resulting fractional solution of Theorem \ref{thm-packing}, let $y'$ be the resulting integral solution of Theorem \ref{thm-packing} and let $B'_t$ be the corresponding packing. Properties $(1)$ to $(4)$ are fulfilled 
	for $x$, $y$ and $B_t$ by induction hypothesis. Hence all conditions are fulfilled to use Theorem \ref{thm-packing}. 
	By Theorem \ref{thm-packing} the properties $(1)$ to $(4)$ are still fulfilled for $x'$, $y'$ and $B'_t$ and moreover we get
	$\nor{x'}_1 \leq (1+ \Delta) \OPT(I(t),s)- \alpha$ and $\nor{y'}_1 = \max_i B'_t (i) \leq (1+ 2 \Delta) \OPT(I(t),s) + m - \alpha$ for chosen parameter $\alpha$. Let $x''$ and $y''$ be the fractional and integral solution after an operation is applied to $x'$ and $y'$. We have to prove that the properties $(1)$ to $(4)$ are also fulfilled for $x''$ and $y''$.
	\\{\bf operations:} First we take a look at how the operations modify $\nor{x'}_1$ and $\nor{y'}_1 =\max_i B'_t (i)$. By construction of the insertion operation, $\nor{x'}_1$ and $\nor{y'}$ are increased at most by $2$. By construction of the delete operation, $\nor{x'}_1$ and $\nor{y'}_1$ are increased by $1$. By construction of the shiftA and shiftB operation, $\nor{x'}_1$ and $\nor{y'}_1$ are increased by $1$.
	An \textsc{improve}(2) call followed by an insertion operation therefore yields $\nor{y''} = \nor{y'}_1 +2 = (1+ 2\Delta)\OPT(I(t),s) +m -2 +2 = (1+ 2\Delta)\OPT(I(t+1),s) +m$ since $\OPT(I(t),s) \leq \OPT(I(t+1),s)$.
    An \textsc{improve}(4) call followed by a delete operation yields $\nor{y''} = \nor{y'}_1 + 1 = (1+ 2\Delta)\OPT(I(t),s) + m -3 \leq (1+ 2\Delta)\OPT(I(t+1),s) + (1+2\Delta) +m - 3 \leq (1+ 2\Delta)\OPT(I(t+1),s)$ since $\OPT(I(t),s) \leq \OPT(I(t+1),s) +1$ (an item is removed) and $\Delta \leq 1$. In the same way we obtain that $\nor{y''}_1 \leq \nor{y'}_1 +1 \leq (1+ 2\Delta)\OPT(I(t+1),s) +m$ for an \textsc{improve}(1)/\textsc{improve}(3) call followed by a shiftA/shiftB operation. This concludes the proof that property $(1)$ is fulfilled for $I(t+1)$.
	The proof that property $(2)$ holds is analog since $\nor{x'}_1$ increases in the same way as $\nor{y'}_1$ and $\nor{x'}_1   \leq (1+ \Delta) \OPT(I(t),s) - \alpha$.
	For property $(3)$ note that in the operations a configuration $x_i$ of the fractional solution is increased by $1$ if and only if a configuration $y_i$ is increased by $1$. Therefore the property that for all configurations $x''_i \leq y''_i$ retains from $x'$ and $y'$. By Theorem \ref{thm-packing} the number of non-zero components of $x'$ and $y'$ is bounded by $\Delta \OPT(I(t),s) +m \leq \Delta \OPT(I(t+1),s) +m$ in case of an insert operation. If an item is removed, the number of non-zero components of $x'$ and $y'$ is bounded by $\Delta \OPT(I(t),s) +m \leq \Delta \OPT(I(t+1),s) +m +1 = C+1$. By Theorem \ref{thm:reduce} the algorithm \textsc{ReduceComponents} guarantees that there are at most $C=\Delta \OPT(I(t+1),s) +m$ non-zero components. By construction of the shift-operation, $x''$ and $y''$ might have two additional non-zero components. But since these are being reduced by Algorithm \ref{alg-afptas} (note that we increased the number of components being reduced in step 6 by $2$ to- see \cite{jansen2013binpacking} for details), the \ac{lp} solutions $x''$ and $y''$ have at most $\Delta \OPT(I(t+1),s) +m$ non-zero components which proves property $(4)$. Algorithm \ref{alg-afptas} therefore has an asymptotic approximation ratio of $1+\epsilon$.
    
    We still need to examine the migration factor of Algorithm \ref{alg-afptas}. In the case that the offline algorithm is used, the size of the instance is smaller than $8 (\nicefrac{1}{\delta} +1) = \mathcal{O}(\nicefrac{1}{\epsilon})$ or smaller than $(m+2)(\nicefrac{1}{\delta} +2) = \mathcal{O}(\frac{1}{\epsilon^2} \log(\nicefrac{1}{\epsilon}))$. Hence the migration factor in that case is bounded by $\mathcal{O}(\frac{1}{\epsilon^3} \log(\nicefrac{1}{\epsilon}))$. If the instance is bigger the call of \textsc{improve} repacks at most $\mathcal{O}(\nicefrac{m}{\epsilon})$ bins by Theorem \ref{thm-packing}. Since every large arriving item has size $> \nicefrac{\epsilon}{14}$ and $m = \mathcal{O}(\frac{1}{\epsilon} \log (\nicefrac{1}{\epsilon}))$ we obtain a migration factor of $\mathcal{O}(\frac{1}{\epsilon^3} \log (\nicefrac{1}{\epsilon}))$ for the Algorithm \textsc{improve}. Since the migration factor of each operation is also bounded by $\mathcal{O}(\frac{1}{\epsilon^2} \log (\nicefrac{1}{\epsilon}))$, we obtain an overall migration factor of $\mathcal{O}(\frac{1}{\epsilon^3} \log (\nicefrac{1}{\epsilon}))$. 
    
    The main complexity of Algorithm \ref{alg-afptas} lies in the use of
    Algorithm \textsc{improve}. As described by Jansen and Klein
    \cite{jansen2013binpacking} the running time of \textsc{improve} is
    bounded by $\mathcal{O}(M(\nicefrac{1}{\epsilon}
    \log(\nicefrac{1}{\epsilon}))\cdot \nicefrac{1}{\epsilon^3}
    \log(\nicefrac{1}{\epsilon}))$, where $M(n)$ is the time needed to
    solve a system of $n$ linear equations. By using heap structures to
    store the items, each operation can be performed in time
    $\mathcal{O}(\nicefrac{1}{\epsilon}
    \log(\nicefrac{1}{\epsilon})\cdot \log(\epsilon^2\cdot n(t)))$ at
    time $t$, where $n(t)$ denotes the number of items in the instance
    at time $t$. As the number of non-zero components is bounded by
    $\mathcal{O}(\epsilon\cdot n(t))$, the total running time of the
    algorithm is bounded by $\mathcal{O}(M(\nicefrac{1}{\epsilon}
    \log(\nicefrac{1}{\epsilon}))\cdot \nicefrac{1}{\epsilon^3}
    \log(\nicefrac{1}{\epsilon})+\nicefrac{1}{\epsilon}
    \log(\nicefrac{1}{\epsilon}) \log(\epsilon^2\cdot n(t))+\epsilon
    n(t))$. The best known running time for the dynamic \BP problem
    \emph{without} removals was
    $\mathcal{O}(M(\nicefrac{1}{\epsilon^2})\cdot
    \nicefrac{1}{\epsilon^4}+\epsilon
    n(t)+\frac{1}{\epsilon^2}\log(\epsilon^2 n(t)))$ and is due to
    Jansen and Klein \cite{jansen2013binpacking}.  As this is polynomial
    in $n(t)$ and in $\nicefrac{1}{\epsilon}$ we can conclude that
    Algorithm \ref{alg-afptas} is an \ac{afptas}.
\end{proof}

If no deletions are present, we can use a simple FirstFit algorithm (as described by Jansen and Klein \cite{jansen2013binpacking}) to pack the small items into the bins. This does not change the migration factor or the running time of the algorithm and we obtain a robust \ac{afptas} with $\mathcal{O}(\frac{1}{\epsilon^3}\cdot \log(\nicefrac{1}{\epsilon}))$ migration for the case that no items is removed. This improves the best known migration factor of $\mathcal{O}(\frac{1}{\epsilon^4})$ \cite{jansen2013binpacking}.

\section{Handling Small Items}
\label{sec:small}
In this section we present methods for dealing with arbitrary small items in a dynamic online setting. First, we present a robust \ac{afptas} with migration factor of $\mathcal{O}(\nicefrac{1}{\epsilon})$ for the case that only small items arrive and depart. In Section \ref{sec:final} we generalize these techniques to a setting where small items arrive into a packing where large items are already packed and can not be rearranged. Finally we state the \ac{afptas} for the general fully dynamic \BP problem. In a robust setting without departing items, small items can easily be treated by packing them greedily via the classical FirstFit algorithm of Johnson et al. \cite{johnson1974packing} (see Epstein and Levin \cite{epstein2006robust} or Jansen and Klein \cite{jansen2013binpacking}). However, in a setting where items may also depart, small items need to be treated much more carefully. We show that the FirstFit algorithm does not work in this dynamic setting.

\begin{lemma}
\label{lem:firstfitdoesnotwork}
Using the FirstFit algorithm to pack small items may lead to an arbitrarily bad approximation.
\end{lemma}

\begin{proof}
Suppose, that there is an algorithm $\mathcal{A}$ with migration factor $c$ which uses FirstFit on items with size $< \nicefrac{\epsilon}{14}$. We will now construct an instance where $\mathcal{A}$ yields an arbitrary bad approximation ratio. Let $b=\nicefrac{\epsilon}{14}-\delta$ and $a=\nicefrac{\epsilon}{14c}-(\nicefrac{(\delta+c\delta)}{c})$ for a small $\delta$ such that $\nicefrac{(1-b)}{a}$ is integral. Note that $ac<b$ by definition. Furthermore, let $M\in \mathbb{N}$ be an arbitrary integer and consider the instance
\begin{align*}
I_{M}=[\underbrace{A,A,\ldots,A}_M,\underbrace{B,B,\ldots,B}_M]
\end{align*}
with
\begin{align*}
&A=(b,\text{Insert}),\underbrace{(a,\text{Insert}),(a,\text{Insert}),\ldots,(a,\text{Insert})}_{\nicefrac{(1-b)}{a}}\\
&B=\underbrace{(a,\text{Delete}),(a,\text{Delete}),\ldots,(a,\text{Delete})}_{\nicefrac{(1-b)}{a}}.
\end{align*}

After the insertion of all items, there are $M$ bins containing an item of size $b$ and $\nicefrac{1-b}{a}$ items of size $a$ (see Figure \ref{fig:counter}). As $ac<b$, the deletion of the items of size $a$ can not move the items of size $b$. The remaining $M$ bins thus only contain a single item of size $b$ (see Figure \ref{fig:counter1}), while $\lceil M\cdot b\rceil$ bins would be sufficient to pack all of the remaining items. The approximation ratio is thus at least $\nicefrac{M}{M\cdot b}=\nicefrac{1}{b}\approx \frac{1}{\epsilon}$ and thus grows as $\epsilon$ shrinks. In order to avoid this problem, we design an algorithm which groups items of similar size together. Using such a mechanism would therefore put the second item of size $b$ into the first bin by shifting out an appropriate number of items of size $a$ and so on. Our algorithms achieves this grouping of small items by enumerating the bins and maintaining the property, that larger small items are always left of smaller small items. 


\begin{figure}[ht]
\begin{subfigure}{.45\textwidth}
 \begin{tikzpicture}[node distance=0.1cm]
 \draw (0,0) to (0,3);
 \draw (1,0) to (1,3);
 \draw (0,0) to (1,0);
 \draw (0,3) to (1,3);
 
 \node (b) at (0.5,0.4) {$b$};
 
 \node (a) at (2,1.5) {$a$};
 
 \draw (0,0.8) to (1,0.8);
 
 \foreach \x in {1.0,1.2,...,2.8}{
 \draw (0,\x) to (1,\x);
 \draw[->] (a) to (0.5,\x-0.1);
 }
 \draw[->] (a) to (0.5,2.9);
 
 \end{tikzpicture}

\caption{A single bin after the insertion} \label{fig:counter}
\end{subfigure}
\begin{subfigure}{.45\textwidth}

 \begin{tikzpicture}[node distance=1cm]
 \draw (0,0) to (0,3);
 \draw (1,0) to (1,3);
 \draw (0,0) to (1,0);
 \draw (0,3) to (1,3);
 
 \node (b) at (0.5,0.4) {$b$};

 \draw (0,0.8) to (1,0.8);

 \end{tikzpicture}

\caption{A single bin after the deletion}
\label{fig:counter1}

\end{subfigure}

\caption{Construction in the proof of Lemma \ref{lem:firstfitdoesnotwork}}
\end{figure}
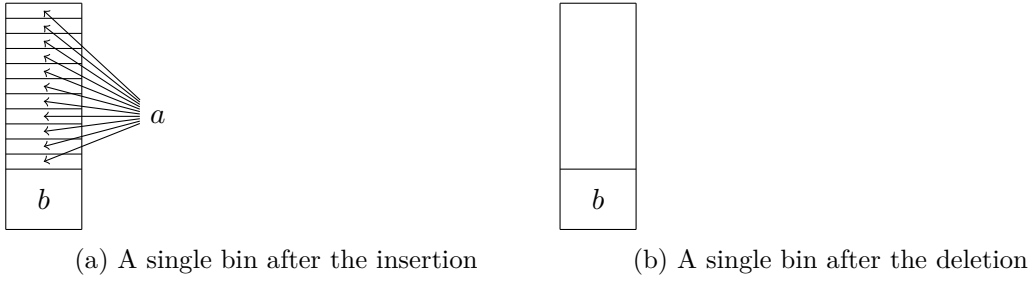

\end{proof}

\subsection{Only Small Items}
\label{sec:smallitems}
We consider a setting where only small items exist, i.\,e., items with a size less than $ \nicefrac{\epsilon}{14}$.
First, we divide the set of small items into different size intervals $S_j$ where $S_j=\left[\frac{\epsilon}{2^{j+1}},\frac{\epsilon}{2^j}\right)$ for $j \geq 1$. Let $b_1,\ldots, b_m$ be the used bins of our packing. We say a size category $S_j$ is bigger than a size category $S_k$ if $j<k$, i.\,e., the item sizes contained in $S_j$ are larger (note that a size category $S_j$ with large index $j$ is called small).
We say a bin $b_i$ is filled completely if it has less than $\frac{\epsilon}{2^j}$ remaining space, where $S_j$ is the biggest size category appearing in $b_i$.
Furthermore we label bins $b_i$ as \emph{normal} or as \emph{buffer bins} and partition all bins $b_1,\ldots, b_m$ into \emph{queues} $Q_1, \ldots, Q_{d}$ for $|Q| \leq m$. A queue is a subsequence of bins $b_i, b_{i+1} \ldots, b_{i+c}$ where bins $b_i, \ldots, b_{i+c-1}$ are normal bins and bin $b_{i+c}$ is a buffer bin. We denote the $i$-th queue by~$Q_i$ and the number of bins in $Q_i$ by $|Q_i|$. The buffer bin of queue $Q_i$ is denoted by $bb_i$.

We will maintain a special form for the packing of small items such that the following properties are always fulfilled. For the sake of simplicity, we assume that $\nicefrac{1}{\epsilon}$ is integral. 
\begin{compactenum}
\item[(1)] For every item $i\in b_d$ with size $s(i)\in S_j $ for some $j,d \in \mathbb{N}$, there is no item $i' \in b_{d'}$ with size $s(i') \in S_{j'}$ such that $d'>d$ and $j' > j$. This means: Items are ordered from left to right by their size intervals.
\item[(2)] Every normal bin is filled completely.
\item[(3)] The length of each queue is at least $\nicefrac{1}{\epsilon}$ and at most $\nicefrac{2}{\epsilon}$ except for the last queue $Q_d$. 
\end{compactenum}
Note that property (1) implies that all items in the same size interval $S_j$ are packed into bins $b_x, b_{x+1}, \ldots , b_{x+c}$ for constants $x$ and $c$. Items in the next smaller size category $S_{j+1}$ are then packed into bins $b_{x+c}, b_{x+c+1}, \ldots$ and so on. We denote by $b_{S(\ell)}$ the last bin in which an item of size interval $S_\ell$ appears. We denote by $S_{>\ell}$ the set of smaller size categories $S_{\ell'}$ with $\ell' > \ell$. Note that items in size category $S_{>\ell}$ are smaller than items in size category $S_\ell$.

\tikzset{   brace/.style={
     decoration={brace, mirror},
     decorate
   },
      position label/.style={
       below = 3pt,
       text height = 1.5ex,
       text depth = 1ex
    }
}

\begin{figure}[ht]
\centering
\resizebox{\textwidth}{!}{
\begin{tikzpicture}
\bin{1}{0}{$b_1$}
\node[xshift=1.5cm] (ld1) {$\ldots$};
\bin{2}{3cm}{$b_{|Q_1|-1}$}
\bin{3}{4cm}{$bb_1$}
\draw [brace] (1.south) -- (3.south) node [position label, pos=0.5] {$Q_1$};

\bin{4}{6cm}{$b_{|Q_1|+1}$}
\node[xshift=7.5cm] (ld2) {$\ldots$};
\bin{5}{9cm}{$b_{|Q_2|-1}$}
\bin{6}{10cm}{$bb_2$}
\draw [brace] (4.south) -- (6.south) node [position label, pos=0.5] {$Q_2$};
\node[xshift=11.5cm] (ld3) {$\ldots$};

\bin{7}{13cm}{$b_{|Q_{d-1}|+1}$}
\node[xshift=14.5cm] (ld2) {$\ldots$};
\bin{8}{16cm}{$b_{|Q_d|-1}$}
\bin{9}{17cm}{$bb_d$}
\draw [brace] (7.south) -- (9.south) node [position label, pos=0.5] {$Q_d$};
\end{tikzpicture}}
\caption{Distribution of bins with small items into queues}
\end{figure}
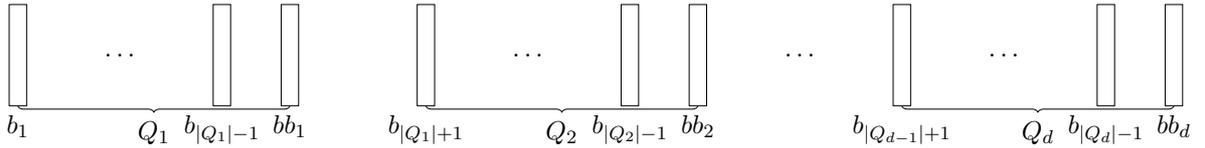

The following lemma guarantees that a packing that fulfills properties $(1)$ to $(3)$ is close to the optimum solution.
\begin{lemma} \label{lem:small_approx}
If properties $(1)$ to $(3)$ hold, then at most $(1+\mathcal{O}(\epsilon))\OPT(I,s)+2$ bins are used in the packing for every $\epsilon \leq \nicefrac{1}{3}$.
\end{lemma}
\begin{proof}
    Let $C$ be the number of used bins in our packing. By property (2) we know that all normal bins have less than $\nicefrac{\epsilon}{14}$ free space. Property (3) implies that there are at most $\epsilon\cdot C +1$ buffer bins and hence possibly empty. The number of normal bins is thus at least $(1-\epsilon)\cdot C-1$. Therefore we can bound the total size of all items by $\geq (1-\nicefrac{\epsilon}{14})\cdot ((1-\epsilon)\cdot C-1)$. As $\OPT(I,s)\geq SIZE(I,s) \geq (1-\nicefrac{\epsilon}{14})\cdot ((1-\epsilon)\cdot C-1)$ and $\frac{1}{(1-\nicefrac{\epsilon}{14})(1-\epsilon)} \leq 1+ 2\epsilon$ for $\epsilon \leq \nicefrac{1}{3}$ we get $C \leq (1+2\epsilon) \OPT(I,s)  +2.$ 
\end{proof}

We will now describe the operations that are applied whenever a small item has to be inserted or removed from the packing. The operations are designed such that properties $(1)$ to $(3)$ are never violated and hence a good approximation ratio can be guaranteed by Lemma \ref{lem:small_approx} at every step of the algorithm. 
The operations are applied recursively such that some items from each size interval are shifted from left to right (insert) or right to left (delete). The recursion halts if the first buffer bin is reached. Therefore, the free space in the buffer bins will change over time. Since the recursion always halts at the buffer bin, the algorithm is applied on a single queue $Q_k$.

The following Insert/Delete operation is defined for a whole set $J = \{i_1, \ldots , i_n \}$ of items. If an item $i$ of size interval $S_{\ell}$ has to be inserted or deleted, the algorithm is called with Insert$(\{ i \},b_{S(\ell)}, Q_k)$ respectively Delete$(\{i\},b_x, Q_k)$, where $b_x$ is the bin containing item $i$ and $Q_k$ is the queue containing bin $b_{S(\ell)}$ or $b_x$. Recall that $S_{j}=\left[\frac{\epsilon}{2^{j+1}},\frac{\epsilon}{2^{j}}\right)$ is a fixed interval for every $j\geq 1$ and $S_{\leq j}=\bigcup_{i=1}^{j}S_{i}$ and $S_{>j}=\bigcup_{i>j}S_{i}$.
\begin{algo}[Insert or Delete for only small items]
\label{alg-insertsmall}
\ 
\begin{compactitem}
\item {\bf Insert$(J,b_x,Q_k)$:} 

\begin{compactitem}
\item Insert the set of small items $J = \{i_1, \ldots , i_n \}$ with size $s(i_j) \in S_{\leq \ell}$ into bin $b_x$. (By Lemma \ref{onlysmallitems} the total size of $J$ is bounded by $\mathcal{O}(\nicefrac{1}{\epsilon})$ times the size of the item which triggered the first Insert operation.)
\item Remove just as many items $J' = \{ i'_1, \ldots , i'_m \}$ of the smaller size interval $S_{>\ell}$ appearing in bin $b_x$ (starting by the smallest) such that the items $i_1, \ldots, i_n$ fit into the bin $b_x$. If there are not enough items of smaller size categories to insert all items from $I$, insert the remaining items from $I$ into bin $b_{x+1}$.
\item Let $J'_{\ell'}\subseteq J'$ be the items in the respective size
  interval $S_{\ell'}$ with $\ell'>\ell$. Put the items $J'_{\ell'}$ recursively into bin $b_{S(\ell')}$ (i.\,e., call Insert$(J'_{\ell'},b_{S(\ell')},Q_k)$ for each $\ell'>\ell$). If the buffer bin $bb_k$ is left of $b_{S(\ell')}$ call Insert$(J'_{\ell'},bb_k,Q_k)$ instead.
\end{compactitem}

\item {\bf Delete$(J,b_x,Q_k)$:} 
\begin{compactitem}
\item Remove the set of items $J= \{ i_1, \ldots , i_n \}$ with size $s(i_j) \in S_{\leq \ell}$ from bin $b_x$ (By Lemma \ref{onlysmallitems} the total size of $J$ is bounded  by $\mathcal{O}(\nicefrac{1}{\epsilon})$ times the size of the item which triggered the first Delete operation.)
\item Insert as many small items $J'=\{i'_1, \ldots , i'_m\}$ from $b_{S(\ell')}$, where $S_{\ell'}$ is the smallest size interval appearing in $b_x$ such that $b_x$ is filled completely. If there are not enough items from the size category $S_{\ell'}$, choose items from size category $S_{\geq \ell'+1}$ in bin $b_{x+1}$.
\item Let $J'_{\ell'}\subseteq J'$ be the items in the respective size interval $S_{\ell'}$ with $\ell'>\ell$. Remove items $J'_{\ell'}$ from bin $b_{S(\ell')}$ recursively (i.\,e., call Delete$(J'_{\ell'},b_{S(\ell')},Q_k)$ for each $\ell'>\ell$). If the buffer bin $bb_k$ is left of $b_{S(\ell')}$, call Delete$(J'_{\ell'},bb_k,Q_k)$ instead.
\end{compactitem}
\end{compactitem}
\end{algo}
Using the above operations maintains the property of normal bins to be filled completely. However, the size of items in buffer bins changes. In the following we describe how to handle buffer bins that are being emptied or filled completely. 
\begin{algo}[Handle filled or emptied buffer bins]
\label{alg-bb}
\
\begin{compactitem}
\item {\bf Case 1: The buffer bin of $Q_i$ is filled completely by an insert operation.}
\begin{compactitem}
\item Label the filled bin as a normal bin and add a new empty buffer bin to the end of $Q_i$. 
\item If $|Q_i|>\nicefrac{2}{\epsilon}$, split $Q_i$ into two new queues $Q'_i,Q''_i$ with $|Q''_i|=|Q'_i|+1$. The buffer bin of $Q''_i$ is the newly added buffer bin. Add an empty bin labeled as the buffer bin to $Q'_i$ such that $|Q'_i|=|Q''_i|$. 
\end{compactitem}
\item {\bf Case 2: The buffer bin of $Q_i$ is being emptied due to a delete operation.}
\begin{compactitem}
\item Remove the now empty bin. 
\item If $|Q_{i}| \geq |Q_{i+1}|$ and $|Q_{i}|> \nicefrac{1}{\epsilon}$, choose the last bin of $Q_{i}$ and label it as new buffer bin of $Q_i$. 
\item If $|Q_{i+1}| > |Q_i|$ and $|Q_{i+1}|> \nicefrac{1}{\epsilon}$, choose the first bin of $Q_{i+1}$ and move the bin to $Q_i$ and label it as buffer bin. 
\item If $|Q_{i+1}| = |Q_i| = \nicefrac{1}{\epsilon}$, merge the two queues $Q_i$ and $Q_{i+1}$. As $Q_{i+1}$ already contains a buffer bin, there is no need to label another bin as buffer bin for the merged queue.
\end{compactitem}
\end{compactitem}
\end{algo}
Creating and deleting buffer bins this way guarantees that property (3) is never violated since queues never exceed the length of $\nicefrac{2}{\epsilon}$ and never fall below $\nicefrac{1}{\epsilon}$.

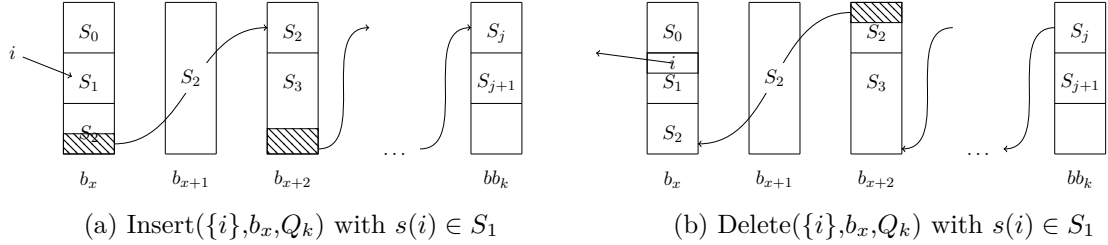
\begin{figure}[ht]
  \begin{subfigure}{.49\textwidth}
 \scalebox{0.67}{
 \begin{tikzpicture}[node distance=0.1cm]
 \draw (0,0) to (0,3);
 \draw (1,0) to (1,3);
 \draw (0,0) to (1,0);
 \draw (0,3) to (1,3);
 
 \node (s0) at (0.5,2.4) {$S_0$};
 \node (s1) at (0.5,1.4) {$S_1$};
 \node (s2) at (0.5,0.4) {$S_2$};
 \node (bx) at (0.5,-0.5) {$b_x$};
 
 \draw[pattern=north west lines, pattern color=black] (0,0) rectangle (1,0.4);
 
 \node (i)  at (-1,2) {$i$};
 
 \draw[->] (i) to (s1);

 \foreach \x in {1.0,2.0}{
 \draw (0,\x) to (1,\x);
  }

\draw[->] (1,0.2) to[out = 0, in = 180, looseness = 1] (4,2.5);

 \draw (2,0) to (2,3);
 \draw (3,0) to (3,3);
 \draw (2,0) to (3,0);
 \draw (2,3) to (3,3);
 
  \node[fill=white] (s2) at (2.5,1.5) {$S_2$};
  
 \node (bx) at (2.5,-0.5) {$b_{x+1}$};
 
 \draw (4,0) to (4,3);
 \draw (5,0) to (5,3);
 \draw (4,0) to (5,0);
 \draw (4,3) to (5,3);
 
 \node (s0) at (4.5,2.4) {$S_2$};
 \node (s1) at (4.5,1.4) {$S_3$};
 \draw (4,2) to (5,2);
 \draw[pattern=north west lines, pattern color=black] (4,0) rectangle (5,0.5);
 \node (bx) at (4.5,-0.5) {$b_{x+2}$};

\node (ld) at (6.5,0) {$\ldots$};

\draw[->] (5,0.1) to[out = 0, in = 180, looseness = 1] (6,2.5);

 \draw (8,0) to (8,3);
 \draw (9,0) to (9,3);
 \draw (8,0) to (9,0);
 \draw (8,3) to (9,3);
 
 \node (s0) at (8.5,2.4) {$S_{j}$};
 \node (s1) at (8.5,1.4) {$S_{j+1}$};
 \node (s2) at (8.5,0.4) {};
 \node (bx) at (8.5,-0.5) {$bb_{k}$};

\foreach \x in {1.0,2.0}{
 \draw (8,\x) to (9,\x);
  }
  
  \draw[->] (7,0.1) to[out = 0, in = 180, looseness = 1] (8,2.5);

 \end{tikzpicture}}

    \caption{Insert($\{i \}$,$b_x$,$Q_k$) with $s(i)\in S_{1}$}
  \label{fig:insert}
  \end{subfigure}%
  \begin{subfigure}{.49\textwidth}
 \scalebox{0.67}{
   \begin{tikzpicture}[node distance=0.1cm]
 \draw (0,0) to (0,3);
 \draw (1,0) to (1,3);
 \draw (0,0) to (1,0);
 \draw (0,3) to (1,3);
 
 \node (s0) at (0.5,2.4) {$S_0$};
 \node (s1) at (0.5,1.4) {$S_1$};
 \node (s2) at (0.5,0.4) {$S_2$};
 \node (bx) at (0.5,-0.5) {$b_x$};
 
 \draw (0,1.6) rectangle (1,2.0);
 
 \node (i)  at (0.5,1.8) {$i$};
 
 \draw[->] (0.5,1.8) to (-1,2);

 \foreach \x in {1.0,2.0}{
 \draw (0,\x) to (1,\x);
  }

\draw[<-] (1,0.2) to[out = 0, in = 180, looseness = 1] (4,2.8);

 \draw (2,0) to (2,3);
 \draw (3,0) to (3,3);
 \draw (2,0) to (3,0);
 \draw (2,3) to (3,3);
 
  \node[fill=white] (s2) at (2.5,1.5) {$S_2$};
  
 \node (bx) at (2.5,-0.5) {$b_{x+1}$};
 
 \draw (4,0) to (4,3);
 \draw (5,0) to (5,3);
 \draw (4,0) to (5,0);
 \draw (4,3) to (5,3);
 
 \node (s0) at (4.5,2.4) {$S_2$};
 \node (s1) at (4.5,1.4) {$S_3$};
 \draw (4,2) to (5,2);
 \draw[pattern=north west lines, pattern color=black] (4,2.6) rectangle (5,3.0);
 \node (bx) at (4.5,-0.5) {$b_{x+2}$};

\node (ld) at (6.5,0) {$\ldots$};

\draw[<-] (5,0.1) to[out = 0, in = 180, looseness = 1] (6,2.5);

 \draw (8,0) to (8,3);
 \draw (9,0) to (9,3);
 \draw (8,0) to (9,0);
 \draw (8,3) to (9,3);
 
 \node (s0) at (8.5,2.4) {$S_{j}$};
 \node (s1) at (8.5,1.4) {$S_{j+1}$};
 \node (s2) at (8.5,0.4) {};
 \node (bx) at (8.5,-0.5) {$bb_{k}$};

\foreach \x in {1.0,2.0}{
 \draw (8,\x) to (9,\x);
  }
  
  \draw[<-] (7,0.1) to[out = 0, in = 180, looseness = 1] (8,2.5);

 \end{tikzpicture}}

  \caption{Delete($\{i \}$,$b_x$,$Q_k$) with $s(i)\in S_{1}$}
  \label{fig:delete}
  \end{subfigure}

\caption{Example calls of Insert and Delete.}
\end{figure}

Figure \ref{fig:insert} shows an example call of Insert($\{ i \}$,$b_x$,$Q_k$). Item $i$ with $s(i)\in S_1$ is put into the corresponding bin $b_x$ into the size interval $S_1$. As $b_x$ now contains too many items, some items from the smallest size interval $S_2$ (marked by the dashed lines) are put into the last bin $b_{x+2}$ containing items from $S_2$. Those items in turn push items from the smallest size interval $S_3$ into the last bin containing items of this size and so on. This process terminates if either no items need to be shifted to the next bin or the buffer bin $bb_k$ is reached.

It remains to prove that the migration of the operations is bounded and that the properties are invariant under those operations.

\begin{lemma}
\label{onlysmallitems}
\ 
    \begin{enumerate}
        \item[(i)] Let $I$ be an instance that fulfills properties $(1)$ to $(3)$. Applying operations insert/delete on $I$ yields an instance $I'$ that also fulfills properties $(1)$ to $(3)$. 
        \item[(ii)] The migration factor of a single insert/delete operation is bounded by $\mathcal{O}(\nicefrac{1}{\epsilon})$ for all $\epsilon\leq \nicefrac{2}{7}$.
    \end{enumerate}
\end{lemma}
\begin{proof}
Proof for (i): Suppose the insert/delete operation is applied to a packing which fulfills properties $(1)$ to $(3)$. By construction of the insert operation, items from a size category $S_\ell$ in bin $b_x$ are shifted to a bin $b_y$. The bin $b_y$ is either $b_{S(\ell)}$ or the a buffer bin left of $b_{S(\ell)}$. By definition $b_y$ contains items of size category $S_{\ell}$. Therefore property $(1)$ is not violated. Symmetrically, by construction of the delete operation, items from a size category $S_\ell$ in bin $b_{S(\ell)}$ are shifted to a bin $b_x$. By definition $b_x$ contains items of size category $S_{\ell}$ and property $(1)$ is therefore not violated.
For property $(2)$: Let $b_x$ be a normal bin, where items $i_1, \ldots , i_n$ of size category $S_{\leq \ell}$ are inserted. We have to prove that the free space in $b_x$ remains smaller than $\nicefrac{\epsilon}{2^j}$, where $S_j$ is the smallest size category appearing in bin $b_x$. By construction of the insert operation, just as many items of size categories $S_{>\ell}$ are shifted out of bin $b_x$ such that $i_1,\ldots,i_n$ fit into $b_x$. Hence the remaining free space is less than $\frac{\epsilon}{2^{\ell}}$ and bin $b_x$ is filled completely. The same argumentation holds for the delete operation.
Property $(3)$ is always fulfilled by definition of Algorithm \ref{alg-bb}.

    Proof for (ii): According to the insert operation, in every recursion step of the algorithm, it tries to insert a set of items into a bin $b_{x'}$,  starting with an Insert$(\{ i \},b_{x'},Q_k)$ operation. Let $\INSERT(S_{\leq \ell+y}, b_x)$ ($x\geq x'$) be the size of all items in size categories $S_{j}$ with $j \leq \ell+y$ that the algorithm tries to insert into $b_x$ as a result of an Insert$(\{ i \},b_{x'},Q_k)$ call. Let $\PACK(b_x)$ be the size of items that are actually packed into bin $b_x$. 
    We have to distinguish between two cases. In the case that $\INSERT(S_{\leq \ell+y}, b_x) = \PACK(b_x)$ there are enough items of smaller size categories $S_{>\ell+y}$ that can be shifted out, such that items $I$ fit into bin $b_x$. In the case that $\INSERT(S_{\leq \ell+y}, b_x) > \PACK(b_x)$ there are not enough items of smaller size category that can be shifted out and the remaining size of $\INSERT(S_{\leq \ell+y}, b_x) - \PACK(b_x)$ has to be shifted to the following bin $b_{x+1}$. Under the assumption that each $\INSERT(S_{\leq \ell}, b_x)\leq 1$ for all $x$ and $\ell$ (which is shown in the following) all items fit into $b_{x+1}$. Note that no items from bins left of $b_{x}$ can be shifted into $b_{x+1}$ since $b_x = b_{S(\ell+y)}$ is the last bin where items of size category $S_{\leq \ell+y}$ appear. Hence all items shifted out from bins left of $b_{x}$ are of size categories $S_{\leq \ell +y}$ (property $(1)$) and they are inserted into bins left of $b_{x+1}$.
    We prove by induction that for each $\INSERT(S_{\leq \ell+y}, b_x)$ the total size of moved items is at most 
    \begin{align*} \INSERT(S_{\leq \ell+y}, b_x) \leq s(i) + 3 \sum_{j=1}^{y} \frac{\epsilon}{2^{\ell+j}} \end{align*}
The claim holds obviously for $\INSERT(S_{\leq \ell}, b_{x'})$ since $b_{x'}=b_{S(\ell)}$ is the bin where only item $i$ is inserted. 

\begin{figure}[ht]
  \begin{subfigure}{0.3\textwidth}
    
\scalebox{0.55}{
 \begin{tikzpicture}[node distance=0.1cm]
 \draw (0,0) to (0,3);
 \draw (1,0) to (1,3);
 \draw (0,0) to (1,0);
 \draw (0,3) to (1,3);
 
 \node (bx) at (0.5,-0.5) {$b_x$};
 
 \node (s1) at (0.5,2.5) {$S_1$};
 \draw (0,1) to (1,1);
 \node (s2) at (0.5,0.5) {$S_2$};

 \node (i)  at (-2,2) {Insert $[S_1,S_2]$};

 \draw (2,0) to (2,3);
 \draw (3,0) to (3,3);
 \draw (2,0) to (3,0);
 \draw (2,3) to (3,3);

 \node (bx) at (2.5,-0.5) {$b_{x+1}$};
 
 \node (s2p) at (2.5,2.5) {$S_2$};
 \draw (2,1.5) to (3,1.5);
 \node (s3) at (2.5,0.5) {$S_3$};
 
\draw[->] (s2) to (s2p);
 \draw[->] (i.315) to (0.5,1.5);
\draw[->] (i.20) to[out=90] (s2p);

 \end{tikzpicture}}

    \caption{Case 1}
  \end{subfigure}
  \begin{subfigure}{0.3\textwidth}
\scalebox{0.55}{
 \begin{tikzpicture}[node distance=0.1cm]
 \draw (0,0) to (0,3);
 \draw (1,0) to (1,3);
 \draw (0,0) to (1,0);
 \draw (0,3) to (1,3);
 
 \node (bx) at (0.5,-0.5) {$b_{\hat{x}}$};
 
 \node (s1) at (0.5,2.5) {$S_1$};
 \draw (0,1) to (1,1);
 \node (s2) at (0.5,0.5) {$S_2$};

 \node (i)  at (-2,2) {Insert $[S_1]$};

\node at (2.5,0) {$\ldots$};

 \draw (4,0) to (4,3);
 \draw (5,0) to (5,3);
 \draw (4,0) to (5,0);
 \draw (4,3) to (5,3);

 \node (bx) at (4.5,-0.5) {$b_{x+1}$};
 
 \node (s2p) at (4.5,2.5) {$S_2$};
 \draw (4,1.5) to (5,1.5);
 \node (s3) at (4.5,0.5) {$S_3$};
 
\draw[->] (s2) to (s2p);
 \draw[->] (i.315) to (0.5,1.5);

 \end{tikzpicture}}
    
    \caption{Case 2a}
  \end{subfigure}
  \begin{subfigure}{0.3\textwidth}

\scalebox{0.55}{
 \begin{tikzpicture}[node distance=0.1cm]
 \draw (0,0) to (0,3);
 \draw (1,0) to (1,3);
 \draw (0,0) to (1,0);
 \draw (0,3) to (1,3);
 
 \node (bx) at (0.5,-0.5) {$b_{\hat{x}}$};
 
 \node (s1) at (0.5,2.5) {$S_1$};
 \draw (0,1) to (1,1);
 \node (s2) at (0.5,0.5) {$S_2$};

 \node (i)  at (-2,2) {Insert $[S_1, S_{2}]$};

 \draw (2,0) to (2,3);
 \draw (3,0) to (3,3);
 \draw (2,0) to (3,0);
 \draw (2,3) to (3,3);
 
\node (s2pp) at (2.5,2.2) {$S_{2}$};
  
 \node (bx) at (2.5,-0.5) {$b_{\hat{x}+1}$};

\node at (4.5,0) {$\ldots$};

 \draw (6,0) to (6,3);
 \draw (7,0) to (7,3);
 \draw (6,0) to (7,0);
 \draw (6,3) to (7,3);

 \node (bx) at (6.5,-0.5) {$b_{x+1}$};
 
 \node (s2p) at (6.5,2.5) {$S_2$};
 \draw (6,1.5) to (7,1.5);
 \node (s3) at (6.5,0.5) {$S_3$};
 
\draw[->] (s2) to (s2p);
 \draw[->] (i.315) to (0.5,1.5);
\draw[->] (i.20) to[out=90] (s2pp);

\draw[->] (2.5,0.5) to (s2p);
 \end{tikzpicture}}
    
    \caption{Case 2b}
  \end{subfigure}

\caption{All cases to consider in Lemma \ref{onlysmallitems}}
\end{figure}
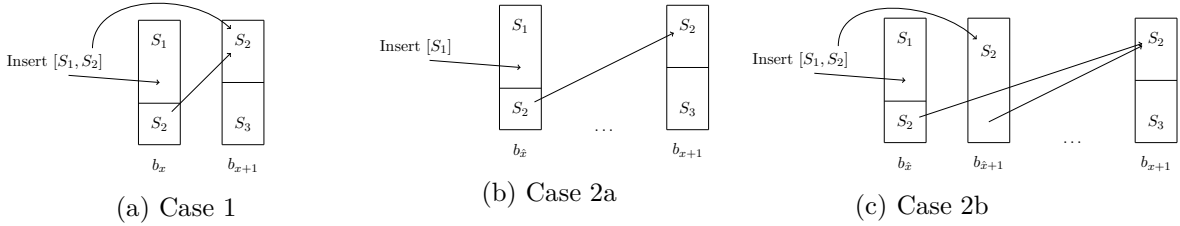

Case 1: $\INSERT(S_{\leq \ell + y}, b_x) > \PACK(b_x)$\\ 
In this case, the size of all items that have to be inserted into $b_{x+1}$ can be bounded by the size of items that did not fit into bin $b_x$ plus the size of items that were removed from bin $b_x$. We can bound $\INSERT(S_{\leq \ell+\bar{y}}, b_{x+1})$ where $\bar{y} > y$ is the largest index $S_{\ell+\bar{y}}$ appearing in bin $b_{x}$ by
\begin{align*}
 \INSERT(S_{\leq \ell + y}, b_x) + \frac{\epsilon}{2^{\ell+y}} \leq s(i) + 3 \sum_{j=1}^{y} \frac{\epsilon}{2^{\ell+j}} + 2 \frac{\epsilon}{2^{\ell+y +1}} < s(i) + 3 \sum_{j=1}^{y+1} \frac{\epsilon}{2^{\ell+j}}
\end{align*}

Case 2: $\INSERT(S_{\leq \ell+y}, b_x) = \PACK(b_x)$\\
Suppose that the algorithm tries to insert a set of items $I$ of size categories $S_{\leq \ell+\bar{y}}$  into the bin $b_{x+1} = b_{S(\ell+\bar{y})}$. The items $I$ can only be shifted from previous bins where items of size category $S_{\leq \ell+\bar{y}}$ appear. 
There are only two possibilities remaining. Either all items $I$ are shifted from a single bin $b_{\hat{x}}$ ($\hat{x} \leq x$) or from two consecutive bins $b_{\hat{x}},b_{\hat{x}+1}$ with $\INSERT(S_{\leq \ell+y}, b_{\hat{x}}) > \PACK(b_{\hat{x}})$. 

Note that $b_{x+1}$ can only receive items from more than one bin if there are two bins $b_{\hat{x}},b_{\hat{x}+1}$ with $\INSERT(S_{\leq \ell+y}, b_{\hat{x}}) > \PACK(b_{\hat{x}})$ such that $b_{x+1}=b_{S(\ell+\bar{y})}$ and all items shifted out of $b_{\hat{x}},b_{\hat{x}+1}$ and into $b_{x+1}$ are of size category $S_{\ell+\bar{y}}$. Hence bins left of $b_{\hat{x}}$ or right of $b_{\hat{x}+1}$ can not shift items into $b_{x+1}$.

Case 2a: All items $I$ are shifted from a single bin $b_{\hat{x}}$ with $\hat{x} \leq x$ (note that $\hat{x}<x$ is possible since $\PACK(b_x) = \INSERT(S_{\leq \ell+y}, b_x)$ can be zero). The total size of items that are shifted out of $b_{\hat{x}}$ can be bounded by $\INSERT(S_{\leq \ell+y},b_{\hat{x}})+\frac{\epsilon}{2^{\ell+y}}$. By induction hypothesis $\INSERT(S_{\leq \ell+y}, b_{\hat{x}})$ is bounded by $s(i)+3\sum_{j=1}^{y}\frac{\epsilon}{2^{\ell+j}}$. Since all items that are inserted into $b_{x+1}$ come from $b_{\hat{x}}$, the value $\INSERT(S_{\leq \ell+\bar{y}},b_{x+1})$ ($\bar{y}>y$) can be bounded by $\INSERT(S_{\leq \ell+y}, b_{\hat{x}}) + \frac{\epsilon}{2^{\ell+y}} \leq s(i)+3\sum_{j=1}^{y}\frac{\epsilon}{2^{\ell+j}}+\frac{\epsilon}{2^{\ell+y}}< s(i) + 3 \sum_{j=1}^{\bar{y}} \frac{\epsilon}{2^{\ell+j}}$ where $S_{\ell+\bar{y}}$ is the smallest size category inserted into $b_{x+1}$. Note that the items $I$ belong to only one size category $S_{\ell+\bar{y}}$ if $\hat{x}<x$ since all items that are in size intervals $S_{<\ell+\bar{y}}$ are inserted into bin $b_{\hat{x}+1}$.

Case 2b: Items $I$ are shifted from bins $b_{\hat{x}}$ and $b_{\hat{x}+1}$ ($\hat{x}+1 \leq x$) with $\INSERT(S_{\leq \ell+y}, b_{\hat{x}}) > \PACK(b_{\hat{x}})$. In this case, all items $I$ belong to the size category $S_{\ell+\bar{y}}$ since $b_{\hat{x}}$ is left of $b_x$. Hence all items which are inserted into $b_{\hat{x}+1}$ are from $I$, i.\,e., $\INSERT(S_{\leq \ell+y},b_{\hat{x}}) = \PACK(b_{\hat{x}}) + \PACK(b_{\hat{x}+1})$ as all items in $I$ belong to the same size category $S_{\ell+\bar{y}}$. We can bound $\INSERT(S_{\ell+\bar{y}},b_{x+1})$ by the size of items that are shifted out of $b_{\hat{x}}$ plus the size of items that are shifted out of $b_{\hat{x}+1}$. We obtain
\begin{align*}
&\INSERT(S_{\leq \ell+\bar{y}},b_{x+1})
\leq \PACK(b_{\hat{x}}) + \frac{\epsilon}{2^{\ell+ y}} + \PACK(b_{\hat{x}+1}) +\frac{\epsilon}{2^{\ell+ \bar{y}}} \\
&= \INSERT(S_{\leq \ell+y},b_{\hat{x}}))+\frac{\epsilon}{2^{\ell+ y}}+\frac{\epsilon}{2^{\ell+ \bar{y}}}\\
&\leq s(i) + 3 \sum_{j=1}^{y} \frac{\epsilon}{2^{\ell+j}}+\frac{\epsilon}{2^{\ell+ y}}+\frac{\epsilon}{2^{\ell+ \bar{y}}}\\
&\leq s(i) + 3\sum_{j=1}^{y} \frac{\epsilon}{2^{\ell+j}}+3\frac{\epsilon}{2^{\ell+ \bar{y}}} 
\leq s(i) + 3 \sum_{j=1}^{\bar{y}} \frac{\epsilon}{2^{\ell+j}}
\end{align*}

This yields that $\INSERT(S_{\leq \ell+y},b_x)$ is bounded by $s(i) + 3 \sum_{j=1}^{\bar{y}} \frac{\epsilon}{2^{\ell+j}}$ for all bins $b_x$ in $Q_k$. 
Now, we can bound the migration factor for every bin $b_x$ of $Q_k$ for any $y\in \mathbb{N}$ by $\PACK(b_x) + \frac{\epsilon}{2^{\ell+y}} \leq \INSERT(S_{\leq \ell+y},b_x) + \frac{\epsilon}{2^{\ell+y}}$. Using the above claim, we get:
\begin{align*}
    &\INSERT(S_{\leq \ell+y},b_x) + \frac{\epsilon}{2^{\ell+y}} 
    \leq s(i) + 3 \sum_{j=1}^{y} \frac{\epsilon}{2^{\ell+j}} + 2 \frac{\epsilon}{2^{\ell+y+1}}\\ 
    &< s(i) + 3 \sum_{j=1}^{\infty} \frac{\epsilon}{2^{\ell+j}} 
    = s(i) + 3 \frac{\epsilon}{2^\ell}\sum_{j=1}^{\infty} \frac{1}{2^{j}} 
    = s(i) + 3 \cdot \frac{\epsilon}{2^\ell} 
    \leq 7 s(i)
\end{align*}
    Since there are at most $\nicefrac{2}{\epsilon}$ bins per queue, we
    can bound the total migration of Insert$(\{ i \},b_{S(\ell)}, Q_k)$
    by $7 \cdot \nicefrac{2}{\epsilon} \in
    \mathcal{O}(\nicefrac{1}{\epsilon})$. Note also that $s(i) \leq
    \nicefrac{\epsilon}{14}$ for every $i$ implies that $\INSERT(S_{\leq
      \ell}, b_x)$ is bounded by $\nicefrac{\epsilon}{2}$ for all $x$ and $\ell$ .

    Suppose that items $i_1, \ldots , i_n$ of size interval $S_{\ell+y}$ have to be removed from bin $b_x$. In order to fill the emerging free space, items from the same size category are moved out of $b_{S(\ell)}$ into the free space. As the bin $b_x$ may already have additional free space, we need to move at most a size of $\SIZE(i_1,\ldots,i_n)+\nicefrac{\epsilon}{2^{\ell+y}}$. Using a symmetric proof as above yields a migration factor of $\mathcal{O}(\frac{1}{\epsilon})$.
\end{proof}

\subsection{Handling small items in the general setting}
\label{sec:general}
In the scenario that there are mixed item types (small and large items), we need to be more careful in the creation and the deletion of buffer bins. To maintain the approximation guarantee, we have to make sure that as long as there are bins containing only small items, the remaining free space of all bins can be bounded. Packing small items into empty bins and leaving bins with large items untouched does not lead to a good approximation guarantee as the free space of the bins containing only large items is not used. In this section we consider the case where a sequence of small items is inserted or deleted. We assume that the packing of large items does not change. Therefore the number of bins containing large items equals a fixed constant $\Lambda(B)$. In the previous section, the bins $b_1,\ldots, b_{m(B)}$ all had a capacity of $1$. In order to handle a mixed setting, we will treat a bin $b_i$ containing large items as having capacity of $c(b_i) = 1-S$, where $S$ is the total size of the large items in $b_i$. The bins containing small items are enumerated by $b_1, \ldots ,b_{L(B)}, b_{L(B)+1}, \ldots , b_{m(B)}$ for some $L(B)\leq m(B)$ where $c(b_1),\ldots,c(b_{L(B)})< 1$ and $c(b_{L(B)+1})=\ldots = c(b_{m(B)})=1$. Additionally we have a separate set of bins, called the \emph{heap bins}, which contain only large items. This set of bins is enumerated by $h_1, \ldots h_{h(B)}$. Note that $L(B)+h(B)=\Lambda(B)$. In general we may consider only bins $b_i$ and $h_i$ with capacity $c(b_i) \geq \nicefrac{\epsilon}{14}$ and $c(h_i) \geq \nicefrac{\epsilon}{14}$ since bins with less capacity are already packed well enough for our approximation guarantee as shown by Lemma \ref{onlysmallitems}. Therefore, full bins are not considered in the following.

\tikzset{   brace/.style={
     decoration={brace, mirror},
     decorate
   },
      position label/.style={
       below = 3pt,
       text height = 1.5ex,
       text depth = 1ex
    }
}

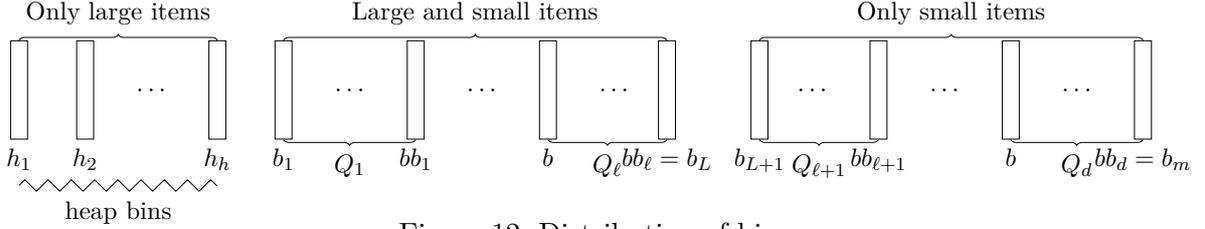
\begin{figure}[ht]
\centering
\resizebox{\textwidth}{!}{
\begin{tikzpicture}
\bin{h1}{0}{$h_1$}
\bin{h2}{1cm}{$h_2$}
\node[xshift=2cm] {$\ldots$};
\bin{h3}{3cm}{$h_{h}$};
\draw [decorate, decoration={brace}] (h1.north) -- (h3.north) node [above=3pt, pos=0.5] {Only large items};
\draw [decorate, decoration={zigzag},transform canvas={yshift=-0.7cm}] (h1.south) -- (h3.south) node [below=3pt, pos=0.5] {heap bins};

\bin{1}{4cm}{$b_1$}
\node[xshift=5cm] (ld1) {$\ldots$};
\bin{2}{6cm}{$bb_1$}
\draw [brace] (1.south) -- (2.south) node [position label, pos=0.5] {$Q_1$};
\node[xshift=7cm] (ld3) {$\ldots$};
\bin{3}{8cm}{$b$}
\node[xshift=9cm] (ld2) {$\ldots$};
\bin{4}{9.8cm}{$bb_{\ell}=b_L$}
\draw [brace] (3.south) -- (4.south) node [position label, pos=0.5] {$Q_{\ell}$};
\draw [decorate, decoration={brace}] (1.north) -- (4.north) node [above=3pt, pos=0.5] {Large and small items};

\bin{s1}{11.2cm}{$b_{L+1}$}
\node[xshift=12cm] (ld1) {$\ldots$};
\bin{s2}{13cm}{$bb_{\ell+1}$}
\draw [brace] (s1.south) -- (s2.south) node [position label, pos=0.5] {$Q_{\ell+1}$};
\node[xshift=14cm] (ld3) {$\ldots$};
\bin{s3}{15cm}{$b$}
\node[xshift=16cm] (ld2) {$\ldots$};
\bin{s4}{17cm}{$bb_{d}=b_m$}
\draw [brace] (s3.south) -- (s4.south) node [position label, pos=0.5] {$Q_d$};
\draw [decorate, decoration={brace}] (s1.north) -- (s4.north) node [above=3pt, pos=0.5] {Only small items};

\end{tikzpicture}}

\caption{Distribution of bins}
\label{figmixedsetting}

\end{figure}

As before, we partition the bins $b_1,\ldots,b_{L(B)}, b_{L(B)+1}, \ldots ,b_{m(B)}$ into several different queues $Q_1, \ldots , Q_{\ell(B)} , Q_{\ell(B) +1}, \ldots ,Q_{d(B)}$ such that $b_1, \ldots b_{L(B)} = Q_1, \ldots Q_{\ell(B)}$ and $b_{L(B)+1}, \ldots b_{m(B)} = Q_{\ell(B)+1}, \ldots , Q_{d(B)}$. If the corresponding packing $B$ is clear from the context, we will simply write $h,L,\ell,d,m,\Lambda$ instead of $h(B),L(B),\ell(B),d(B),m(B),\Lambda(B)$.
We denote the last bin of queue $Q_i$ by $bb_i$ which is a buffer bin. The buffer bin $bb_{\ell}$ is special and will be treated differently in the insert and delete operation.
Note that the bins containing large items $b_1,\ldots,b_{L(B)}$ are enumerated first. This guarantees that the free space in the bins containing large items is used before new empty bins are opened to pack the small items. However, enumerating bins containing large items first, leads to a problem if according to Algorithm \ref{alg-bb} when a buffer bin is being filled and a new bin has to be inserted right to the filled bin. Instead of inserting a new empty bin, we insert a heap bin at this position. Since the heap bin contains only large items, we do not violate the order of the small items (see Figure 
\ref{figmixedsetting}). As the inserted heap bin has remaining free space (is not filled completely) for small items, it can be used as a buffer bin.
In order to get an idea of how many heap bins we have to reserve for Algorithm \ref{alg-bb} where new bins are inserted or deleted, we define a potential function. As a buffer bin is being filled or emptied completely the Algorithm \ref{alg-bb} is executed and inserts or deletes buffer bins. The potential function $\Phi(B)$ thus bounds the number of buffer bins in $Q_1, \ldots , Q_{\ell(B)}$ that are about to get filled or emptied.
The potential $\Phi(B)$ is defined by 
\begin{align*}
\Phi(B) =  \sum_{i=1}^{\ell-1} r_i + \lceil \epsilon \Lambda \rceil - \ell
\end{align*}
where the \emph{fill ratio} $r_i$ is defined by $r_i=\frac{s(bb_i)}{c(bb_i)}$ and $s(bb_i)$ is the total size of all small items in $bb_i$ . Note that the potential only depends on the queues $Q_1, \ldots , Q_{\ell(B)}$ and the bins which contain small and large items. The term $r_i$ intends to measure the number of buffer bins that become full. According to Case 1 of the previous section a new buffer bin is opened when $bb_i$ is filled i.\,e., $r_i \approx 1$. Hence the sum $\sum_{i=1}^{\ell-1} r_i$ bounds the number of buffer bins getting filled. The term $\epsilon \Lambda$ in the potential measures the number of bins that need to be inserted due to the length of a queue exceeding $\nicefrac{2}{\epsilon}$, as we need to split the queue $Q_i$ into two queues of length $\nicefrac{1}{\epsilon}$ according to Case 1. Each of those queues needs a buffer bin, hence we need to insert a new buffer bin out of the heap bins. Therefore the potential $\Phi(B)$ bounds the number of bins which will be inserted as new buffer bins according to Case 1.

Just like in the previous section we propose the following properties to bound the approximation ratio and the migration factor. The first three properties remain the same as in Section \ref{sec:smallitems} and the last property gives the desired connection between the potential function and the heap bins. 
\begin{enumerate}
\item[(1)] For every item $i\in b_d$ with size $s(i)\in S_j $ for some $j,d \in \mathbb{N}$, there is no item $i' \in b_{d'}$ with size $s(i') \in s_{j'}$ such that $d'>d$ and $j' > j$. This means: Items are ordered from left to right by their size intervals.
\item[(2)] Every normal bin of $b_1,\ldots,b_m$ is filled completely
\item[(3)] The length of each queue is at least $\nicefrac{1}{\epsilon}$ and at most $\nicefrac{2}{\epsilon}$ except for $Q_{\ell}$ and $Q_{d}$. The length of $Q_{\ell}$ and $Q_{d}$ is only limited by $1\leq |Q_{\ell}|,|Q_d|\leq \nicefrac{1}{\epsilon}$. Furthermore, $|Q_{\ell+1}| = 1$ and $1 \leq |Q_{\ell+2}| \leq \nicefrac{2}{\epsilon}$.
\item[(4)] The number of heap bins $H_1, \ldots , H_{h}$ is exactly $h = \lfloor \Phi(B) \rfloor$
\end{enumerate}
Since bins containing large items are enumerated first, property $(1)$ implies in this setting that bins with large items are filled before bins that contain no large items. Note also that property (3) implies that $\Phi(B) \geq 0$ for arbitrary packings $B$ since $\epsilon\Lambda \geq \ell-1+\epsilon$ and thus $\lceil \epsilon \Lambda \rceil \geq \ell$. The following lemma proves that a packing which fulfills properties $(1)$ to $(4)$ provides a solution that is close to the optimum.
\begin{lemma}
\label{lem-smallitems-approximation}
    Let $M = m + h$ be the number of used bins and $\epsilon \leq \nicefrac{1}{4}$.
    If properties $(1)$ to $(4)$ hold, then at most $\max \{ \Lambda , (1+\mathcal{O}(\epsilon))\OPT(I,s)+\mathcal{O}(1) \}$ bins are used in the packing.
\end{lemma}
\begin{proof}
 Case 1: There is no bin containing only small items, i.\,e., $L=m$. Hence all items are packed into $M=L+h = \Lambda$ bins.
 
 Case 2: There are bins containing only small items, i.\,e., $L<m$. Property (3) implies that the number of queues $d$ is bounded by $d\leq \epsilon m+4$. Hence the number of buffer bins is bounded by $\epsilon m+4$ and the number of heap bins $\Phi(B)$ (property (4)) is bounded by $\Phi(B) = \sum_{i=1}^{\ell-1} r_i + \lceil \epsilon \Lambda \rceil - \ell \leq \ell -1 + \epsilon \Lambda +1 - \ell = \epsilon \Lambda$ as $r_i\leq 1$. Since $\Lambda < M$, we can bound $\Phi (B)$ by $\Phi(B) < \epsilon M$. 
 The number of normal bins is thus at least $M - (\epsilon m + 5) - (\epsilon M - 1) \geq M - 2 \epsilon M - 4 = (1-2\epsilon)M- 4$. By property (2) every normal bin has less than $\nicefrac{\epsilon}{14}$ free space and the total size $S$ of all items is thus at least $S \geq (1-\nicefrac{\epsilon}{14})(1-2\epsilon)M-4$. Since $\OPT(I,s)\geq S$, we have $\OPT(I,s)\geq (1-\nicefrac{\epsilon}{14}(1-2\epsilon)M-4$. 
 A simple calculation shows that $\frac{1}{(1-\nicefrac{\epsilon}{14})(1-2\epsilon)}\leq (1+5\epsilon)$ for $\epsilon\leq \nicefrac{1}{4}$.
 Therefore we can bound the number of used bins by $(1+5\epsilon)\OPT(I,s)+4$. 
\end{proof}

According to property (4) we have to guarantee, that if the rounded potential $\lfloor \Phi(B) \rfloor$ changes, the number of heap bins has to be adjusted accordingly. The potential $\lfloor \Phi(B) \rfloor$ might increases by $1$ due to an insert operation. Therefore the number of heap bins has to be incremented. If the potential $\lfloor \Phi(B) \rfloor$ decreases due to a delete operation, the number of heap bins has to be decremented. In order to maintain property $(4)$ we have to make sure, that the number of heap bins can be adjusted whenever $\lfloor \Phi(B) \rfloor$ changes. Therefore we define the fractional part $\{ \Phi(B) \}=\Phi(B)-\lfloor\Phi(B)\rfloor$ of $\Phi(B)$ and put it in relation to the fill ratio $r_{\ell}$ of $bb_{\ell}$ (the last bin containing large items) through the following equation:
\begin{align*}
\tag{Heap Equation}
    | (1 - r_{\ell}) - \{ \Phi(B) \} | \leq \frac{s}{c(bb_\ell)}
\end{align*}
where $s$ is the biggest size of a small item appearing in $bb_{\ell}$. The Heap Equation ensures that the potential $\Phi(B)$ is correlated to $1-r_{\ell}$. The values may only differ by the small term $\frac{s}{c(bb_{\ell})}$. Note that the Heap Equation can always be fulfilled by shifting items from $bb_{\ell}$ to queue $Q_{\ell+1}$ or vice versa.

Assuming the Heap Equation holds and the potential $\lfloor \Phi(B)\rfloor$ increases by $1$, we can guarantee that buffer bin $bb_{\ell}$ is nearly empty. Hence the remaining items can be shifted to $Q_{\ell+1}$ and $bb_\ell$ can be moved to the heap bins. The bin left of $bb_\ell$ becomes the new buffer bin of $Q_\ell$. Vice versa, if $\lfloor \Phi(B)\rfloor$ decreases, we know by the Heap Equation that $bb_{\ell}$ is nearly full, hence we can label $bb_{\ell}$ as a normal bin and open a new buffer bin from the heap at the end of queue $Q_{\ell}$.
Our goal is to ensure that the Heap Equation is fulfilled at every step of the algorithm along with properties $(1)$ to $(4)$. Therefore we enhance the delete and insert operations from the previous section. 
Whenever a small item $i$ is inserted or removed, we will perform the operations described in Algorithm \ref{alg-insertsmall} (which can be applied to bins of different capacities) in the previous section. This will maintain properties $(1)$ to $(3)$. If items are inserted or deleted from queue $Q_{\ell}$ (the last queue containing large and small items) the recursion does not halt at $bb_{\ell}$. Instead the recursion goes further and halts at $bb_{\ell+1}$. So, when items are inserted into bin $bb_\ell$ according to Algorithm \ref{alg-insertsmall} the bin $bb_\ell$ is treated as a normal bin. Items are shifted from $bb_{\ell}$ to queue $Q_{\ell+1}$ until the Heap Equation is fulfilled. This way we can make sure that the Heap Equation maintains fulfilled whenever an item is inserted or removed from $Q_{\ell}$. 
 
\begin{algo}[Insert or Delete small items for the mixed setting]
\label{algo-small-items-general}
\ 

{\bf Insert$(i,b_x,Q_j)$:} 
\begin{compactitem}
\item Use Algorithm \ref{alg-insertsmall} to insert item $i$ into $Q_j$ with $j < \ell$.
\item Let $i_1, \ldots , i_m$ be the items that are inserted at the last step of Algorithm \ref{alg-insertsmall} into $bb_j$.
\item For $k = 1, \ldots , m$ do
    \begin{compactenum}
    \item Insert item $i_k$ into bin $bb_j$.
    \item If $bb_j$ is completely filled use Algorithm \ref{alg-bb}.
    \item If the potential $\lfloor \Phi(B) \rfloor$ increases use Algorithm \ref{alg-potential} (see below) to adjust the number of heap bins (property (4)).
    \item Decrease the fill ratio $r_\ell$ of $bb_\ell$ by shifting the smallest items in $bb_\ell$ to $Q_{\ell+1}$ until $(1 - r_{\ell}) \leq \{ \Phi(B) \}$ to fulfill the Heap Equation.
    \end{compactenum}
\end{compactitem}


{\bf Delete$(i,b_x,Q_j)$:} 
\begin{compactitem}
\item Use Algorithm \ref{alg-insertsmall} to remove item $i$ from bin $b_x$ in queue $Q_j$ with $j < \ell$.
\item Let $i_1, \ldots , i_m$ be the items that are removed at the last step of Algorithm \ref{alg-insertsmall} from $bb_j$.
\item For $k = 1, \ldots , m$ do
    \begin{compactenum}
    \item If $bb_j$ is empty use Algorithm \ref{alg-bb}.
    \item Remove item $i_k$ from bin $bb_j$.
    \item If the potential $\lfloor \Phi(B) \rfloor$ decreases use Algorithm \ref{alg-potential}.
    \item Increase the fill ratio $r_\ell$ of $bb_\ell$ by shifting the smallest items in $bb_\ell$ to $Q_{\ell+1}$ until $(1 - r_{\ell}) \geq \{ \Phi(B) \}$ to fulfill the Heap Equation.
    \end{compactenum}
\end{compactitem}
\end{algo}
For the correctness of step 4 (the adjustment to $r_{\ell}$) note the following:
In case of the insert operation, the potential $\Phi(B)$ increases and we have $\Phi(B) \geq 1-r_\ell$. As items are being shifted from $bb_\ell$ to $Q_{\ell+1}$, the first time that $(1 - r_{\ell}) \leq \{ \Phi(B) \}$ is fulfilled, the Heap Equation is also fulfilled. Since the fill ratio of $bb_\ell$ changes at most by $\frac{s}{c(bb_\ell)}$ as an item (which has size at most $s$) is shifted to $Q_{\ell+1}$ we know that $| (1 - r_{\ell}) - \{ \Phi(B) \} | \leq \frac{s}{c(bb_\ell)}$. Correctness of step 4 in the delete operation follows symmetrically.

The potential $\Phi(B)$ changes if items are inserted or deleted into queues $Q_1, \ldots , Q_{\ell-1}$.
Due to these insert or delete operations it might happen that the potential $\lfloor \Phi(B) \rfloor$ increases or that a buffer bin is being filled or emptied. The following operation is applied as soon as an item is inserted or deleted into a buffer bin and the potential $\lfloor \Phi(B) \rfloor$ increases or decreases.

\begin{algo}[Change in the potential]
\label{alg-potential}
\ 
\begin{compactitem}
\item {\bf Case 1: The potential $\lfloor \Phi(B) \rfloor$ increases by $1$.} 
\begin{compactitem}
\item According to the Heap Equation the remaining size of small items in $bb_{\ell}$ can be bounded. Shift all small items from $bb_{\ell}$ to $Q_{\ell+1}$. 
\item If $|Q_{\ell}| >1$ then label the now empty buffer bin $bb_{\ell}$ as a heap bin and the last bin in $Q_{\ell}$ is labeled as a buffer bin.
\item If $Q_{\ell}$ only consists of the buffer bin (i.\,e., $|Q_{\ell}| = 1$)
shift items from $bb_{\ell-1}$ to $Q_{\ell+1}$ until the heap equation is fulfilled.
If $bb_{\ell-1}$ becomes empty remove $bb_{\ell-1}$ and $bb_{\ell}$. The bin left to $bb_{\ell-1}$ becomes the new buffer bin of $Q_{\ell-1}$. The queue $Q_{\ell}$ is deleted and $Q_{\ell-1}$ becomes the new last queue containing large items.
\end{compactitem}
\item {\bf Case 2: The potential $\lfloor \Phi(B) \rfloor$ decreases by $1$.}
\begin{compactitem}
\item According to the Heap Equation the remaining free space in $bb_{\ell}$ can be bounded. Shift items from $bb_{\ell+1}$ to $bb_{\ell}$ such that the buffer bin $bb_{\ell}$ is filled completely. 
\item Add the new buffer bin from the heap to $Q_{\ell}$. 
\item If $|Q_{\ell}| = \nicefrac{1}{\epsilon}$ label an additional heap bin as a buffer bin to create a new queue $Q_{\ell+1}$ with $|Q_{\ell+1}| = 1$.
\end{compactitem}
\end{compactitem}
\end{algo}

Like in the last section we also have to describe how to handle buffer bins that are being emptied or filled completely. We apply the same algorithm when a buffer bin is being emptied or filled but have to distinguish now between buffer bins of $Q_1, \ldots , Q_{\ell}$ and buffer bins of $Q_{\ell+1}, \ldots, Q_{d}$. 
Since the buffer bins in $Q_{\ell+1},\ldots,Q_{d}$ all have capacity $1$, we will use the same technique as in the last section. If a buffer bin in $Q_{1},\ldots,Q_{\ell}$ is emptied or filled we will also use similar technique. But instead of inserting a new empty bin as a new buffer bin, we take an existing bin out of the heap. And if a buffer bin from $Q_1, \ldots Q_{\ell}$ is being emptied (it still contains large items), it is put into the heap. This way we make sure that there are always sufficiently many bins containing large items which are filled completely.

\begin{lemma} \label{lem-algorithm}
    Let $B$ be an packing which fulfills the properties $(1)$ to $(4)$ and the Heap Equation. Applying Algorithm \ref{alg-potential} or Algorithm \ref{alg-bb} on $B$ during an insert/delete operation yields an packing $B'$ which also fulfills properties $(1)$ to $(4)$. The migration to fulfill the Heap Equation is bounded by $\mathcal{O}(\nicefrac{1}{\epsilon})$.
\end{lemma}

\begin{proof}
{\bf Analysis of Algorithm \ref{alg-potential}}\\
Properties $(1)$ and $(2)$ are never violated by the algorithm because the items are only moved by shift operations. Property $(3)$ is never violated because no queue (except for $Q_{\ell}$) exceeds $\nicefrac{2}{\epsilon}$ or falls below $\nicefrac{1}{\epsilon}$ by construction.
Algorithm \ref{alg-potential} is called during an insert or delete operation. The Algorithm is executed as items are shifted into or out of buffer $bb_j$ such that $\lfloor \Phi(B)\rfloor$ changes.

In the following we prove property $(4)$ for the packing $B'$ assuming that $\lfloor \Phi(B) \rfloor = h(B)$  holds by induction. Furthermore we give a bound for the migration to fulfill the heap equation:
\begin{itemize}
\item Case 1: The potential $\lfloor \Phi(B) \rfloor$ increases during an insert operation, i.\,e., it holds $\lfloor \Phi(B') \rfloor = \lfloor \Phi(B) \rfloor +1$. Let item $i^*$ be the first item that is shifted into a bin $bb_j$ such that $\lfloor \Phi(B) + r^* \rfloor = \lfloor \Phi(B') \rfloor$, where $r^*$ is the fill ratio being added to $bb_j$ by item $i^*$. In this situation, the fractional part changes from $\{\Phi(B)\} \approx 1$ to $\{\Phi(B')\} \approx 0$. 

\begin{itemize}
\item In the case that $|Q_{\ell}|> 1$, the buffer bin $bb_{\ell}$ is being emptied and moved to the heap bins. The bin left of $bb_{\ell}$ becomes the new buffer bin $bb'_\ell$ of $Q_{\ell}$. Hence the number of heap bins increases and we have $h(B') = h(B) +1 = \lfloor \Phi(B) \rfloor +1 =  \lfloor \Phi(B') \rfloor$, which implies property (4).

To give a bound on the total size of items needed to be shifted out of (or into) bin $bb_\ell$ to fulfill the heap equation, we bound the term $|(1 - r'_{\ell}) - \{ \Phi(B') \}|$ by some term $C\leq \mathcal{O}(\nicefrac{s(i)}{\epsilon})$, where $r'_\ell$ is the fill ratio of $bb'_\ell$ and $s(i)$ is the size of the arriving or departing item. If the term  $|(1 - r'_{\ell}) - \{ \Phi(B') \}|$ can be bounded by $C$, the fill ratio of $bb'_\ell$ has to be adjusted to fulfill the heap equation according to the insert and delete operation. This can be done be shifting a total size of at most $C$ items out of (or into) $bb'_\ell$.

The bin $bb'_\ell$ is completely filled by property (3) and therefore has a fill ratio of $r'_\ell \geq \frac{c(bb_\ell) - s}{c(bb_\ell)} \geq 1- 2 \frac{s}{\epsilon}$, where $s \leq \frac{\epsilon}{2^k}$ is the largest size of a small item appearing in $bb_\ell$ and $S_k$ is the largest size category appearing in $bb'_\ell$. Let $k'$ be the largest size category appearing in bin $bb_j$. As the bin $bb'_\ell$ is right of $bb_j$ we know $k \leq k'$ (property $(1)$) and hence $s \leq 2 s(i^*)$. We get $r'_{\ell} \geq 1 - 4 \frac{s(i^*)}{\epsilon}$.
Using that $\{ \Phi(B') \} \leq r^* \leq 2 \nicefrac{s(i^*)}{\epsilon}$, we can bound $|(1 - r'_{\ell}) - \{ \Phi(B') \}|$ by $4 \frac{s(i^*)}{\epsilon} + 2 \nicefrac{s(i^*)}{\epsilon} = \mathcal{O}(\nicefrac{s(i^*)}{\epsilon})$. Hence the Heap Equation can be fulfilled by shifting items of total size $\mathcal{O}(\nicefrac{s(i^*)}{\epsilon})$ at the end of the insert operation. 

\item If $|Q_{\ell}| = 1$ a set of items in the buffer bin $bb_{\ell-1}$ is shifted to $Q_{\ell+1}$ to fulfill the Heap Equation. Since items are being removed from $bb_{\ell-1}$ the potential decreases. If $r_{\ell-1}  > \{ \Phi(B') \}$, there are enough items which can be shifted out of $bb_{\ell-1}$ such that we obtain a new potential $\Phi(B'') < \Phi(B') - \{\Phi(B')\}$. Hence $\lfloor \Phi(B'')\rfloor = \lfloor \Phi(B)\rfloor$ and the Heap Equation is fulfilled.

Note that the size of items that are shifted out of $bb_{\ell-1}$ is bounded by $r^*+s = \mathcal{O}(\nicefrac{s(i^*)}{\epsilon})$, where $s$ is the biggest size of an item appearing in $bb_{\ell-1}$.

If $r_{\ell-1}\leq \{ \Phi(B') \}$ all items are shifted out of $bb_{\ell-1}$. As the number of queues decreases, we obtain the new potential $\Phi(B'') = \Phi(B') - r_{\ell -1} + 1= \lfloor \Phi(B') \rfloor +\{\Phi(B')\}-r_{\ell-1} + 1 \geq \lfloor \Phi(B')\rfloor +1 $. Hence $\lfloor\Phi(B'')\rfloor = \lfloor \Phi(B)\rfloor +2$. The buffer bins $bb_{\ell-1}$ and $bb_{\ell}$ are moved to the heap and thus $h(B'')=h(B)+2=\lfloor \Phi(B)\rfloor + 2=\lfloor \Phi(B'')\rfloor$ (property (4)).

Note that if $r_{\ell-1} \leq \{ \Phi(B') \}$, item $i^*$ is not inserted into bin $bb_{\ell-1}$ as $r_{\ell-1} \geq r^* > \{ \Phi(B') \}$. Therefore the bin $bb_j$ is left of $bb_{\ell-1}$ and we can bound the fill ratio of the bin left of $bb_{\ell-1}$ called $r''_{\ell}$ by  $1 - 2 \frac{s(i^*)}{\epsilon}$. Using $\{\Phi(B'') \} \leq r^* = \mathcal{O}(\nicefrac{s(i^*)}{\epsilon})$ the heap equation can be fulfilled by shifting items of total size $\mathcal{O}(\nicefrac{s(i)}{\epsilon})$ at the end of the insert operation.
\end{itemize}

\item Case 2: The potential $\lfloor \Phi(B) \rfloor$ decreases during a delete operation, i.\,e., it holds $\lfloor \Phi(B') \rfloor = \lfloor \Phi(B) \rfloor -1$ = $\lfloor \Phi(B) - r^* \rfloor$, where $r^*$ is the fill ratio being removed from a buffer bin $bb_j$ due to the first shift of an item $i^*$ that decreases the potential.\\
According to Algorithm \ref{alg-potential}, buffer bin $bb_{\ell}$ is being filled completely and a new buffer bin for $Q_{\ell}$ is inserted from the heap. Hence the number of heap bins decreases and we have $\lfloor \Phi(B') \rfloor = h(B) - 1 = h(B')$.

As $\lfloor \Phi(B)\rfloor -1=\Phi(B)-\{\Phi(B)\}-1=\lfloor \Phi(B)-r^*\rfloor$, it holds that $\{\Phi(B)\}\leq r^*$ and by
the heap equation the fill ratio of $bb_\ell$ is $r_\ell \geq r^*+s$, where $s$ is the largest size of a small item in $bb_{\ell}$. As above, $r^*$ and $s$ can be bounded by $\mathcal{O}(\frac{s(i^*)}{\epsilon})$. Hence the total size that is shifted from $Q_{\ell+1}$ into bin $bb_\ell$ can be bounded by $\mathcal{O}(\frac{s(i^*)}{\epsilon})$.

Furthermore $\{ \Phi(B') \} \geq 1 - r^*$ (as $\Phi(B')=\Phi(B)-r^*$) and $r'_\ell = 0$, therefore we can bound $|(1 - r'_{\ell}) - \{ \Phi(B') \}|$ by $r^*\leq \mathcal{O}(\nicefrac{s(i^*)}{\epsilon})$ and the Heap Equation can be fulfilled by shifting a total size of at most $\mathcal{O}(\nicefrac{s(i^*)}{\epsilon})$ items.

In the case that $|Q_{\ell}|=\nicefrac{1}{\epsilon}$ a new queue $Q_{\ell+1}$ is created which consists of a single buffer bin (inserted from the heap), which does not contain small items, i.\,e., $h(B'') = h(B') -1 = h(B) -2$, where $B''$ is the packing after the insertion of item $i^*$. Let $\Phi(B'')$ be the potential after the queue $Q_{\ell+1}$ is created.
Then $\Phi(B'')= \sum_{i=1}^{\ell(B'')-1}r_i +\epsilon\Lambda -\ell(B'') = \sum_{i=1}^{\ell(B')-2}r_i+\epsilon\Lambda-\ell(B')-1=\Phi(B')-1$, as the buffer bin $bb_{\ell}$ is now counted in the potential, but does not contain any small items and thus $r''_\ell=0$. Hence $\Phi(B'') = \Phi(B') -1 = h(B') -1 = h(B'')$.
\end{itemize}

{\bf Analysis of Algorithm \ref{alg-bb}}\\
Algorithm \ref{alg-bb} is executed as an item $i^*$ is moved into a buffer bin $bb_{j}$ such that $bb_j$ is completely filled or Algorithm \ref{alg-bb} is executed if the buffer bin $bb_j$ is emptied by moving the last item $i^*$ out of the bin. As in the analysis of Algorithm \ref{alg-potential}, properties $(1)$ and $(2)$ are never violated by the algorithm because the items are only moved by shift operations. Property $(3)$ is never violated because no queue (except for $Q_{\ell}$) exceeds $\nicefrac{2}{\epsilon}$ or falls below $\nicefrac{1}{\epsilon}$ by construction. 

It remains to prove property $(4)$ and a bound for the migration to fulfill the heap equation:
\begin{itemize}
\item Case 1: An item $i^*$ is moved into the buffer bin $bb_j$ such that $bb_j$ is filled completely for some $j<\ell$. According to Algorithm \ref{alg-bb} a bin is taken out of the heap and labeled as the new buffer bin $bb'_j$ with fill ratio $r'_j= 0$ of queue $Q_j$, i.\,e., the number of heap bins decreases by $1$. Let $\Phi(B)$ be the potential before Algorithm \ref{alg-bb} is executed and let $\Phi(B')$ be the potential after Algorithm \ref{alg-bb} is executed. The potential changes as follows:
\begin{align*}
\Phi(B)-\Phi(B')= (r_j - r'_j) - (\ell(B) - \ell(B'))
\end{align*}
Since $r'_j = 0$ the new potential is $\Phi(B') = \Phi(B) - r_j \approx \Phi(B) -1$ (assuming $\ell(B) = \ell(B')$, as the splitting of queue is handled later on).
\begin{itemize}
\item If $\lfloor \Phi(B') \rfloor = \lfloor \Phi(B) \rfloor - 1$ property (4) is fulfilled since the number of heap bins decreases by $h(B') = h(B) - 1 = \lfloor \Phi(B) \rfloor  -1 = \lfloor \Phi(B') \rfloor$. As $r_j \geq \frac{c(bb_j)-s}{c(bb_j)}$, where $s$ is the biggest size category appearing in $bb_j$ and $s \leq 2 s(i^*)$, we obtain for the fractional part of the potential that $\{ \Phi(B) \} - \{ \Phi(B') \} \leq 2 \frac{s}{\epsilon} \leq 4  \frac{s(i^*)}{\epsilon}$. Hence the Heap Equation can be fulfilled by shifting items of total size $\mathcal{O}(\nicefrac{s(i^*)}{\epsilon})$ at the end of the insert operation as in the above proof.

\item In the case that $\lfloor \Phi(B') \rfloor = \lfloor \Phi(B) \rfloor = \lfloor \Phi(B) - r_j \rfloor$ we know that the fractional part changes by $\{ \Phi(B') \} = \{ \Phi(B) \} - r_j$. Since the bin $bb_j$ is filled completely we know that $r_j \geq \frac{c(bb_j)-s}{c(bb_j)} \approx 1$ and hence $\{ \Phi(B) \} \geq r_j \approx 1$ and $\{ \Phi(B') \} \leq 1-r_j \approx 0$. According to the Heap Equation, items have to be shifted out of $r_{\ell}$ such that the fill ratio $r_{\ell}$ changes from $r_{\ell} \leq 1-r_j$ to $r_{\ell} \approx 1$. Therefore we know that as items are shifted out of $bb_{\ell}$ to fulfill the Heap Equation, the buffer bin $bb_{\ell}$ is being emptied and moved to the heap (see Algorithm \ref{alg-potential}). We obtain for the number of heap bins that $h(B') = h(B) +1 -1 = h(B)$ and hence $h(B') = \lfloor \Phi(B') \rfloor$ (property (4)).

As $\{ \Phi(B) \} \geq r_j \geq 1- 4 \frac{s(i^*)}{\epsilon}$, the Heap Equation implies that $r_\ell \leq 4 \frac{s(i^*)}{\epsilon} + \frac{s}{c(bb_\ell)} = \mathcal{O}(\nicefrac{s(i^*)}{\epsilon})$.
The buffer bin $bb_\ell$ is thus emptied by moving a size of $\mathcal{O}(\nicefrac{s(i^*)}{\epsilon})$ items out of the bin. Let $bb'_\ell$ be the new buffer bin of $Q_\ell$ that was left of $bb_\ell$. The Heap Equation can be fulfilled by shifting at most $\mathcal{O}(\nicefrac{s(i)}{\epsilon})$ out of $bb'_\ell$ since $\{ \Phi(B') \}$ is bounded by $1 - r_j = \mathcal{O}(\nicefrac{s(i^*)}{\epsilon})$.

\item In the case that $|Q_j| > \nicefrac{2}{\epsilon}$ the queue is split into two queues and an additional heap bin is inserted, i.\,e., $h(B'') = h(B') -1$. As the potential changes by $\Phi(B'') = \Phi(B') + (\ell(B') - \ell(B'')) = \Phi(B') - 1$ we obtain again that $h(B'') = \lfloor \Phi(B'')\rfloor $.
\end{itemize}

\item Case 2: Algorithm \ref{alg-bb} is executed if bin $bb_j$ is emptied due to the removal of an item $i^*$ as a result of a Delete$(i,b_x,Q_j)$ call. According to Algorithm \ref{alg-bb}, the emptied bin is moved to the heap, i.\,e., the number of heap bins increases by $1$. Depending on the length of $Q_j$ and $Q_{j+1}$, the bin right of $bb_j$ or the bin left of $bb_{j}$ is chosen as the new buffer bin $bb'_j$. The potential changes by $\Phi(B') = \Phi(B) + r'_j$, where $r'_j$ is the fill ratio of $bb'_j$ as in case 1.
\begin{itemize}
\item If $\lfloor \Phi(B') \rfloor = \lfloor \Phi(B) \rfloor + 1$ property (4) is fulfilled since the number of heap bins increases by $h(B') = h(B) + 1$. 

As bin $bb'_j$ is completely filled, the fill ratio is bounded by $r'_j \geq 1-2 \frac{s}{\epsilon}$, where $s$ is the largest size appearing in $bb'_j$. Since the bin $b_x$ has to be left of $bb_j$ we know that $s \leq 2s(i)$. We obtain for the fractional part of the potential that $\{ \Phi(B) \}  \geq \{ \Phi(B') \} - 2 \frac{s}{\epsilon} \leq 4 \frac{s(i)}{\epsilon}$. Hence the Heap Equation can be fulfilled by shifting items of total size $\mathcal{O}(\nicefrac{s(i)}{\epsilon})$ at the end of the remove operation.

\item In the case that $\lfloor \Phi(B') \rfloor = \lfloor \Phi(B) \rfloor = \lfloor \Phi(B) + r'_j \rfloor$ we know that the fractional part changes similar to case 1 by $\{ \Phi(B') \} = \{ \Phi(B) \} + r'_j$. Since the bin $bb_j$ is filled completely we know that $r_j \geq \frac{c(bb_j)-s}{c(bb_j)} \approx 1$ and hence $\{ \Phi(B') \} \geq r_j \approx 1$ and $\{ \Phi(B) \} \leq 1-r_j \approx 0$. According to the Heap Equation items have to be shifted to $bb_{\ell}$ such that the fill ratio $r_{\ell}$ changes from $r_{\ell} \approx 0$ to $r_{\ell} \approx 1$. Therefore we know that as items are shifted into $bb_{\ell}$ to fulfill the Heap Equation, $bb_{\ell}$ is filled completely and a bin from the heap is labeled as the new buffer bin of $Q_{\ell}$ (see Algorithm \ref{alg-potential}). We obtain for the number of heap bins that $h(B') = h(B) -1 +1 = h(B)$ and hence $h(B') = \Phi(B')$ (property (4)). The Heap Equation can be fulfilled similarly to case 1 by shifting items of total size $\mathcal{O}(\nicefrac{s(i)}{\epsilon})$.
\end{itemize}
\end{itemize}
\end{proof}

Using the above lemma for, we can finally prove the following central theorem, which states that the migration of an insert/delete operation is bounded and that properties $(1)$ to $(4)$ are maintained. 

\begin{theorem}
\label{thm-main-small}
\ 
    \begin{enumerate}
            \item[(i)] Let $B$ be a packing which fulfills properties $(1)$ to $(4)$ and the Heap Equation. Applying operations insert$(i,b_x,Q_j)$ or delete$(i,b_x,Q_j)$ on a packing $B$ yields an instance $B'$ which also fulfills properties $(1)$ to $(4)$ and the Heap Equation.
        \item[(ii)] The migration factor of an insert/delete operation is bounded by $\mathcal{O}(\nicefrac{1}{\epsilon})$.
    \end{enumerate}
\end{theorem}

\begin{proof}
Suppose a small item $i$ with size $s(i)$ is inserted or deleted from queue $Q_j$. The insert and delete operation basically consists of application of Algorithm \ref{alg-insertsmall} and iterated use of steps $(1)$ to $(3)$ where Algorithms \ref{alg-bb} and \ref{alg-potential} are used and items in $bb_\ell$ are moved to $Q_{\ell+1}$ and vice versa.
Let $B$ be the packing before the insert/delete operation and let $B'$ be the packing after the operation.

Proof for (i):
Now suppose by induction that property $(1)$ to $(4)$ and the Heap Equation is fulfilled for packing $B$. We prove that property $(4)$ and the Heap Equation maintain fulfilled after applying an insert or delete operation on $B$ resulting in the new packing $B'$.
Properties $(1)$ to $(3)$ hold by conclusion of Lemma \ref{onlysmallitems} and Lemma \ref{lem-algorithm}. Since the potential and the number of heap bins only change as a result of Algorithm \ref{alg-bb} or Algorithm \ref{alg-potential}, property (4) maintains fulfilled also.
By definition of step 4 in the insert operation, items are shifted from $bb_\ell$ to $Q_{\ell+1}$ until the Heap Equation is fulfilled. By definition of step 4 of the delete operation, the size of small items in $bb_{\ell}$ is adjusted such that the Heap Equation is fulfilled. Hence the Heap Equation is always fulfilled after application of 
Insert$(i,b_x,Q_j)$ or Delete$(i,b_x,Q_j)$.

Proof for (ii): According to Lemma \ref{onlysmallitems} the migration factor of the usual insert operation is bounded by $\mathcal{O}(\nicefrac{1}{\epsilon})$. By Lemma \ref{lem-algorithm} the migration in Algorithm \ref{alg-bb} and Algorithm \ref{alg-potential} is also bounded by $\mathcal{O}(\nicefrac{1}{\epsilon})$. It remains to bound the migration for step 4 in the insert/delete operation. Therefore we have to analyze the total size of items to be shifted out or into $bb_\ell$ in order to fulfill the Heap Equation.

Since the size of all items $i_1, \ldots, i_k$ that are inserted into $bb_j$ is bounded by $7 s(i)$ (see Lemma \ref{onlysmallitems}) and the capacity of $bb_j$ is at least $\nicefrac{\epsilon}{14}$ the potential $\Phi(B)$ changes by at most $\mathcal{O}(\nicefrac{s(i)}{\epsilon})$. By Lemma \ref{lem-algorithm} the size of items that needs to be shifted out or into $bb_\ell$ as a result of Algorithm \ref{alg-bb} or \ref{alg-potential} is also bounded by $\mathcal{O}(\nicefrac{s(i)}{\epsilon})$. Therefore the size of all items that need to be shifted out or into $bb_\ell$ in step (4) of the insert/delete operation is bounded by $\mathcal{O}(\nicefrac{s(i)}{\epsilon})$.

Shifting a size of $\mathcal{O}(\nicefrac{s(i)}{\epsilon})$ to $Q_{\ell+1}$ or vice versa leads to a migration factor of $\mathcal{O}(\nicefrac{1}{\epsilon^2})$ (Lemma \ref{onlysmallitems}). Fortunately we can modify the structure of queues $Q_{\ell+1}$ and $Q_{\ell+2}$ such that we obtain a smaller migration factor. Assuming that $Q_{\ell+1}$ consists of a single buffer bin, i.\,e., $|Q_{\ell+1}| = 1$ items can directly be shifted from $bb_\ell$ to $bb_{\ell+1}$ and therefore we obtain a migration factor of $\mathcal{O}(\nicefrac{1}{\epsilon})$. A structure with $|Q_{\ell+1}| = 1$ and $1 \leq |Q_{\ell+2}| \leq \nicefrac{2}{\epsilon}$ (see property (3)) can be maintained by changing Algorithm \ref{alg-bb} in the following way:

\begin{itemize}
\item If $bb_{\ell+1}$ is filled completely, move the filled bin to $Q_{\ell+2}$.
\begin{itemize}
    \item If $|Q_{\ell+2}| > \nicefrac{2}{\epsilon}$, split $Q_{\ell+2}$ into two queues.
\end{itemize}
\item If $bb_{\ell+1}$ is being emptied, remove the bin and label the first bin of $Q_{\ell+2}$ as $bb_{\ell+1}$.
\begin{itemize}
    \item If  $|Q_{\ell+2}| = 0$, remove $Q_{\ell+2}$.
\end{itemize}
\end{itemize}

\end{proof}

\subsection{Handling the General Setting}
\label{sec:final}
In the previous section we described how to handle small items in a mixed setting. It remains to describe how large items are handled in this mixed setting. Algorithm \ref{alg-afptas} describes how to handle large items only. However, in a mixed setting, where there are also small items, we have to make sure that properties $(1)$ to $(4)$ and the Heap Equation maintain fulfilled as a large item is inserted or deleted. Algorithm \ref{alg-afptas} changes the configuration of at most $\mathcal{O}(\nicefrac{1}{\epsilon}^2 \cdot \log \nicefrac{1}{\epsilon})$ bins (Theorem \ref{thm-main}). Therefore, the size of large items in a bin $b$ ($= 1-c(b)$) changes, as Algorithm \ref{alg-afptas} may increase or decrease the capacity of a bin. Changing the capacity of a bin may violate properties (2) to (4) and the Heap Equation. We describe an algorithm to change the packing of small items such that all properties and the Heap Equation are fulfilled again after Algorithm \ref{alg-afptas} was applied.

The following algorithm describes how the length of a queue $Q_{j}$ is adjusted if the length $|Q_{j}|$ falls below $\nicefrac{1}{\epsilon}$:
\begin{algo}[Adjust the queue length]
\ 
\begin{itemize}
\item Remove all small item $I_S$ from $bb_{j}$ and add $bb_j$ to the heap.
\item Merge $Q_j$ with $Q_{j+1}$. The merged queue is called $Q_j$.
\item If $|Q_{j}|> \nicefrac{2}{\epsilon}$ split queue $Q_j$ by adding a heap bin in the middle.
\item Insert items $I_S$ using Algorithm \ref{algo-small-items-general}.
\end{itemize}
\end{algo}

The following algorithm describes how the number of heap bins can be adjusted.
\begin{algo}[Adjust number of heap bins]
\label{algo-adjust-heap}
\ 
\begin{itemize}
\item Decreasing the number of heap bins by $1$.
    \begin{itemize}
    \item Shift small items from $Q_{\ell+1}$ to $bb_{\ell}$ until $bb_{\ell}$ is filled completely
    \item Label a heap bin as the new buffer bin of $Q_{\ell}$
    \end{itemize}
\item Increasing the number of heap bins by $1$.
    \begin{itemize}
    \item Shift all small items from $bb_{\ell}$ to $Q_{\ell+1}$
    \item Label $bb_{\ell}$ as a heap bin
    \item Label the bin left of $bb_{\ell}$ as new buffer bin of $Q_{\ell}$
    \end{itemize}
\end{itemize}
\end{algo}
Note that the Heap Equation can be fulfilled in the same way, by shifting items from $bb_{\ell}$ to $Q_{\ell+1}$ or vice versa.

Using these algorithms, we obtain our final algorithm for the fully dynamic $\BP$ problem.

\begin{algo}[\ac{afptas} for the mixed setting]
\label{algo-final-afptas}
\ 
\begin{itemize}
\item If $i$ is large do
    \begin{enumerate}
    \item Use Algorithm \ref{alg-afptas}.
    \item Remove all small items $I_S$ of bins $b$ with changed capacity.
    \item Adjust queue length.
    \item Adjust the number of heap bins.
    \item Adjust the Heap Equation.
    \item Insert all items $I_S$ using Algorithm \ref{algo-small-items-general}.
    \end{enumerate}
\item If $i$ is small use Algorithm \ref{algo-small-items-general}
\end{itemize}
\end{algo}

Combining all the results from the current and the previous section, we finally prove the central result that there is fully dynamic \ac{afptas} for the $\BP$ problem with polynomial migration.

\begin{theorem}
    Algorithm \ref{algo-final-afptas} is a fully dynamic \ac{afptas} for
    the $\BP$ problem,
    that achieves a migration factor of at most $\mathcal{O}(\nicefrac{1}{\epsilon}^4 \cdot \log \nicefrac{1}{\epsilon})$ by repacking items from at most $\mathcal{O}(\nicefrac{1}{\epsilon}^3 \cdot \log \nicefrac{1}{\epsilon})$ bins.
\end{theorem}

\begin{proof}
{\bf Approximation guarantee:} By definition of the algorithm, it generates at every timestep $t$ a packing $B_t$ of instance $I(t)$ such that properties $(1)$ to $(4)$ are fulfilled. According to Lemma \ref{lem-smallitems-approximation}, at most $\max\{\Lambda,(1+\mathcal{O}(\epsilon))\OPT(I(t),s)+\mathcal{O}(1)\}$ bins are used where $\Lambda$ is the number of bins containing large items. Since we use Algorithm \ref{alg-afptas} to pack the large items, Theorem \ref{thm-main} implies that $\Lambda\leq (1+\mathcal{O}(\epsilon))\OPT(I(t),s)+\mathcal{O}(\nicefrac{1}{\epsilon}\log \nicefrac{1}{\epsilon})$. Hence the number of used bins can be bounded in any case by $(1+\mathcal{O}(\epsilon))\OPT(I(t),s)+\mathcal{O}(\nicefrac{1}{\epsilon}\log \nicefrac{1}{\epsilon})$.

{\bf Migration Factor:} Note that the Algorithm uses Algorithm \ref{algo-small-items-general} or Algorithm \ref{alg-afptas} to insert and delete small or large items. The migration factor for Algorithm \ref{algo-small-items-general} is bounded by $\mathcal{O}(\nicefrac{1}{\epsilon})$ due to Theorem \ref{thm-main-small} while the migration factor for Algorithm \ref{alg-afptas} is bounded by $\mathcal{O}(\nicefrac{1}{\epsilon^3}\cdot \log \nicefrac{1}{\epsilon})$ due to Theorem \ref{thm-main}.

It remains to bound the migration that is needed to adjust the heap bins, the length of a queue falling below $\nicefrac{1}{\epsilon}$ and the Heap Equation in case a large item arrives and Algorithm \ref{alg-afptas} is applied. 

Suppose the number of heap bins has to be adjusted by $1$. In this case Algorithm \ref{algo-adjust-heap} shifts items from $Q_{\ell+1}$ to $bb_{\ell}$ or vice versa until $bb_{\ell}$ is either filled or emptied. Hence, the size of moved items is bounded by $1$. Since the size of the arriving or departing item is $\geq \nicefrac{\epsilon}{14}$ the migration factor is bounded by $\mathcal{O}(\nicefrac{1}{\epsilon})$. In the same way, a migration of at most $\mathcal{O}(\nicefrac{1}{\epsilon})$ is used to fulfill the Heap Equation which implies that the migration in step 5 is bounded by $\mathcal{O}(\nicefrac{1}{\epsilon})$. 

If $|Q_j|$ falls below $\nicefrac{1}{\epsilon}$, the two queues $Q_j$ and $Q_{j+1}$ are merged by emptying $bb_j$. The removed items are inserted by Algorithm \ref{algo-small-items-general}. As their total size is bounded by $1$ and the algorithm has a migration factor of $\mathcal{O}(\nicefrac{1}{\epsilon})$, the size of the moved items is bounded by $\mathcal{O}(\nicefrac{1}{\epsilon})$. The migration to merge two queues can thus be bounded by $\mathcal{O}(\nicefrac{1}{\epsilon^2})$. 

Note that the proof of Theorem \ref{thm-main} implies that at most $\gamma=\mathcal{O}(\nicefrac{1}{\epsilon^2}\log \nicefrac{1}{\epsilon})$ bins are changed by Algorithm \ref{alg-afptas}. The total size of the items $I_S$ which are removed in step 2 is thus bounded by $\gamma$. Similarly, the length of at most $\gamma$ queues can fall below $\nicefrac{1}{\epsilon}$. The migration of step 3 is thus bounded by $\gamma\cdot \nicefrac{1}{\epsilon^2}$. As at most $\gamma$ buffer bins are changed, the change of the potential (and thus the number of heap bins) is also bounded by $\gamma$ and the migration in step 4 can be bounded by $\gamma\cdot \nicefrac{1}{\epsilon}$. The migration in step 6 is bounded by $s(I_S) \cdot \nicefrac{1}{\epsilon}\leq \gamma\cdot \nicefrac{1}{\epsilon}$ as Algorithm \ref{algo-small-items-general} has migration factor $\nicefrac{1}{\epsilon}$. The total migration of the adjustments is thus bounded by $\gamma\cdot \nicefrac{1}{\epsilon^2}=\mathcal{O}(\nicefrac{1}{\epsilon^4}\log \nicefrac{1}{\epsilon})$.

{\bf Running Time:} The handling of small items can be performed in linear time while the handling of large items requires $\mathcal{O}(M(\nicefrac{1}{\epsilon} \log(\nicefrac{1}{\epsilon}))\cdot \nicefrac{1}{\epsilon^3} \log(\nicefrac{1}{\epsilon})+\nicefrac{1}{\epsilon} \log(\nicefrac{1}{\epsilon}) \log(\epsilon^2\cdot n(t))+\epsilon n(t))$, where $M(n)$ is the time needed to solve a system of $n$ linear equations (see Theorem \ref{thm-main}). The total running time of the algorithm is thus $\mathcal{O}(M(\nicefrac{1}{\epsilon} \log(\nicefrac{1}{\epsilon}))\cdot \nicefrac{1}{\epsilon^3} \log(\nicefrac{1}{\epsilon})+\nicefrac{1}{\epsilon} \log(\nicefrac{1}{\epsilon}) \log(\epsilon^2\cdot n(t))+ n(t))$. 
\end{proof}

\subsubsection*{Acknowledgements}
We would like to thank Till Tantau for his valuable comments and suggestions to improve the presentation of the paper.

\bibliography{library}


\end{document}